\documentclass[12pt]{article}
\usepackage{amsmath,amsthm,amssymb,mathtools}

\usepackage{array,hyperref,graphicx}
\hypersetup{colorlinks=true,allcolors=blue}
\usepackage{xurl}   
\usepackage{enumitem}
\usepackage{color}
\usepackage{tikz}
\usetikzlibrary{matrix,arrows.meta,positioning}
\tikzset{gn/.style={circle,draw,minimum size=6.5mm,inner sep=0pt,font=\small},
  rt/.style={circle,fill=black!60,minimum size=2.5mm,inner sep=0pt},
  bad/.style={draw=red!70!black,fill=red!12},
  pth/.style={line width=1.6pt,blue!60!black},
  mat/.style={matrix of math nodes,nodes={minimum width=15mm,minimum height=6.5mm,anchor=center,font=\small},
    column sep=1pt,row sep=1pt}}

\makeatletter
\@addtoreset{equation}{section}
\makeatother

\newtheorem{theorem}{Theorem}[section]
\newtheorem{proposition}[theorem]{Proposition}
\newtheorem{lemma}[theorem]{Lemma}

\newtheorem{definition}[theorem]{Definition}
\newtheorem{corollary}[theorem]{Corollary}
\theoremstyle{definition}
\newtheorem{remark}[theorem]{Remark}

\newcommand{\C}{\mathbb C}
\newcommand{\N}{\mathbb N}
\def\cA{{\cal A}} \def\cM{{\cal M}} \def\cN{{\cal N}} \def\cS{{\cal S}}  \def\cY{{\cal Y}}
\newcommand{\Ch}{\operatorname{Ch}}
\newcommand{\Monom}{\mathbf{Monom}}
\newcommand{\Exp}{\operatorname{Exp}}

\newcommand{\cZ}{\mathcal{Z}}

\newcommand{\Sstrut}{\rule[-1.3ex]{0pt}{3.8ex}}
\colorlet{YAtop}{blue!12}\colorlet{YAleft}{blue!35}\colorlet{YAright}{blue!55}
\colorlet{YBtop}{orange!15}\colorlet{YBleft}{orange!45}\colorlet{YBright}{orange!70}
\colorlet{YCtop}{black!8}\colorlet{YCleft}{black!25}\colorlet{YCright}{black!40}

\title{ \hfill { \normalsize QMUL-PH-26-35} \\[1em]
{\Large { \bf Degeneracy graph determinants\\
and higher-dimensional Young diagrams } } }

\author{
	Sanjaye Ramgoolam$^{1}$\thanks{\href{mailto:s.ramgoolam@qmul.ac.uk}{s.ramgoolam@qmul.ac.uk}},
	Micha{\l} Studzi{\'n}ski$^{2}$\thanks{\href{mailto:michal.studzinski@ug.edu.pl}{michal.studzinski@ug.edu.pl}}\\[1em]
	\normalsize $^{1}$Centre for Theoretical Physics, Department of Physics and Astronomy,\\
	\normalsize  Queen Mary University of London, UK\\
	\normalsize $^{2}$International Centre for Theory of Quantum Technologies, \\
	\normalsize  University of Gda{\'n}sk, Poland
}

\begin{document}

\date{}

\maketitle
\vspace{-2.0em}
\begin{center}
{ \bf ABSTRACT  }
\end{center}

\vspace{-0.5em}

\normalsize{
Kemp and Ramgoolam proposed a basis of monomials for a commutative semisimple
algebra with generators \(C_1,\ldots, C_L\), constructed from a degeneracy
graph which records how the eigenspaces of \( C_1\) are successively split by
\( C_2,\ldots, C_L\).  They showed that the basis property is equivalent to the
non-vanishing of the determinant of the matrix of values of these monomials
on the projector basis, and conjectured a factorised formula for this
determinant, as a product of powers of eigenvalue differences.  We prove this
formula for any number of generators.  The proof uses graph analogues of
Newton polynomials.  Their evaluation matrix is not triangular in general,
but expanding the rows of the children of each node of the penultimate layer
in a Newton basis of their eigenvalues makes it block triangular, with diagonal blocks given by smaller degeneracy
graphs in which the final layer is removed.  The formula then follows by
induction on the number of layers.
We show that the exponents in the determinant are overlap volumes of
higher-dimensional Young diagrams attached to the nodes of the graph:
intervals, ordinary Young diagrams, plane partitions and their higher
analogues.  Finally, the block triangular structure gives an explicit
recursive inversion of the evaluation matrix, expressing the projectors in the
monomial basis.
}

\newpage

\tableofcontents

\section{ Introduction }

The duality between \(\mathcal N=4\) super Yang--Mills theory with gauge
group \(U(N)\) and type IIB string theory on \(AdS_5\times S^5\) is the
paradigmatic example of the AdS/CFT correspondence
\cite{Maldacena,GKP,Witten}. The half-BPS sector includes quantum states describing ordinary gravitons as well as giant gravitons \cite{MST,GMT,HHI}.   The quantum states in this sector  correspond to local operators which are holomorphic gauge invariant polynomials in a single complex matrix \(Z\), and
at dimension \(n\) an orthogonal basis, in the free field inner product, is
given by Schur polynomial operators labelled by Young diagrams \(R\) with \(n\)
boxes and at most \(N\) rows \cite{CJR}.  For large \(n\), of order \(N^2\),
these operators provide duals to the smooth
bubbling geometries of Lin, Lunin and Maldacena \cite{LLM}.  A natural question, with a bearing on information loss
in black hole physics, is how much of the microscopic data \(R\) can be
recovered by a probe which only has access to a finite number of conserved
charges, such as the asymptotic multipole moments of the gravitational field
\cite{BCLS}, while related asymptotic information is in the one-point functions in the CFT \cite{SkenderisTaylor}.  The
relevant charges are Casimirs of \(U(N)\), and by Schur--Weyl duality they
correspond to central elements of the symmetric group algebra
\(\C(S_n)\).  In \cite{KR2019}, this led to the study of the centre
\(\cZ(\C(S_n))\) through the class sums \(T_k\) of permutations with a single
non-trivial cycle of length \(k\). We refer to these as cycle class sums.  The normalised characters of \(T_2\) alone distinguish all Young diagrams only for a few small \(n\), those of \(T_2,T_3\) do so up to
\(n=14\), and in general one asks for the smallest \(k_*(n)\) such that
\(T_2,\ldots,T_{k_*(n)}\) distinguish all Young diagrams with \(n\) boxes.
Probes of increasing classical dimension correspond to adjoining these
generators one at a time.  The sequences, and the associated Shannon
entropies, were studied in \cite{KR2019}.  Lower bounds on the growth of
\(k_*(n)\) were obtained in \cite{BGR-detection}, and upper bounds, based on a
conjectured generalised dominance ordering, were proposed in
\cite{Kemp-dominance}.

Distinguishing the Young diagrams is equivalent to the statement that the
projectors \(P_R\) which form a basis for the centre can be written as polynomials in
\(T_2,\ldots,T_{k_*(n)} \) \cite{KR2019}: the class sums form a \emph{non-linear
generating set} of the centre.  This connects the physics question to
quantum information and complexity.  Quantum algorithms which detect the
projectors \(P_R\) using a minimal generating set have complexity polynomial
in \(n\), even though the number of Young diagrams grows as
\(e^{\sqrt n}\) \cite{BGR-detection}; related questions for Kronecker
coefficients were studied in \cite{BGR-ribbon,Bravyi}.  Minimal generating sets
also enter the eigenvalue systems used to construct integer orthogonal bases
of multi-matrix invariants at finite \(N\) \cite{PRS-eigen}, bases which are useful for studying one-loop finite $N$ physics \cite{PRS-OneLoop}, and they enter the
integrality questions of combinatorial topological field theories and
topological strings built from finite groups \cite{dMKHKR,RamgoolamSharpe}.

A complementary question concerns the explicit \emph{construction} of the
projectors: given a non-linear generating set, can one write down a basis of
the algebra consisting of monomials in the generators, and express the
projectors in it?  In \cite{KempRamgoolam2026}, Kemp and Ramgoolam formulated
this question for an arbitrary finite-dimensional commutative associative
semisimple (CASS) algebra \(\cA\) with an ordered generating set
\(C_1,\ldots,C_L\).  Adjoining the generators one at a time gives a sequence of
subalgebras \(\cA_1\to\cA_2\to\cdots\to\cA_L=\cA\), and the way the projectors
of each subalgebra split in the next is recorded in a layered tree, the
\emph{degeneracy graph}, whose nodes carry the eigenvalues of the generators.
From the combinatorics of this graph, Kemp and Ramgoolam proposed an explicit
set of graph monomials in \(C_1,\ldots,C_L\), proved that their number is the
dimension of \(\cA\), and proved that they form a basis for one and two
layers.  For any number of layers, the basis property is equivalent to the
non-vanishing of the determinant of the matrix of values of the graph monomials
on the projectors, and they conjectured a factorised formula for this
determinant, as a product of powers of differences of eigenvalues of sibling
nodes, with exponents computed from truncated graphs.  The conjecture was
supported by extensive computer checks.  Examples include the centres of
symmetric group algebras generated by the \(T_k\), and the maximal commutative
subalgebras generated by Jucys--Murphy elements, whose degeneracy graphs are
Gelfand--Tsetlin patterns; applications to matrix units of semisimple
algebras and to multi-matrix invariants are described in
\cite{KempRamgoolam2026}.

Schur--Weyl duality gives these questions a direct interpretation in quantum
information. On \(n\) quantum systems of local dimension \(N\), the commuting
actions of \(U(N)\) and \(S_n\) decompose \((\C^N)^{\otimes n}\) into sectors
labelled by Young diagrams \(R\) with \(n\) boxes and at most \(N\) rows. The
images of the central projectors \(P_R\) select these sectors. Measuring them
on independently prepared copies \(\rho^{\otimes n}\) produces a Young diagram
whose normalised row lengths estimate the eigenvalues of \(\rho\)~\cite{KW}.
The same decomposition underlies the quantum Schur transform, for which
efficient quantum circuits are known~\cite{BCH,SchurHighDim2025}.

The cycle class sums considered here also have a direct interpretation in
terms of spectral moments. We also write \(T_k\) for the operator representing
the unnormalised class sum by permutations of the tensor factors. For a density
operator \(\rho\) on \(\C^N\),
\begin{equation}
\begin{aligned}
\operatorname{Tr}\!\left(\rho^{\otimes n}T_k\right)
=\frac{n!}{k(n-k)!}\operatorname{Tr}(\rho^k),
\qquad 2\leq k\leq n.
\end{aligned}
\label{eq:cycle-class-moments}
\end{equation}
The combinatorial factor counts the permutations with one non-trivial cycle
of length \(k\). Thus the generators considered here give access to spectral
moments of a quantum state. The underlying cyclic-permutation identity is
used in quantum procedures for estimating purity and higher moments without
reconstructing the full density matrix~\cite{E}.

The results obtained here establish the algebraic connection between these commuting
observables and the projectors representing their spectral decompositions. The
application of the determinant formula to a family of sums of cycle classes that
distinguish all Young diagrams with \(n\) cells yields the graph monomial basis for
the centre, and the recursion formula gives each projector \(P_R\) in that basis.
Expectation values of these projectors yield probabilities of the Schur--Weyl measurement
outcomes as expectation values of the graph monomials. The approach is also applicable to
the commutative subalgebra generated by the Jucys--Murphy elements whose spectral
projectors describe the branching paths of the symmetric group representations.

The same questions arise for other algebras of interest in physics.  Walled
Brauer algebras \(B_N(m,n)\) encode the combinatorics of mixed tensors of
\(U(N)\) with \(m\) fundamental and \(n\) antifundamental indices.  They
describe brane--anti-brane systems in gauge--gravity duality \cite{KimuraRamgoolam}.
In this setting, port-based teleportation provides a further application in
quantum information. In this protocol, the receiver selects one of
several output systems according to the sender's measurement outcome, without
applying a correcting unitary. The measurement operators belong to the algebra
of partially transposed permutations, which is the image of a walled Brauer
algebra on mixed tensor space~\cite{PBT1,GrinkoOzols2024,StudzinskiWalledBrauer2025}.
Resolving these operators into irreducible sectors gives the spectral data
needed to evaluate the teleportation fidelity~\cite{PBT1}.
The same algebraic methods enter the optimisation of the shared resource state
and the measurement~\cite{PBT2}. Efficient algorithms for deterministic and
probabilistic PBT with qubits have also been developed~\cite{WillsHsiehStrelchuk2024}.

More generally, quantum channels covariant under tensor powers of \(U(N)\) have
Choi operators commuting with the corresponding mixed tensor action
\(U^{\otimes m}\otimes\overline{U}^{\otimes n}\), where \(m\) and \(n\) count the
output and input systems, respectively. Mixed Schur--Weyl duality therefore
places these operators in the image of \(B_N(m,n)\). Efficient quantum circuits
for the mixed Schur transform have been constructed~\cite{GrinkoBurchardtOzols2023,Nguyen2023};
a related twisted Schur transform is used in an efficient algorithm for
port-based teleportation~\cite{FeiTimmermanHayden2023}. Projectors adapted to the
branching of these algebras have been used to reduce quantum optimisation
problems with additional symmetries, including problems in asymmetric
cloning~\cite{GrinkoOzols2024}. These applications provide an additional setting
for constructing projectors from commuting generators in commutative
semisimple subalgebras.

Recently, the regime \(N<m+n\), below the stable range and including
non-semisimple specialisations, was studied
using restricted Bratteli diagrams \cite{RS-walledBrauer}, and it was found that the counting  of nodes in these graphs is  governed by an infinite tower of simple harmonic oscillators.  Constructing the projectors of the semisimple images of such algebras on mixed tensor space, and of their commutative semisimple
subalgebras, from a small set of generators is a problem of exactly the kind
considered here, and is one of the motivations for the present work.

In this paper we prove the determinant conjecture of \cite{KempRamgoolam2026}
for degeneracy graphs with any number of layers
(Theorem~\ref{thm:KR}), and hence that the Kemp--Ramgoolam graph monomials
form a basis of every CASS algebra with a given ordered generating set.  Our
results are as follows.
\begin{itemize}
\item In Section~\ref{sec:proof} we give a direct bijection between the graph
  monomials and the projectors of the algebra (Proposition~\ref{prop:bijection}),
  and introduce graph analogues of Newton polynomials.  In one
  variable, the Vandermonde determinant is most easily computed by changing
  from the monomials \(1,x,x^2,\ldots\) to the Newton polynomials
  \(1,\,x-x_1,\,(x-x_1)(x-x_2),\ldots\), whose evaluation matrix is triangular.
  The graph Newton polynomials are related to the graph monomials by a
  unitriangular change of basis, but their evaluation matrix is in general
  \emph{not} triangular.  Nevertheless its determinant is the product of
  its diagonal entries (Theorem~\ref{thm:newtondet}).  To prove this, we
  expand the rows of sibling nodes in the Newton basis of their eigenvalues
  (Lemma~\ref{lem:newton}).  This factorises the evaluation matrix as a block
  diagonal matrix of one-layer Newton--Vandermonde matrices, one for each set
  of siblings, times a matrix which, by a degree argument
  (Lemma~\ref{lem:degree}), is block triangular, with diagonal blocks given by
  smaller degeneracy graphs, obtained by removing the final layer.  The one-variable
  Newton construction thus enters twice, in the graph Newton polynomials and in
  the expansion among siblings.  The determinant formula follows by induction
  on the number of layers.
\item In Section~\ref{sec:young} we give a geometric description of the exponents.  To each
  node we attach a higher-dimensional Young diagram, which records the
  continuations of the graph through that node: an interval, an ordinary
  Young diagram, a plane partition, and higher-dimensional analogues, as the
  depth increases.  The exponent of a sibling difference is the overlap
  volume of the Young diagrams of the two siblings, i.e.\ the number of points
  in their intersection.  Plane partitions and
  their four-dimensional analogues are familiar in physics from the melting
  crystal picture of topological strings \cite{ORV} and from the
  solid partitions of the \emph{magnificent four} model \cite{Nekrasov4}.
\item In Section~\ref{sec:projectors} we return to the motivating question of
  \cite{KempRamgoolam2026}, the expansion of the projectors of the algebra in
  the monomial basis.  The determinant formula guarantees that this expansion
  exists.  The block triangular structure found in the proof, together with
  the divided differences, which invert the Newton expansion among siblings,
  also gives an explicit recursive procedure for computing it.
\end{itemize}

\section{Newton--Vandermonde proof of the graph determinant formula}
\label{sec:proof}

In this section we explain the mechanism behind the Kemp--Ramgoolam graph
determinant formula.  The argument is the graph analogue of the usual
Newton-basis proof of the Vandermonde determinant.

Recall first the ordinary one-variable situation.  Let \(x_1,\ldots,x_n\) be
distinct numbers, listed in a fixed order.  The Vandermonde determinant
\(\prod_{1\le k<m\le n}(x_m-x_k)\) is often written as the determinant of the
matrix of values at \(x_1,\ldots,x_n\) of the monomial basis
\begin{equation}
1,x,x^2,\ldots,x^{n-1}.
\end{equation}
It can equally be computed using the \emph{Newton basis} of
\(x_1,\ldots,x_n\), whose \(d\)-th element is
\(\pi_d(x)=\prod_{m<d}(x-x_m)\), for \(1\le d\le n\):
\begin{equation}
1,\qquad
x-x_1,\qquad
(x-x_1)(x-x_2),\qquad
\ldots,\qquad
\prod_{m=1}^{n-1}(x-x_m).
\label{eq:newtonbasis1var}
\end{equation}
The polynomial \(\pi_d\) is monic of degree \(d-1\), depends only on
\(x_1,\ldots,x_{d-1}\), and vanishes at these points.  Hence the change from
monomials to Newton polynomials is triangular with diagonal entries equal to
\(1\), while the matrix of values of the Newton basis, with entry
\(\pi_d(x_k)\) in row \(k\) and column \(d\), is lower triangular, with
diagonal entries \(\pi_k(x_k)=\prod_{m<k}(x_k-x_m)\), the Vandermonde factors.
We call it the \emph{Newton--Vandermonde matrix} of \(x_1,\ldots,x_n\).

We follow the same strategy for the Kemp--Ramgoolam set of graph
monomials.  Sections~\ref{sec:ordconv} and~\ref{sec:admchld} set
up the data of the graph from which the graph monomials are built, and
Section~\ref{sec:admchld} establishes a bijection between the graph monomials
and the final nodes of the graph (Proposition~\ref{prop:bijection}).  In
Section~\ref{sec:graphNewton} we use the same data to define graph analogues
of the Newton polynomials, and relate them to the graph monomials.  In
Section~\ref{sec:evalNewton} we evaluate them at the final nodes of the graph
and show that the determinant of the resulting evaluation matrix is the
product of its diagonal entries (Theorem~\ref{thm:newtondet}), by induction on
the number of layers.  The induction step expands the rows of siblings in the
Newton basis \eqref{eq:newtonbasis1var} of their eigenvalues (Lemmas~\ref{lem:newton}
and~\ref{lem:degree}), so the one-variable Newton construction is used a second
time.  In Section~\ref{sec:exponents} the result is put in the form of the
Kemp--Ramgoolam determinant formula.

\subsection{The ordering convention}
\label{sec:ordconv} 

Let
\begin{equation}
V_0,V_1,\ldots,V_L
\end{equation}
be the layers of the Kemp--Ramgoolam degeneracy graph
\(\Gamma\).  Here \(V_0\) consists of the root, and the
elements of \(V_L\) are called final nodes.  The final nodes label the
projector basis of the commutative algebra \(\cA_L\) generated by
\(C_1,\ldots,C_L\) (more generally, \(\cA_i\) denotes the subalgebra generated
by \(C_1,\ldots,C_i\)), and each final node carries the eigenvalues of
\(C_1,\ldots,C_L\) on its projector (Section~\ref{sec:projectors}). 

\paragraph{Parents, children and siblings.}  Edges of the graph join nodes in
consecutive layers only.  Every node \(u\in V_i\) with \(i\geq1\) is joined to
exactly one node of \(V_{i-1}\), called its parent, and \(u\) is called a
child of its parent (\cite{KempRamgoolam2026} uses ``daughter'').  For
\(b\in V_{i-1}\) we write
\begin{equation}
\Ch_i(b)\subset V_i
\end{equation}
for the set of children of \(b\).  The root has all of \(V_1\) as its
children, \(\Ch_1(v_0)=V_1\).  Two nodes with the same parent are siblings; in
particular all nodes of \(V_1\) are siblings.  Following edges to the right
from a node reaches its descendants, and following them to the left reaches
its ancestors, ending at the root.  Every node of \(V_0,\ldots,V_{L-1}\) has
at least one child, so every node has at least one descendant in \(V_L\).  For
example, in the graph of Figure~\ref{fig:D24}, \(T_1\) and \(T_2\) are the
children of node \(1\in V_1\) and are siblings, while \(T_2\) and \(T_3\) are
not siblings.

In the systematic construction used for the degeneracy graphs in \cite{KempRamgoolam2026}, the vertices at
layer \(i\) are labelled by integers
\begin{equation}
V_i=[D_i]=\{1,2,\ldots,D_i\}.
\end{equation}
The constructions below (rank data, graph Newton polynomials, divided
differences and the determinant formula) compare only siblings, and use the
integer order restricted to each set of children. Thus, if \(u,v\) are children of the same parent, then \(u<v\)
means that the integer label of \(u\) is smaller than that of \(v\).

Following \cite{KempRamgoolam2026}, we record the shape of a graph by
\(D=(D_1,\ldots,D_L)\), together with the lists
\(p_1=[p_{1,1},\ldots,p_{1,D_1}]\) of the numbers of children of the nodes of
\(V_1\), and \(c_i=[c_{i,1},\ldots,c_{i,D_i}]\) of the numbers of children of
the nodes of \(V_i\), \(2\le i\le L-1\).  For example, \(D=(2,4,7)\),
\(p_1=[2,2]\), \(c_2=[2,1,2,2]\) describes a graph whose two nodes in \(V_1\)
have two children each, and whose four nodes in \(V_2\) have \(2,1,2,2\)
children.

As in \cite{KempRamgoolam2026}, the parts of \(p_1\) are in weakly increasing
order, \(p_{1,1}\le p_{1,2}\le\cdots\le p_{1,D_1}\) (eq.~(2.8) there), and the
children of the nodes of layer \(i-1\) form successive blocks of
\([D_i]\) (eqs.~(2.14)--(2.15) there).  Thus the children of each node carry
consecutive labels, assigned in the order of their parents: if \(u<u'\) in
\(V_{i-1}\), every child of \(u\) has a smaller label than every child of
\(u'\).  Subsets of the children of a node, such as the admissible child sets
below, are ordered by the same integer order.

The weakly increasing order of \(p_1\) only fixes the labelling of the
original graph, and is not used in the arguments below.  All constructions and
results of this paper hold for any degeneracy graph in which an arbitrary
order is fixed among the children of each node, the labels at each layer being
assigned in successive blocks as above.  This larger class is the one needed
in the induction of Section~\ref{sec:general}: it is preserved by the
final-layer-excised graphs \(\Gamma_{[d]}\) of
Definition~\ref{def:final-layer-excised}, whereas the weakly increasing order
is not.  For example, for \(D=(2,4,7)\), \(p_1=[2,2]\), \(c_2=[2,2,2,1]\), the
two nodes of \(V_1\) keep \(2\) and \(1\) children in \(\Gamma_{[2]}\).

In figures we follow the convention of \cite{KempRamgoolam2026}: the layers
are drawn as columns from left to right, starting with the root, and within
each column the labels increase upwards.

As in \cite{KempRamgoolam2026}, we work with a concrete labelling as in
the computational construction.  The constructions below hold for any such
labelling; how the determinant depends on the labelling is discussed after
Theorem~\ref{thm:KR}.

\subsection{Rank data: a bijection between graph monomials and final nodes}
\label{sec:admchld}

In this subsection we recall the Kemp--Ramgoolam graph monomials and establish
a bijection between the final nodes \(V_L\) and the graph monomials
(Proposition~\ref{prop:bijection}).  The bijection is described by assigning
to each final node what we call its {\it rank data}.

Every final node \(T\in V_L\) determines a unique path
\begin{equation}
v_0(T)\longrightarrow v_1(T)\longrightarrow\cdots\longrightarrow v_L(T)=T,
\qquad v_i(T)\in V_i.
\label{eq:path}
\end{equation}
It is obtained by reading the graph from right to left: \(v_{i-1}(T)\) is the
parent of \(v_i(T)\).

\paragraph{Path order.} 

 For \(T,T'\in V_L\), we say that \(T'\) precedes
\(T\) in path order if, at the first layer \(i\) where their paths differ,
\(v_i(T')<v_i(T)\).  The nodes \(v_i(T')\) and \(v_i(T)\) are then siblings.
Path order is the integer order on \(V_L=\{1,\ldots,D_L\}\).  Indeed, by the
successive-block labelling of Section~\ref{sec:ordconv}, \(v_i(T')<v_i(T)\) implies that
every child of \(v_i(T')\) precedes every child of \(v_i(T)\), so
\(v_{i+1}(T')<v_{i+1}(T)\), and so on up to \(T'<T\).  For example, in the
graph of Figure~\ref{fig:D24}, \(T_1,T_2\) descend from node \(1\in V_1\) and
\(T_3,T_4\) from node \(2\), so \(T_1,T_2\) precede \(T_3,T_4\) in path order.

\paragraph{The sets \(S^{(i)}\).}  The graph monomials are built from tuples
of positive integers \((d_2,\ldots,d_L)\).  For such a tuple and
\(1\le i\le L-1\), the set \(S^{(i)}_{[d_{i+1},\ldots,d_L]}\subset V_i\) is
defined from right to left, as in eqs.~(3.1)--(3.3) of
\cite{KempRamgoolam2026}:
\begin{align}
S^{(L-1)}_{[d_L]}
&=\bigl\{u\in V_{L-1}:\ |\Ch_L(u)|\ge d_L\bigr\},
\label{eq:SLm1}\\
S^{(i)}_{[d_{i+1},\ldots,d_L]}
&=\bigl\{u\in V_i:\ \bigl|\Ch_{i+1}(u)\cap S^{(i+1)}_{[d_{i+2},\ldots,d_L]}\bigr|
\ge d_{i+1}\bigr\},\qquad 1\le i\le L-2.
\label{eq:Si}
\end{align}
In words, \(u\in V_i\) lies in \(S^{(i)}_{[d_{i+1},\ldots,d_L]}\) if it has at
least \(d_{i+1}\) children that have at least \(d_{i+2}\) children that have at
least \(d_{i+3}\) children, and so on, up to nodes of layer \(L-1\) with at
least \(d_L\) children.  The conditions to the right restrict which children
of \(u\) are counted; \(u\) may have other children that do not satisfy them.
For
example, in the graph of Figure~\ref{fig:D24} both nodes of \(V_1\) have two
children, so \(S^{(1)}_{[1]}=S^{(1)}_{[2]}=\{1,2\}\); in the graph of
Figure~\ref{fig:L2}, \(S^{(1)}_{[1]}=\{1,2\}\) and \(S^{(1)}_{[2]}=\{2\}\).
Decreasing any \(d_j\) can only enlarge these sets.

\paragraph{Admissible child sets.}  The sets \(S^{(i)}\) select, among the
children of a node, those which can continue a path with the prescribed
integers to their right.

At the first layer define
\begin{equation}
R_1(d_2,\ldots,d_L):=
S^{(1)}_{[d_2,\ldots,d_L]}.
\label{eq:R1}
\end{equation}
For \(2\leq i\leq L-1\), if \(v_{i-1}\in V_{i-1}\), define
\begin{equation}
R_i(v_{i-1};d_{i+1},\ldots,d_L)
:=
\Ch_i(v_{i-1})\cap
S^{(i)}_{[d_{i+1},\ldots,d_L]}.
\label{eq:Ri}
\end{equation}
At the final layer there is no further descendant condition, so we set
\begin{equation}
R_L(v_{L-1}):=\Ch_L(v_{L-1}).
\label{eq:RL}
\end{equation}
For uniformity of notation we also write
\(R_1(v_0;d_2,\ldots,d_L)=R_1(d_2,\ldots,d_L)\), where \(v_0\) is the root.

\paragraph{Graph monomials.}  Following eq.~(3.5) of
\cite{KempRamgoolam2026}, for a tuple \((d_2,\ldots,d_L)\) of positive integers
define the set of monomials in the commuting variables \(C_1,\ldots,C_L\)
\begin{equation}
\Monom(d_2,\ldots,d_L)
=
\bigl\{1,C_1,C_1^2,\ldots,C_1^{|S^{(1)}_{[d_2,\ldots,d_L]}|-1}\bigr\}
\times (C_2)^{d_2-1}(C_3)^{d_3-1}\cdots(C_L)^{d_L-1},
\label{eq:Monom}
\end{equation}
which is empty when \(S^{(1)}_{[d_2,\ldots,d_L]}\) is empty; we call
\((d_2,\ldots,d_L)\) admissible when it is not.  The Kemp--Ramgoolam graph
monomials (eq.~(3.6) of \cite{KempRamgoolam2026}) form the disjoint union
\begin{equation}
\begin{aligned}
\cS(\cA_1\to\cdots\to\cA_L)
&=
\bigsqcup_{d_2,\ldots,d_L}\Monom(d_2,\ldots,d_L)\\
&=
\Bigl\{(C_1)^a(C_2)^{d_2-1}\cdots(C_L)^{d_L-1}:\ (d_2,\ldots,d_L)
\text{ admissible},\\
&\hspace{6em} 0\le a\le \bigl|S^{(1)}_{[d_2,\ldots,d_L]}\bigr|-1\Bigr\}.
\end{aligned}
\label{eq:Mset}
\end{equation}
We abbreviate this set to \(\cS\).

\paragraph{From a monomial to a final node.}  Given such
\((a,d_2,\ldots,d_L)\), a final path is constructed from left to right as
follows:
\begin{align}
v_1&=\text{the }(a+1)\text{-st element of }R_1(d_2,\ldots,d_L),\label{eq:v1}\\
v_i&=\text{the }d_i\text{-th element of }
R_i(v_{i-1};d_{i+1},\ldots,d_L),
\qquad 2\leq i\leq L-1,\label{eq:vi}\\
v_L&=\text{the }d_L\text{-th element of }R_L(v_{L-1}).\label{eq:vL}
\end{align}
Each set \(R_i\) here depends only on the node \(v_{i-1}\) already chosen and
on the integers to the right of position \(i\).

\paragraph{From a final node to a monomial: rank data.}  Conversely, from a
final node \(T\) we read off a tuple by reading its path from right to left:
\begin{itemize}
\item \(d_L\) is the position of \(T\) among the children of \(v_{L-1}(T)\);
\item for \(i=L-1,\ldots,2\), \(d_i\) is the position of \(v_i(T)\) in
  \(R_i(v_{i-1}(T);d_{i+1},\ldots,d_L)\);
\item \(a+1\) is the position of \(v_1(T)\) in \(R_1(d_2,\ldots,d_L)\).
\end{itemize}
Each step uses the integers already found to its right.  We call
\((a,d_2,\ldots,d_L)\) the rank data of \(T\), and
\(m_T=(C_1)^a(C_2)^{d_2-1}\cdots(C_L)^{d_L-1}\) its graph monomial.  Note that
the admissible sets used for \(T\) are computed from its own rank data; this
will be important in Section~\ref{sec:evalNewton}.

\begin{proposition}\label{prop:bijection}
The construction \eqref{eq:v1}--\eqref{eq:vL}, from a graph monomial to a
final node, and the rank data, from a final node to a graph monomial, are
well defined and inverse to each other.  They give a bijection
\begin{equation}
V_L\longleftrightarrow \cS,
\qquad
T\;\longmapsto\;(a,d_2,\ldots,d_L)\;\longmapsto\;
m_T=(C_1)^a(C_2)^{d_2-1}\cdots(C_L)^{d_L-1},
\label{eq:bijection}
\end{equation}
and in particular \(|\cS|=D_L=\dim\cA_L\).
\end{proposition}

\begin{proof}
\emph{The construction \eqref{eq:v1}--\eqref{eq:vL} is well defined.}  The
bound on \(a\) in \eqref{eq:Mset} guarantees that \(R_1\) has an
\((a+1)\)-st element.  For \(i\ge2\), \(v_{i-1}\) was chosen in
\(R_{i-1}\subseteq S^{(i-1)}_{[d_i,\ldots,d_L]}\), and by
\eqref{eq:SLm1}--\eqref{eq:Si} this membership means that
\(R_i(v_{i-1};d_{i+1},\ldots,d_L)\) has at least \(d_i\) elements.

\emph{The rank data are well defined and give a graph monomial.}  At each
step, \(v_i(T)\) is the \(d_i\)-th element of \(R_i\), so \(R_i\) has at least
\(d_i\) elements, which says that \(v_{i-1}(T)\) lies in the admissible set
used at the next step to the left.  At the last step \(v_1(T)\in R_1\), so
\(R_1=S^{(1)}_{[d_2,\ldots,d_L]}\) is non-empty, the tuple
\((d_2,\ldots,d_L)\) is admissible, and \(a\le|R_1|-1\); thus \(m_T\in\cS\).

\emph{The two maps are inverse.}  At each layer both directions use the same
admissible set, since \(R_i\) depends only on \(v_{i-1}\) and on
\(d_{i+1},\ldots,d_L\).  Starting from a tuple, the construction chooses
\(v_i\) as the \(d_i\)-th element of \(R_i\) (the \((a+1)\)-st for \(i=1\)), and
reading off the position of \(v_i\) in the same \(R_i\) returns \(d_i\) (resp.\
\(a+1\)); reading from right to left, each set is recovered from the integers
already returned.  Conversely, starting from \(T\), the construction applied to
its rank data chooses, layer by layer from the left, the nodes
\(v_1(T),\ldots,v_L(T)\).  Hence the maps are mutually inverse, and the second
arrow in \eqref{eq:bijection} simply records the rank data as the exponent
vector \((a,d_2-1,\ldots,d_L-1)\).  Finally \(|V_L|=D_L\), and \(D_L=\dim\cA_L\)
since the final nodes label the projector basis of \(\cA_L\).
\end{proof}

This gives a direct proof that the number of graph monomials is \(D_L\), which
was proved in \cite{KempRamgoolam2026} by a counting argument with truncated
graphs.  For example, in the graph of
Figure~\ref{fig:D24}, \(T_3\) is the first child of node \(2\in V_1\), so
\(d_2=1\), and node \(2\) is the second element of \(R_1(1)=\{1,2\}\), so
\(a=1\); hence \(m_{T_3}=C_1\).  The table in Section~\ref{sec:evalNewton} lists all four final
nodes.

The set \eqref{eq:Mset} depends only on the sizes of the sets \(S^{(1)}\).
Which final node is paired with which monomial depends on the order inside the
admissible sets, i.e.\ on the labelling of Section~\ref{sec:ordconv} (see
the discussion after Theorem~\ref{thm:KR}).

\subsection{Graph Newton polynomials}
\label{sec:graphNewton}

Let \(T\in V_L\) have path
\begin{equation}
v_0\longrightarrow v_1\longrightarrow\cdots\longrightarrow v_L=T
\end{equation}
and rank data \((a,d_2,\ldots,d_L)\).  The associated graph monomial is
\begin{equation}
m_T=(C_1)^a(C_2)^{d_2-1}\cdots(C_L)^{d_L-1}.
\end{equation}
For each \(v_i\in V_i\), let \(x^{(i)}_{v_i}\) be the eigenvalue of \(C_i\) on
the corresponding projector.

Define the graph Newton polynomial, a polynomial in the commuting variables \(C_1,\ldots,C_L\), for the final node \(T\in V_L\)
\begin{align}
N_T
&:=
\prod_{\substack{w<v_1\\ w\in R_1(d_2,\ldots,d_L)}}
\bigl(C_1-x_w^{(1)}\bigr)
\nonumber\\
&\quad\times
\prod_{i=2}^{L-1}
\prod_{\substack{w<v_i\\
w\in R_i(v_{i-1};d_{i+1},\ldots,d_L)}}
\bigl(C_i-x_w^{(i)}\bigr)
\nonumber\\
&\quad\times
\prod_{\substack{w<v_L\\ w\in R_L(v_{L-1})}}
\bigl(C_L-x_w^{(L)}\bigr).
\label{eq:graphNewton}
\end{align}
This display is written for \(L\ge2\).  For \(L=1\) its first and last
products coincide, and \(N_T\) is the single product over \(R_1=V_1\) (see the
one-layer case in Section~\ref{sec:evalNewton}).

Since \(v_1\) is the \((a+1)\)-st element of \(R_1\), the degree of \(N_T\)
in \(C_1\) is \(a\).  For \(2\leq i\leq L\), the node \(v_i\) is the
\(d_i\)-th admissible child, so the degree in \(C_i\) is \(d_i-1\).
Consequently, the leading term (LT) of the graph Newton polynomial is 
\begin{equation}
\operatorname{LT}(N_T)
=
(C_1)^a(C_2)^{d_2-1}\cdots(C_L)^{d_L-1}
=
m_T,
\label{eq:leadingterm}
\end{equation}
with leading coefficient \(1\).

Every other term in the expansion of \(N_T\) has exponents bounded above by
those of \(m_T\).  The monotonicity
\begin{equation}
d'_j\leq d_j
\quad\Longrightarrow\quad
S^{(i)}_{[d'_{i+1},\ldots,d'_L]}
\supseteq
S^{(i)}_{[d_{i+1},\ldots,d_L]}
\label{eq:Smonotonicity}
\end{equation}
shows that every such lower monomial is again a graph monomial. 
Therefore
\begin{equation}
N_T=\sum_{S\in V_L}U_{S,T}\,m_S,
\label{eq:Umatrix}
\end{equation}
where \(U_{T,T}=1\), and \(U_{S,T}\neq0\) with \(S\neq T\) only if \(m_S\) has
smaller total degree than \(m_T\).  With the final nodes listed by the total
degree of \(m_T\), the matrix \(U\) is upper triangular with diagonal entries
\(1\), so \(\det U=1\); since simultaneous reordering of rows and columns does
not change a determinant, this holds in any ordering of \(V_L\).  Write
\begin{equation}
\cM_{T',T}=m_T(T'),\qquad T',T\in V_L,
\label{eq:Mmatrix}
\end{equation}
for the graph monomial evaluation matrix, whose entry in row \(T'\) and column
\(T\) is the value of \(m_T\) at the final node \(T'\), i.e.\ with each \(C_i\)
replaced by the eigenvalue \(x^{(i)}_{v_i(T')}\).  Similarly, the Newton
evaluation matrix is
\begin{equation}
\cN_{T',T}=N_T(T'),\qquad T',T\in V_L,
\label{eq:Nmatrix}
\end{equation}
whose entry in row \(T'\) and column \(T\) is the value of the graph Newton
polynomial \(N_T\) at the final node \(T'\).  We write \(N_T(T')\) when the
point at issue is whether the polynomial \(N_T\) vanishes at \(T'\), and
\(\cN_{T',T}\) or \(\cN\) when the point at issue is the matrix and its
determinant.  Evaluating \eqref{eq:Umatrix} at \(T'\) gives
\(\cN_{T',T}=\sum_{S}\cM_{T',S}U_{S,T}\), i.e.\ \(\cN=\cM U\), and since
\(\det U=1\),
\begin{equation}
\det\cM
= \det\cN .
\label{eq:changebasisdet}
\end{equation}

\subsection{Evaluation of the graph Newton polynomials}
\label{sec:evalNewton}

By \eqref{eq:changebasisdet}, it suffices to compute the determinant of the
Newton evaluation matrix \(\cN\) of \eqref{eq:Nmatrix}, with rows labelled by
the points \(T'\) and columns by the polynomials \(N_T\).  We begin with the one-layer case, where this is
the classical Newton--Vandermonde construction recalled at the start of this
section.

\paragraph{The one-layer case.}
Take \(L=1\) and \(D_1=4\): the root has four children \(T_1,T_2,T_3,T_4\), with
eigenvalues \(x_k=x^{(1)}_{T_k}\) of \(C_1\) (Figure~\ref{fig:L1}).  There is no
descendant condition, so \(R_1=V_1\), and the rank data of \(T_k\) is just
\(a=k-1\).  The graph monomials are \(m_{T_k}=C_1^{\,k-1}\), i.e.
\(1,C_1,C_1^2,C_1^3\), and the graph Newton polynomials are the elements of
the Newton basis \eqref{eq:newtonbasis1var} of \(x_1,\ldots,x_4\), evaluated
at \(C_1\),
\begin{equation}
N_{T_1}=1,\quad N_{T_2}=C_1-x_1,\quad N_{T_3}=(C_1-x_1)(C_1-x_2),\quad
N_{T_4}=(C_1-x_1)(C_1-x_2)(C_1-x_3).
\end{equation}
The Newton evaluation matrix is the Newton--Vandermonde matrix of
\(x_1,\ldots,x_4\),
\begin{equation}
\cN=
\begin{pmatrix}
1&0&0&0\\[4pt]
1&x_2-x_1&0&0\\[4pt]
1&x_3-x_1&(x_3-x_1)(x_3-x_2)&0\\[4pt]
1&x_4-x_1&(x_4-x_1)(x_4-x_2)&(x_4-x_1)(x_4-x_2)(x_4-x_3)
\end{pmatrix},
\end{equation}
and its determinant is the Vandermonde determinant
\(\prod_{1\le k<m\le4}(x_m-x_k)\).  The graph monomial evaluation matrix
\(\cM=(x_{k}^{\,j-1})\) is the Vandermonde matrix itself, and
\eqref{eq:changebasisdet} is the familiar statement that the two have the same
determinant.  The reason for triangularity is simple: \(N_T\) has a factor
\(C_1-x_w\) for every earlier sibling \(w\) of \(T\), so it vanishes at every
earlier final node.

\begin{figure}[ht]
\centering
\begin{tikzpicture}
\node[rt] (r) at (0,1.5) {};
\foreach \k/\y in {1/0,2/1,3/2,4/3} {
  \node[gn] (t\k) at (1.8,\y) {$T_\k$};
  \draw (r)--(t\k);
  \node[right=1mm] at (t\k.east) {\small $x_\k$};
}
\node[font=\small] at (1.8,3.8) {$V_1$};
\end{tikzpicture}
\caption{The one-layer graph with \(D_1=4\).  The graph Newton polynomials are
the ordinary Newton polynomials in \(C_1\).}
\label{fig:L1}
\end{figure}
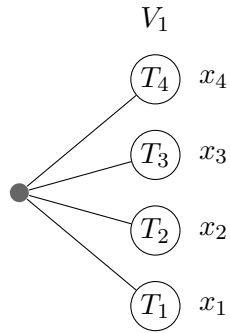

With more layers, the factors of \(N_T\) in layer \(i\) are taken only from the
admissible set \(R_i(v_{i-1};d_{i+1},\ldots,d_L)\), which depends on the rank
data of \(T\).  The next example shows that triangularity can survive this.

\subsubsection{A two-layer  example with triangular Newton evaluation matrix }

Take \(D=(2,4)\) and \(p_1=[2,2]\), the first example of
\cite{KempRamgoolam2026} (Figure~\ref{fig:D24}).  Node \(1\in V_1\) has
children \(T_1,T_2\) and node \(2\in V_1\) has children \(T_3,T_4\).  Both
nodes of \(V_1\) have two children, so
\begin{equation}\label{eq:R1d2} 
R_1(d_2)=S^{(1)}_{[d_2]}=\{1,2\} \quad \text{ for }  d_2\in \{ 1,2 \} \, . 
\end{equation} 
The rank data
\((a,d_2)\), graph monomials and graph Newton polynomials are
\begin{equation}
{\renewcommand{\arraystretch}{1.3}\begin{array}{c|c|c|l}
T & (a,d_2) & m_T & N_T\\\hline
T_1 & (0,1) & 1 & 1\\
T_2 & (0,2) & C_2 & C_2-x^{(2)}_{T_1}\\
T_3 & (1,1) & C_1 & C_1-x^{(1)}_1\\
T_4 & (1,2) & C_1C_2 & (C_1-x^{(1)}_1)(C_2-x^{(2)}_{T_3})
\end{array}}
\end{equation}
and the Newton evaluation matrix, with rows \(T'\) and columns \(T\) in path
order, is
\begin{equation}
\cN=
\begin{pmatrix}
1&0&0&0\\[4pt]
1&x^{(2)}_{T_2}-x^{(2)}_{T_1}&0&0\\[4pt]
1&x^{(2)}_{T_3}-x^{(2)}_{T_1}&x^{(1)}_2-x^{(1)}_1&0\\[4pt]
1&x^{(2)}_{T_4}-x^{(2)}_{T_1}&x^{(1)}_2-x^{(1)}_1&
(x^{(1)}_2-x^{(1)}_1)(x^{(2)}_{T_4}-x^{(2)}_{T_3})
\end{pmatrix}.
\end{equation}
It is lower triangular, and its determinant
\((x^{(1)}_2-x^{(1)}_1)^2(x^{(2)}_{T_2}-x^{(2)}_{T_1})(x^{(2)}_{T_4}-x^{(2)}_{T_3})\)
is the product of the diagonal entries.

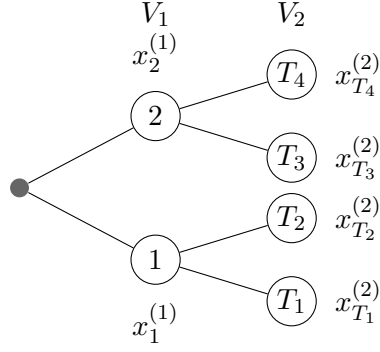
\begin{figure}[ht]
\centering
\begin{tikzpicture}
\node[rt] (r) at (0,1.5) {};
\node[gn] (a1) at (1.8,0.55) {$1$};
\node[gn] (a2) at (1.8,2.45) {$2$};
\node[gn] (t1) at (3.6,0) {$T_1$};
\node[gn] (t2) at (3.6,1.1) {$T_2$};
\node[gn] (t3) at (3.6,1.9) {$T_3$};
\node[gn] (t4) at (3.6,3.0) {$T_4$};
\draw (r)--(a1) (r)--(a2) (a1)--(t1) (a1)--(t2) (a2)--(t3) (a2)--(t4);
\node[below=1mm] at (a1.south) {\small $x^{(1)}_1$};
\node[above=1mm] at (a2.north) {\small $x^{(1)}_2$};
\foreach \t/\k in {t1/1,t2/2,t3/3,t4/4} \node[right=1mm] at (\t.east) {\small $x^{(2)}_{T_\k}$};
\node[font=\small] at (1.8,3.8) {$V_1$};
\node[font=\small] at (3.6,3.8) {$V_2$};
\end{tikzpicture}
\caption{The graph with \(D=(2,4)\), \(p_1=[2,2]\).  Layers are drawn as
columns from left to right, with labels increasing upwards.}
\label{fig:D24}
\end{figure}

The mechanism is the following.  Suppose \(T'\) precedes \(T\) in path order,
and their paths first differ at layer \(i\), so that \(v_i(T')\) is an earlier
sibling of \(v_i(T)\).  If \(v_i(T')\) belongs to the admissible set
\(R_i(v_{i-1};d_{i+1},\ldots,d_L)\) computed from the rank data of \(T\), then
\(N_T\) contains the factor \(C_i-x^{(i)}_{v_i(T')}\), and \(N_T(T')=0\).
In the example every admissible set contains all the children (at layer \(1\)
by \eqref{eq:R1d2}, and at layer \(2\) by \eqref{eq:RL}), so this
applies to every pair. Entries below the diagonal, such as
\(N_{T_2}(T_3)=x^{(2)}_{T_3}-x^{(2)}_{T_1}\), may be differences of eigenvalues
of nodes which are not siblings; they do not affect the determinant.

\subsubsection{A two-layer example with non-triangular but block-triangularisable Newton evaluation matrix }

The admissible set depends on the rank data of \(T\), so an earlier sibling
\(v_i(T')\) need not belong to it.  In that case \(N_T\) has no factor
\(C_i-x^{(i)}_{v_i(T')}\), and \(N_T(T')\) is in general nonzero.  The smallest
example has \(D=(2,3)\) and \(p_1=[1,2]\): node \(1\in V_1\) has one child
\(T_1\), and node \(2\in V_1\) has children \(T_2<T_3\).  The rank data give
\(N_{T_1}=1\), \(N_{T_2}=C_1-x^{(1)}_1\) and, since
\(R_1(d_2{=}2)=S^{(1)}_{[2]}=\{2\}\), \(N_{T_3}=C_2-x^{(2)}_{T_2}\)
(Figure~\ref{fig:L2}).  The Newton evaluation matrix, with rows
\(T'=T_1,T_2,T_3\) from top to bottom and columns \(N_{T_1},N_{T_2},N_{T_3}\)
from left to right, is
\begin{equation}
\cN=
\begin{pmatrix}
1\Sstrut & 0 & x^{(2)}_{T_1}-x^{(2)}_{T_2}\\
1\Sstrut & x^{(1)}_2-x^{(1)}_1 & 0\\
1\Sstrut & x^{(1)}_2-x^{(1)}_1 & x^{(2)}_{T_3}-x^{(2)}_{T_2}
\end{pmatrix}
\begin{array}{l}
\leftarrow T_1\Sstrut\\ \leftarrow T_2\Sstrut\\ \leftarrow T_3\Sstrut
\end{array}
\label{eq:E23}
\end{equation}
The entry in the first row and third column is
\(N_{T_3}(T_1)=x^{(2)}_{T_1}-x^{(2)}_{T_2}\), although \(T_1\) precedes
\(T_3\).  It is a difference of eigenvalues of \(T_1\) and \(T_2\), which are
not siblings.  Eigenvalues of nodes which are not siblings may coincide, so
such an entry is non-zero only for generic eigenvalues.  Here and below,
statements that such entries are non-zero, or that a graph Newton polynomial
does not vanish at a final node, refer to generic eigenvalues; the determinant
formula holds without this assumption.  For generic eigenvalues every row of
\eqref{eq:E23} has at least two nonzero entries, so no reordering of rows and
columns makes this matrix triangular.  Its
determinant is nevertheless
\((x^{(1)}_2-x^{(1)}_1)(x^{(2)}_{T_3}-x^{(2)}_{T_2})\), the product of the
diagonal entries.

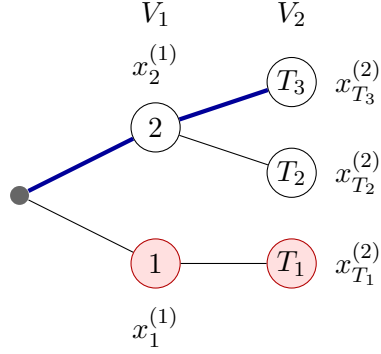
\begin{figure}[ht]
\centering
\begin{tikzpicture}
\node[rt] (r) at (0,1.2) {};
\node[gn,bad] (a1) at (1.8,0.3) {$1$};
\node[gn] (a2) at (1.8,2.1) {$2$};
\node[gn,bad] (t1) at (3.6,0.3) {$T_1$};
\node[gn] (t2) at (3.6,1.5) {$T_2$};
\node[gn] (t3) at (3.6,2.7) {$T_3$};
\draw (r)--(a1) (a1)--(t1) (a2)--(t2);
\draw[pth] (r)--(a2) (a2)--(t3);
\node[below=1mm] at (a1.south) {\small $x^{(1)}_1$};
\node[above=1mm] at (a2.north) {\small $x^{(1)}_2$};
\foreach \t/\k in {t1/1,t2/2,t3/3} \node[right=1mm] at (\t.east) {\small $x^{(2)}_{T_\k}$};
\node[font=\small] at (1.8,3.6) {$V_1$};
\node[font=\small] at (3.6,3.6) {$V_2$};

\end{tikzpicture}
\caption{The graph with \(D=(2,3)\), \(p_1=[1,2]\), whose evaluation matrix
\eqref{eq:E23} is not triangular in path order.  The path of \(T_3\) is drawn
thick.  Its rank datum is \(d_2=2\), and \(R_1(d_2{=}2)=S^{(1)}_{[2]}=\{2\}\):
node \(1\) has only one child, so it is not admissible and
\(N_{T_3}=C_2-x^{(2)}_{T_2}\) has no factor \(C_1-x^{(1)}_1\).  Hence
\(N_{T_3}\) does not vanish at the shaded node \(T_1\), for generic eigenvalues.}
\label{fig:L2}
\end{figure}

The same phenomenon occurs in other examples of \cite{KempRamgoolam2026}.
Two of them, the running example \(D=(2,4,7)\) of \cite{KempRamgoolam2026} and
the graph \(D=(4,10)\), \(p_1=[2,2,3,3]\), are worked out in
Appendix~\ref{app:examples}.

\paragraph{Why the determinant still factorises.}
Write \(b=x^{(1)}_2-x^{(1)}_1\) and \(c=x^{(2)}_{T_3}-x^{(2)}_{T_2}\) for the
diagonal differences in \eqref{eq:E23}.  Expanding \eqref{eq:E23} along the
first row gives
\begin{equation}
\det\cN
=
\det\begin{pmatrix}b&0\\ b&c\end{pmatrix}
+\bigl(x^{(2)}_{T_1}-x^{(2)}_{T_2}\bigr)
\det\begin{pmatrix}1&b\\ 1&b\end{pmatrix}
= bc .
\end{equation}
The minor multiplying the off-diagonal entry vanishes because its two rows
are equal.  They are the rows of \(T_2\) and \(T_3\), restricted to the
columns \(N_{T_1}=1\) and \(N_{T_2}=C_1-x^{(1)}_1\).  These polynomials involve
only \(C_1\), and \(T_2,T_3\) are siblings, so they have the same eigenvalue of
\(C_1\).  Hence the determinant is the product of the diagonal entries,
although the matrix is not triangular.

The same observation can be turned into a factorisation of \(\cN\), which is
the idea of the general proof.  Write
\(\mathbf r_{T'}=\bigl(\cN_{T',T_1},\cN_{T',T_2},\cN_{T',T_3}\bigr)\) for the row
of the final node \(T'\) in \eqref{eq:E23}.  The rows of \(T_2\) and \(T_3\)
agree in their first two entries and differ only in the last:
\(\mathbf r_{T_3}-\mathbf r_{T_2}=(0,0,c)\).  Thus
\begin{equation}
\mathbf r_{T_3}=\mathbf r_{T_2}+c\,\Delta_{2,2},
\qquad
\Delta_{2,2}
=\frac{\mathbf r_{T_3}-\mathbf r_{T_2}}{x^{(2)}_{T_3}-x^{(2)}_{T_2}}
=\bigl(0,\,0,\,1\bigr):
\label{eq:rowopexample}
\end{equation}
the row of \(T_3\) is the row of its sibling \(T_2\), plus the difference
\(c\) of their eigenvalues times a new row \(\Delta_{2,2}\).  In other words,
the rows of the siblings \(T_2,T_3\), which are the children of node
\(2\in V_1\), are expanded in the Newton basis \eqref{eq:newtonbasis1var} of
their eigenvalues of \(C_2\), here \(1,\ \xi-x^{(2)}_{T_2}\), with coefficients \(\mathbf r_{T_2}\) and
\(\Delta_{2,2}\).  (The subscript records the parent, \(2\), and the number of
siblings used, \(2\).)  The entries of \(\Delta_{2,2}\) can be understood
without computation: the columns \(N_{T_1},N_{T_2}\) do not involve \(C_2\),
so they take equal values on the siblings \(T_2,T_3\) and give \(0\); the
column \(N_{T_3}=C_2-x^{(2)}_{T_2}\) is linear in \(C_2\) with leading
coefficient \(1\), and gives \(1\).  Replacing the row of \(T_3\) by
\(\Delta_{2,2}\), the matrix \eqref{eq:E23} becomes
\begin{equation}
\cN'=
\left(\begin{array}{cc|c}
1\Sstrut & 0 & x^{(2)}_{T_1}-x^{(2)}_{T_2}\\
1\Sstrut & x^{(1)}_2-x^{(1)}_1 & 0\\ \hline
0\Sstrut & 0 & 1
\end{array}\right)
\begin{array}{l}
\leftarrow T_1\Sstrut\\ \leftarrow T_2\Sstrut\\ \leftarrow \Delta_{2,2}\Sstrut
\end{array}
\label{eq:E23after}
\end{equation}
This matrix is block triangular, and the entry which spoiled triangularity
now lies above the diagonal blocks.  The blocks are labelled by the position
\(d_2\) of a final node among its siblings.  The first block, with rows
\(T_1,T_2\) and columns \(N_{T_1},N_{T_2}\), corresponds to \(d_2=1\): its rows
are labelled by the nodes of \(V_1\) with at least one child, \(\{1,2\}\), and
it is the \(2\times2\) Newton--Vandermonde matrix of the eigenvalues
\(x^{(1)}_1,x^{(1)}_2\), i.e.\ the Newton evaluation matrix of the one-layer
graph obtained by deleting the final layer.  The second block, \((1)\),
corresponds to \(d_2=2\): its row is labelled by the nodes of \(V_1\) with at
least two children, \(\{2\}\), and it is the Newton evaluation matrix of the
one-layer graph obtained by deleting the final layer and keeping only node
\(2\).  These are the final-layer-excised graphs \(\Gamma_{[1]}\) and
\(\Gamma_{[2]}\), defined in general in Section~\ref{sec:general}.  There are
two blocks because the largest number of children of a node of \(V_1\) is
\(2\), and their sizes \(2\) and \(1\) are the numbers of nodes of \(V_1\) with at
least one and at least two children.

By \eqref{eq:rowopexample}, the original matrix is recovered as
\begin{equation}
\cN=W\cN',\qquad
W=\left(\begin{array}{c|cc}
1&0&0\\ \hline
0&1&0\\
0&1&c
\end{array}\right),
\label{eq:E23factor}
\end{equation}
with rows and columns of \(W\) in the order of \eqref{eq:E23} and
\eqref{eq:E23after}.  The block of \(W\) for the children \(T_2,T_3\) of node
\(2\) is \(\begin{pmatrix}1&0\\1&c\end{pmatrix}\), the Newton evaluation matrix
of the one-layer graph formed by these two siblings, with their eigenvalues
of \(C_2\); the block for the single child \(T_1\) of node \(1\) is \((1)\).
Since \(\det W=c\),
\begin{equation}
\det\cN=c\cdot\det\begin{pmatrix}1&0\\1&b\end{pmatrix}\cdot 1
=bc .
\end{equation}

\subsubsection{The general case: block triangular form and induction}
\label{sec:general}

In general a node of layer \(L-1\) may have more than two children.  The
relation \eqref{eq:rowopexample} then generalises to the expansion of the
rows of siblings in the Newton basis \eqref{eq:newtonbasis1var} of their
eigenvalues (Lemma~\ref{lem:newton} below).  As in the example, it
writes \(\cN\) as a block diagonal matrix of one-layer Newton--Vandermonde
matrices, one for each set of siblings, times a block triangular matrix whose
diagonal blocks are Newton evaluation matrices of smaller graphs, obtained by
excising the final layer.  An induction on the number of layers then gives the
result.

The diagonal entries of the Newton evaluation matrix are
\begin{align}
\cN_{T,T}=N_T(T)
&=
\prod_{\substack{w<v_1\\w\in R_1(d_2,\ldots,d_L)}}
\bigl(x_{v_1}^{(1)}-x_w^{(1)}\bigr)
\nonumber\\
&\quad\times\prod_{i=2}^{L-1}
\prod_{\substack{w<v_i\\
w\in R_i(v_{i-1};d_{i+1},\ldots,d_L)}}
\bigl(x_{v_i}^{(i)}-x_w^{(i)}\bigr)
\;\prod_{\substack{w<v_L\\w\in R_L(v_{L-1})}}
\bigl(x_{v_L}^{(L)}-x_w^{(L)}\bigr),
\label{eq:diagonal}
\end{align}
for \(L\ge2\); for \(L=1\), \(\cN_{T,T}\) is the single product over
\(R_1=V_1\), as for \eqref{eq:graphNewton}.  In the examples above the determinant of \(\cN\) is the product of these
diagonal entries, whether or not \(\cN\) is triangular.  This holds in
general.

\begin{theorem}[Determinant of the Newton evaluation matrix]\label{thm:newtondet}
For every degeneracy graph, the determinant of the Newton evaluation matrix
\(\cN\) of \eqref{eq:Nmatrix} is the product of its diagonal entries
\eqref{eq:diagonal}:
\begin{equation}
\det\cN
=
\prod_{T\in V_L}\cN_{T,T}.
\label{eq:diagonalproduct}
\end{equation}
\end{theorem}

By \eqref{eq:changebasisdet}, this is equivalent to
\(\det\cM=\prod_{T\in V_L}\cN_{T,T}\).  The proof is by induction on the number
of layers.  It uses the final-layer-excised graphs of
Definition~\ref{def:final-layer-excised}, through Lemma~\ref{lem:restrict} and
Corollary~\ref{cor:factor}, and two lemmas on the Newton expansion among
siblings, Lemmas~\ref{lem:newton} and~\ref{lem:degree}; it is given after
them.

\paragraph{The final-layer-excised graphs \(\Gamma_{[d]}\).}

\begin{definition}\label{def:final-layer-excised}
Let \(\Gamma\) be the degeneracy graph, with \(L\geq2\), and let
\(1\le d\le d_{\max}\), where \(d_{\max}\) is the largest number of children of
a node of \(V_{L-1}\).
The final-layer-excised graph \(\Gamma_{[d]}\) is the graph with \(L-1\) layers obtained
from \(\Gamma\) by
\begin{itemize}
\item deleting the final layer \(V_L\);
\item keeping at layer \(L-1\) only the nodes of \(S^{(L-1)}_{[d]}\), i.e.\ the
  nodes with at least \(d\) children;
\item keeping at each layer \(i<L-1\) only the nodes having a descendant in
  \(S^{(L-1)}_{[d]}\).
\end{itemize}
At each layer the surviving nodes are relabelled \(1,2,\ldots\), keeping
their relative order.  This preserves which nodes are siblings and the order
among siblings.  We identify each surviving node with the node of \(\Gamma\)
it comes from, so that sets of nodes of \(\Gamma_{[d]}\) may be compared with
sets of nodes of \(\Gamma\).  The final layer of \(\Gamma_{[d]}\) is
\(S^{(L-1)}_{[d]}\).
\end{definition}

The graph \(\Gamma_{[d]}\) is in the spirit of the truncated graphs
used in Section~4 of \cite{KempRamgoolam2026} to prove that the number of
graph monomials is \(D_L\): in both cases, nodes with too few children are
dropped, so that the sets \(S^{(i)}\) of the original graph become sets of the
same type in a smaller graph.  Here this is done one layer at a time, and is
the basis of the induction.

Quantities computed in \(\Gamma_{[d]}\) carry an argument \((\Gamma_{[d]})\);
quantities without an argument refer to \(\Gamma\).

\begin{lemma}\label{lem:restrict}
Let \(1\le d\le d_{\max}\).  For \(1\le i\le L-2\) and all \(d_{i+1},\ldots,d_{L-1}\),
\begin{equation}
S^{(i)}_{[d_{i+1},\ldots,d_{L-1},\,d]}
=
S^{(i)}_{[d_{i+1},\ldots,d_{L-1}]}(\Gamma_{[d]}).
\label{eq:Srestrict}
\end{equation}
Consequently, for \(1\le i\le L-1\),
\begin{equation}
R_i(v_{i-1};d_{i+1},\ldots,d_{L-1},d)
=
R_i(v_{i-1};d_{i+1},\ldots,d_{L-1})(\Gamma_{[d]}),
\label{eq:Rrestrict}
\end{equation}
where for \(i=L-1\) the right-hand side is the final-layer set
\(R_{L-1}(v_{L-2})(\Gamma_{[d]})=\Ch_{L-1}(v_{L-2})\cap S^{(L-1)}_{[d]}\) of
\eqref{eq:RL}.
\end{lemma}

\begin{proof}
Induction on \(i\), decreasing from \(L-2\).  For \(i=L-2\): a node \(u\in V_{L-2}\) lies in
\(S^{(L-2)}_{[d_{L-1},d]}\) exactly when
\(|\Ch_{L-1}(u)\cap S^{(L-1)}_{[d]}|\ge d_{L-1}\).  The set
\(\Ch_{L-1}(u)\cap S^{(L-1)}_{[d]}\) is the set of children of \(u\) in
\(\Gamma_{[d]}\), which lie in its final layer, so this is the condition
\(u\in S^{(L-2)}_{[d_{L-1}]}(\Gamma_{[d]})\).  (Such a \(u\) has a descendant
in \(S^{(L-1)}_{[d]}\), so it is a node of \(\Gamma_{[d]}\).)  The inductive step is the same argument one layer to the left, using \eqref{eq:Srestrict} at layer
\(i+1\).  Equation \eqref{eq:Rrestrict} follows by intersecting with
\(\Ch_i(v_{i-1})\), using \eqref{eq:R1}--\eqref{eq:Ri}.
\end{proof}

\begin{corollary}[Factorisation of graph Newton polynomials]\label{cor:factor}
Let \(T\in V_L\) have \(v_{L-1}(T)=q\) and rank data
\((a,d_2,\ldots,d_{L-1},d)\).  Then \(q\in S^{(L-1)}_{[d]}\), so \(q\) is a
final node of \(\Gamma_{[d]}\); its rank data in \(\Gamma_{[d]}\) are
\((a,d_2,\ldots,d_{L-1})\); and
\begin{equation}
N_T
=
N_q(\Gamma_{[d]})(C_1,\ldots,C_{L-1})
\prod_{\substack{w<T\\ w\in\Ch_L(q)}}\bigl(C_L-x^{(L)}_w\bigr).
\label{eq:Nfactor}
\end{equation}
In particular the final nodes of \(\Gamma\) are in bijection with the pairs
\((q,d)\) with \(q\in S^{(L-1)}_{[d]}\), via \(T\mapsto(v_{L-1}(T),d_L(T))\):
\(T\) is the \(d\)-th child of \(q\).
\end{corollary}

\begin{proof}
\(T\) is the \(d\)-th child of \(q\), so \(q\) has at least \(d\) children.
By \eqref{eq:Rrestrict}, the admissible sets along the path of \(T\) at layers
\(1,\ldots,L-1\) coincide with those along the path of \(q\) in
\(\Gamma_{[d]}\); hence so do the ranks, and the corresponding factors in
\eqref{eq:graphNewton}.  The remaining factor is the layer-\(L\) product.
\end{proof}

\paragraph{Newton expansion among siblings.}
Fix \(q\in V_{L-1}\), let \(n=|\Ch_L(q)|\) be its number of children, and for
\(1\le k\le n\) let \(\xi_k\) be the eigenvalue of \(C_L\) on the \(k\)-th
child of \(q\), in the order of Section~\ref{sec:ordconv}.  These eigenvalues
are distinct, since the children of \(q\) are siblings.

\emph{The goal.}  In the proof of Theorem~\ref{thm:newtondet} below, the
rows of \(\cN\) belonging to the children of \(q\) are expanded in the Newton
basis \eqref{eq:newtonbasis1var} of \(\xi_1,\ldots,\xi_n\): the row of the \(k\)-th child is
written as a combination of \(n\) new rows \(\Delta_{q,1},\ldots,\Delta_{q,n}\),
with differences of eigenvalues of siblings as explicit coefficients.  This
generalises \eqref{eq:rowopexample}, where \(n=2\): there
\(\mathbf r_{T_3}=\Delta_{2,1}+c\,\Delta_{2,2}\), with \(\Delta_{2,1}\) the row
of \(T_2\).  The row \(\Delta_{q,d}\) exists exactly when \(q\) has at least
\(d\) children, i.e.\ when \(q\in S^{(L-1)}_{[d]}\), which is the final layer
of \(\Gamma_{[d]}\).  Thus the integer \(d\) plays three roles which coincide:
it labels the \(d\)-th polynomial of the Newton basis, the position of a final
node \(T\) among its siblings in the label \((q,d)\) of
Corollary~\ref{cor:factor}, and the threshold ``at least \(d\) children'' in
\(\Gamma_{[d]}\).  This is why the rows \(\Delta_{q,d}\) with fixed \(d\) will
label the rows of the diagonal block given by \(\Gamma_{[d]}\).

\emph{The Newton basis.}  Newton interpolation in one variable is classical;
see for instance \cite{deBoor,StoerBulirsch,Davis}.  Let \(f\) be a function on \(V_L\).  In the proof of
Theorem~\ref{thm:newtondet}, \(f\) will be a column of \(\cN\),
\(f(T')=N_T(T')\); in Section~\ref{sec:projectors} it will also be the
indicator function of a final node.  Since the \(\xi_k\) are distinct, the
values of \(f\) on the children of \(q\) may be labelled by their eigenvalues:
we write \(f(\xi_k)\) for the value of \(f\) at the \(k\)-th child of \(q\).
This labelling is used only within one set of siblings, since final nodes
which are not siblings may share an eigenvalue of \(C_L\).  The Newton basis
\eqref{eq:newtonbasis1var} of \(\xi_1,\ldots,\xi_n\), which we call the Newton
basis of the children of \(q\), consists of the polynomials
\begin{equation}
\pi_{q,d}(\xi)=\prod_{m<d}(\xi-\xi_m),\qquad 1\le d\le n,
\label{eq:newtonbasis}
\end{equation}
which are monic of degree \(d-1\), and vanish at \(\xi_k\) for \(k<d\).

\begin{lemma}[Newton expansion among siblings]\label{lem:newton}
Let \(f\) be a function on \(V_L\).  There are unique numbers
\(\Delta_{q,1}f,\ldots,\Delta_{q,n}f\), the Newton coefficients of \(f\), such
that
\begin{equation}
f(\xi_k)=\sum_{d=1}^{k}\pi_{q,d}(\xi_k)\,\Delta_{q,d}f,
\qquad 1\le k\le n ,
\label{eq:newtoninterp}
\end{equation}
and \(\Delta_{q,d}f\) depends only on the values of \(f\) at the first \(d\)
children of \(q\).  The \(n\times n\) matrix \(W_q\) with entries
\((W_q)_{k,d}=\pi_{q,d}(\xi_k)\), the Newton--Vandermonde matrix of
\(\xi_1,\ldots,\xi_n\), is the Newton evaluation matrix of the
one-layer graph formed by the children of \(q\), with their eigenvalues of
\(C_L\), and
\begin{equation}
\det W_q=\prod_{1\le k<m\le n}(\xi_m-\xi_k).
\label{eq:fact2}
\end{equation}
\end{lemma}

\begin{proof}
Since \(\pi_{q,d}(\xi_k)=0\) for \(k<d\), the matrix \(W_q\) is lower
triangular, with diagonal entries \(\pi_{q,k}(\xi_k)=\prod_{m<k}(\xi_k-\xi_m)\),
which are non-zero because siblings have distinct eigenvalues.  The equations
\eqref{eq:newtoninterp} form the triangular system with matrix \(W_q\), so
they have a unique solution, and their first \(d\) equations determine
\(\Delta_{q,1}f,\ldots,\Delta_{q,d}f\).  The determinant of \(W_q\) is the
product of its diagonal entries.  Finally, the graph Newton polynomials of the
one-layer graph formed by the children of \(q\) are the \(\pi_{q,d}(C_L)\), as
in the one-layer example of Section~\ref{sec:evalNewton}, so \(W_q\) is its
Newton evaluation matrix.
\end{proof}

For instance \(\Delta_{q,1}f=f(\xi_1)\), so \(\Delta_{q,1}\) keeps the row of the
first child.  Applied to each column of \(\cN\), \(\Delta_{q,d}\) gives the
entries of the new row, which we denote by the same symbol.  For \(n=3\),
\begin{equation}
W_q=\begin{pmatrix}
1&0&0\\
1&\xi_2-\xi_1&0\\
1&\xi_3-\xi_1&(\xi_3-\xi_1)(\xi_3-\xi_2)
\end{pmatrix}.
\label{eq:Wn3}
\end{equation}
The second property of the Newton coefficients produces the zero blocks and
the diagonal blocks of the block triangular form.

\begin{lemma}[Degree property]\label{lem:degree}
Let \(1\le d\le n\).  Suppose that \(f(\xi_k)=h(\xi_k)\), \(1\le k\le d\), for
a polynomial \(h\) of degree at most \(e\), and let \(c\) be the coefficient of
\(\xi^{e}\) in \(h\) (possibly zero).  Then
\begin{equation}
\Delta_{q,d}f=
\begin{cases}
0, & d-1>e,\\
c, & d-1=e.
\end{cases}
\label{eq:fact1}
\end{equation}
\end{lemma}

\begin{proof}
In both cases \(e\le d-1\).  Since \(\pi_{q,1},\ldots,\pi_{q,e+1}\) are monic of
degrees \(0,\ldots,e\), we can write \(h=\sum_{d'=1}^{e+1}a_{d'}\,\pi_{q,d'}\),
with \(a_{e+1}=c\).  Evaluating at \(\xi_1,\ldots,\xi_d\) and comparing with the
first \(d\) equations of \eqref{eq:newtoninterp}, which determine
\(\Delta_{q,1}f,\ldots,\Delta_{q,d}f\) by Lemma~\ref{lem:newton}, gives
\(\Delta_{q,d'}f=a_{d'}\) for \(d'\le e+1\) and \(\Delta_{q,d'}f=0\) for
\(e+1<d'\le d\).
\end{proof}

For example, \(h(\xi)=\xi-\xi_1=\pi_{q,2}(\xi)\) gives \(\Delta_{q,2}f=1\); this
is the last entry of \(\Delta_{2,2}\) in \eqref{eq:rowopexample}.  Similarly,
on three children, \(\xi=\xi_1+\pi_{q,2}(\xi)\) and
\(\xi^2=\xi_1^2+(\xi_1+\xi_2)\,\pi_{q,2}(\xi)+\pi_{q,3}(\xi)\), so \(\Delta_{q,3}\)
is \(0\) on \(f(\xi_k)=\xi_k\) (here \(e=1<d-1=2\)) and \(1\) on
\(f(\xi_k)=\xi_k^2\) (here \(e=d-1=2\) and \(c=1\)).

\begin{proof}[Proof of Theorem~\ref{thm:newtondet}]
The proof is by induction on the number of layers.  Each step excises the
final layer: the Newton expansion acts among the children of the nodes of the
current penultimate layer, and the diagonal blocks are the Newton evaluation
matrices of the graphs \(\Gamma_{[d]}\), which have one layer fewer.
Unwinding the induction therefore excises the layers \(L,L-1,\ldots,2\) in
turn, each step being applied to every graph produced by the previous one,
until one-layer graphs remain, whose Newton evaluation matrices are
triangular.

For \(L=1\), \(R_1=V_1\), the graph Newton polynomials
are the ordinary Newton polynomials, and the matrix is triangular with the
stated diagonal.

Let \(L\ge2\).  By Corollary~\ref{cor:factor}, both the rows (points \(T'\))
and the columns (polynomials \(N_T\)) of the matrix are labelled by pairs
\((q,d)\) with \(q\in S^{(L-1)}_{[d]}\).  For each \(q\in V_{L-1}\), expand the
rows belonging to the children of \(q\) in the Newton basis of their
eigenvalues, as in Lemma~\ref{lem:newton}, and let \(\cN'\) be the matrix of
Newton coefficients: its entry in row \((q,d)\) and column \(T\) is
\(\cN'_{(q,d),T}=\Delta_{q,d}N_T\), the \(d\)-th Newton coefficient of the
column \(T'\mapsto N_T(T')\) on the children of \(q\).  Then
\begin{equation}
\cN=W\cN',
\label{eq:NPfactor}
\end{equation}
where \(W\) is block diagonal, with one block for the children of each
\(q\in V_{L-1}\), and the block for \(q\) is the one-layer Newton--Vandermonde
matrix \(W_q\) of Lemma~\ref{lem:newton}.  By \eqref{eq:fact2},
\(\det W=\prod_q\det W_q=\prod_{q}\prod_{k<m}(\xi_m-\xi_k)\).

Consider the entry \(\cN'_{(q,d),T}\), where \(T\) is the final node labelled \((q',d')\).  On the children of \(q\), the eigenvalues of
\(C_1,\ldots,C_{L-1}\) are those of the path to \(q\), so by \eqref{eq:Nfactor}
\(N_T\) restricts to a polynomial in \(\xi\) of degree at most \(d'-1\), whose
coefficient of \(\xi^{d'-1}\) is \(N_{q'}(\Gamma_{[d']})(q)\); this value may
vanish, in which case the degree is smaller.  By the degree property,
Lemma~\ref{lem:degree},
\begin{equation}
\cN'_{(q,d),T}=0 \quad (d>d'),
\qquad
\cN'_{(q,d),T}=N_{q'}(\Gamma_{[d]})(q)=\cN(\Gamma_{[d]})_{q,q'}\quad (d=d').
\end{equation}
Ordering rows and columns by \(d\), the matrix \(\cN'\) is block triangular.  The
diagonal block for \(d\) has rows and columns labelled by \(S^{(L-1)}_{[d]}\),
the final layer of \(\Gamma_{[d]}\), and is the Newton evaluation matrix of
\(\Gamma_{[d]}\).

The diagonal blocks are labelled by the values of \(d\) for which
\(S^{(L-1)}_{[d]}\) is non-empty.  Since these sets decrease as \(d\) increases,
these values are \(d=1,\ldots,d_{\max}\), where \(d_{\max}\) is the largest
number of children of a node of \(V_{L-1}\): there are \(d_{\max}\) diagonal
blocks.  The block for \(d\) has size \(|S^{(L-1)}_{[d]}|\), the number of nodes
of \(V_{L-1}\) with at least \(d\) children.  Thus, if the numbers of children
of the nodes of \(V_{L-1}\) are arranged as a partition of \(D_L\), the sizes of
the diagonal blocks are the parts of the conjugate partition.  Hence, by
\eqref{eq:NPfactor}, and since the block triangular form of \(\cN'\) is
reached by a simultaneous reordering of rows and columns,
\begin{equation}
\det\cN=\det W\,\det\cN'
=\prod_{d=1}^{d_{\max}}\det\cN(\Gamma_{[d]})
\ \times\prod_{q\in V_{L-1}}\ \prod_{\substack{b<b'\\ b,b'\in\Ch_L(q)}}
\bigl(x^{(L)}_{b'}-x^{(L)}_{b}\bigr).
\label{eq:detNblocks}
\end{equation}
Each diagonal block contributes its determinant \(\det\cN(\Gamma_{[d]})\)
exactly once.  The product over \(q\in V_{L-1}\) acts only on the second
factor, which is \(\det W\) by \eqref{eq:fact2}, with one Vandermonde factor
\(\det W_q\) for each set of siblings in \(V_L\).

By induction,
\(\det\cN(\Gamma_{[d]})=\prod_{q\in S^{(L-1)}_{[d]}}\cN(\Gamma_{[d]})_{q,q}\).
On the other hand, by \eqref{eq:Nfactor}, for \(T\) the \(d\)-th child of
\(q\),
\begin{equation}
\cN_{T,T}=\cN(\Gamma_{[d]})_{q,q}\prod_{\substack{w<T\\ w\in\Ch_L(q)}}
\bigl(x^{(L)}_T-x^{(L)}_w\bigr).
\end{equation}
Since the final nodes \(T\) correspond to the pairs \((q,d)\) with
\(q\in S^{(L-1)}_{[d]}\) (Corollary~\ref{cor:factor}), the product over all
\(T\in V_L\) of the first factor is \(\prod_{d}\det\cN(\Gamma_{[d]})\), and
the product of the second factor is the second factor of
\eqref{eq:detNblocks}.  This proves \eqref{eq:diagonalproduct}.
\end{proof}

\begin{remark}[Divided differences]\label{rem:divdiff}
The Newton coefficients have a classical closed form: \(\Delta_{q,d}f\) is the
divided difference of \(f\) over the first \(d\) children of \(q\)
\cite{deBoor,StoerBulirsch},
\begin{equation}
\Delta_{q,d}f
=\sum_{k=1}^{d}
\frac{f(\xi_k)}{\prod_{m\le d,\,m\ne k}(\xi_k-\xi_m)} ,
\label{eq:divdiff}
\end{equation}
so that \(W_q^{-1}=Q_q\), the matrix with entries
\((Q_q)_{d,k}=\prod_{m\le d,\,m\ne k}(\xi_k-\xi_m)^{-1}\) for \(k\le d\) and
\(0\) for \(k>d\).  For the first few values of \(d\),
\begin{align}
\Delta_{q,1}f&=f(\xi_1),\nonumber\\
\Delta_{q,2}f&=\frac{f(\xi_2)-f(\xi_1)}{\xi_2-\xi_1},\nonumber\\
\Delta_{q,3}f&=\frac{f(\xi_1)}{(\xi_1-\xi_2)(\xi_1-\xi_3)}
+\frac{f(\xi_2)}{(\xi_2-\xi_1)(\xi_2-\xi_3)}
+\frac{f(\xi_3)}{(\xi_3-\xi_1)(\xi_3-\xi_2)} ,
\label{eq:divdiff123}
\end{align}
and in general they can be computed recursively: writing
\(f[\xi_j,\ldots,\xi_{j+d-1}]\) for the divided difference over the
consecutive children from the \(j\)-th to the \((j+d-1)\)-th, with
\(f[\xi_k]=f(\xi_k)\), one has \(\Delta_{q,d}f=f[\xi_1,\ldots,\xi_d]\) and, for
\(d\ge2\),
\begin{equation}
f[\xi_j,\ldots,\xi_{j+d-1}]
=\frac{f[\xi_{j+1},\ldots,\xi_{j+d-1}]-f[\xi_j,\ldots,\xi_{j+d-2}]}
{\xi_{j+d-1}-\xi_j}.
\label{eq:divdiffrec}
\end{equation}
For \(d=2\) this is the difference quotient in \eqref{eq:rowopexample}.  The
closed form is not needed in the proof of Theorem~\ref{thm:newtondet}; it is
used in Section~\ref{sec:inverse-evaluation}, where \(\cN\) is inverted.
\end{remark}

\paragraph{Where the factors arise.}  The sibling differences
\(x^{(i)}_q-x^{(i)}_p\) of layer \(i\) are produced at the stage of the
recursion at which layer \(i\) has become the final layer, as the
determinants of the one-layer Newton--Vandermonde matrices \(W_{v_{i-1}}\) of the
children of the nodes \(v_{i-1}\) of layer \(i-1\) (for \(i=1\), as the determinants
of the one-layer graphs at the end of the induction).  For \(i=L\) this is the first
step, applied to \(\Gamma\) itself, and each final-layer difference arises
once.  For \(i<L\), by Lemma~\ref{lem:restrict}
applied repeatedly, the graphs reached at that stage are labelled by tuples
\((d_{i+1},\ldots,d_L)\), and their final layer consists of the nodes of
\(S^{(i)}_{[d_{i+1},\ldots,d_L]}\).  A pair of siblings \(p<q\) in layer \(i\)
therefore contributes the factor \(x^{(i)}_q-x^{(i)}_p\) once for each tuple
with \(p,q\in S^{(i)}_{[d_{i+1},\ldots,d_L]}\), and a sibling difference can
appear with an exponent greater than \(1\).  This count of tuples is the
exponent computed in Section~\ref{sec:exponents}.  Remark~\ref{rem:L3} below
shows the recursion for a graph with three layers, and in the example of
Appendix~\ref{app:D247} the factor \(x^{(1)}_2-x^{(1)}_1\) appears twice in
one diagonal block and once in the other.

\begin{remark}[An example with \(L=3\)]\label{rem:L3}
Let \(D=(2,3,4)\).  Node \(1\in V_1\) has one child, node \(1\in V_2\), which
has children \(T_1<T_2\).  Node \(2\in V_1\) has children \(2<3\) in \(V_2\),
with one child each, \(T_3\) and \(T_4\).  The rank data give
\begin{equation}
N_{T_1}=1,\quad N_{T_2}=C_3-x^{(3)}_{T_1},\quad N_{T_3}=C_1-x^{(1)}_1,\quad
N_{T_4}=C_2-x^{(2)}_2 ,
\end{equation}
and \(N_{T_4}(T_1)=x^{(2)}_1-x^{(2)}_2\), which is generically non-zero.  Here,
for generic eigenvalues, no relabelling of the nodes makes the evaluation
matrix triangular.  Indeed, the final node \(T_0\)
whose monomial is \(1\) has \(N_{T_0}=1\), so every other \(N_T\) would have to
vanish at \(T_0\).  The polynomial with monomial \(C_2\) (rank data
\(d_2=2,d_3=1\), for which only node \(2\in V_1\) is admissible) is
\(C_2-x^{(2)}_{w}\) with \(w\) the first child of node \(2\in V_1\); for
generic eigenvalues it vanishes at \(T_0\) only if \(T_0\) descends from node
\(2\in V_1\).  The
polynomial with monomial \(C_3\) (rank data \(d_2=1,d_3=2\), for which only
node \(1\in V_1\) is admissible) is \(C_3-x^{(3)}_{w'}\) with \(w'\) the first
child of node \(1\in V_2\); for generic eigenvalues it vanishes at \(T_0\) only if \(T_0\)
descends from node \(1\in V_1\).  Both cannot hold.  We return to this
example in Remark~\ref{rem:nested}.

The induction step in the proof of Theorem~\ref{thm:newtondet} removes layer
\(3\): the only Newton expansion with more than one term is that of the rows
of \(T_1,T_2\), and the diagonal blocks are the evaluation matrices of
\(\Gamma_{[1]}\), which is the graph of Figure~\ref{fig:L2}, and of
\(\Gamma_{[2]}\), a single path (Figure~\ref{fig:L3}).  The nodes of \(V_2\)
have \(2,1,1\) children, so there are \(d_{\max}=2\) blocks, of sizes \(3\) and
\(1\), the parts of the conjugate of the partition \((2,1,1)\).  Applying the
induction step once more to \(\Gamma_{[1]}\) gives
\begin{equation}
\det\cN=\bigl(x^{(3)}_{T_2}-x^{(3)}_{T_1}\bigr)
\bigl(x^{(1)}_2-x^{(1)}_1\bigr)\bigl(x^{(2)}_3-x^{(2)}_2\bigr).
\end{equation}
\end{remark}

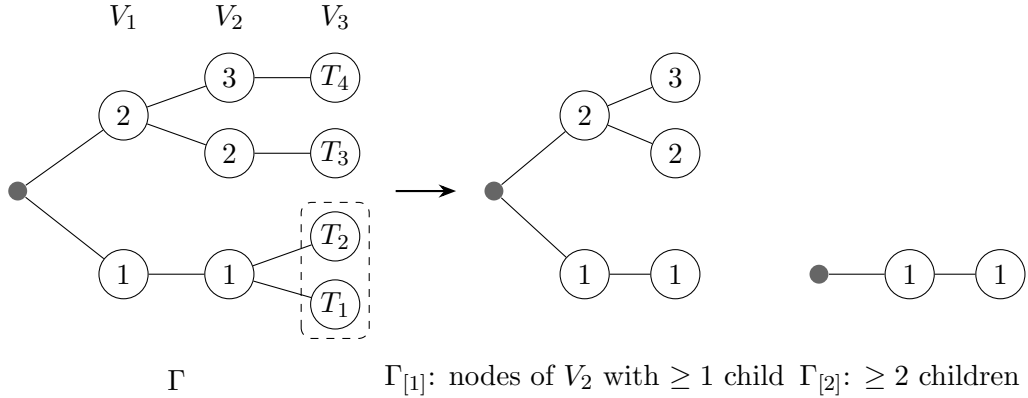
\begin{figure}[ht]
\centering
\begin{tikzpicture}
\node[rt] (r) at (0,1.5) {};
\node[gn] (a1) at (1.4,0.4) {$1$};
\node[gn] (a2) at (1.4,2.5) {$2$};
\node[gn] (b1) at (2.8,0.4) {$1$};
\node[gn] (b2) at (2.8,2.0) {$2$};
\node[gn] (b3) at (2.8,3.0) {$3$};
\node[gn] (t1) at (4.2,0) {$T_1$};
\node[gn] (t2) at (4.2,0.9) {$T_2$};
\node[gn] (t3) at (4.2,2.0) {$T_3$};
\node[gn] (t4) at (4.2,3.0) {$T_4$};
\draw (r)--(a1) (r)--(a2) (a1)--(b1) (a2)--(b2) (a2)--(b3)
 (b1)--(t1) (b1)--(t2) (b2)--(t3) (b3)--(t4);
\draw[dashed,rounded corners] (3.75,-0.45) rectangle (4.65,1.35);
\foreach \x/\l in {1.4/V_1,2.8/V_2,4.2/V_3} \node[font=\small] at (\x,3.8) {$\l$};
\node[font=\small] at (2.1,-1.0) {$\Gamma$};
\draw[-{Stealth},thick] (5.0,1.5)--(5.8,1.5);
\node[rt] (s) at (6.3,1.5) {};
\node[gn] (c1) at (7.5,0.4) {$1$};
\node[gn] (c2) at (7.5,2.5) {$2$};
\node[gn] (e1) at (8.7,0.4) {$1$};
\node[gn] (e2) at (8.7,2.0) {$2$};
\node[gn] (e3) at (8.7,3.0) {$3$};
\draw (s)--(c1) (s)--(c2) (c1)--(e1) (c2)--(e2) (c2)--(e3);
\node[font=\small,align=center] at (7.5,-1.0) {$\Gamma_{[1]}$: nodes of \(V_2\) with \(\ge1\) child};
\node[rt] (u) at (10.6,0.4) {};
\node[gn] (f1) at (11.8,0.4) {$1$};
\node[gn] (g1) at (13.0,0.4) {$1$};
\draw (u)--(f1)--(g1);
\node[font=\small,align=center] at (11.8,-1.0) {$\Gamma_{[2]}$: \(\ge2\) children};
\end{tikzpicture}
\caption{The graph \(\Gamma\) of Remark~\ref{rem:L3} and its final-layer-excised graphs.
The dashed box encloses the only set of siblings in \(V_3\) with more than one
element; the Newton expansion acts there.  \(\Gamma_{[1]}\) is the graph of
Figure~\ref{fig:L2}, so the \(L=3\) example reduces to the \(L=2\) example,
which reduces in turn to the one-variable Newton--Vandermonde case.}
\label{fig:L3}
\end{figure}

\begin{remark}[When the evaluation matrix is triangular]\label{rem:nontri}
The mechanism of the triangular example applies to every pair \(T'\), \(T\)
when every admissible set \(R_i(v_{i-1};d_{i+1},\ldots,d_L)\) is an initial
segment of \(\Ch_i(v_{i-1})\), for all rank data.  The Newton evaluation
matrix is then lower triangular in path order.  For \(L=2\), the admissible
set \(R_1(d_2)=S^{(1)}_{[d_2]}\) consists of the nodes of \(V_1\) with at least
\(d_2\) children, and it is an initial segment for every \(d_2\) exactly when
\(p_{1,1}\ge p_{1,2}\ge\cdots\); with the increasing convention (2.8) of
\cite{KempRamgoolam2026} this holds only when all parts of \(p_1\) are equal,
as in Figure~\ref{fig:D24}.  In Remark~\ref{rem:L3}, for generic eigenvalues, no labelling of the nodes
makes the matrix triangular.  In Section~\ref{sec:young} the condition is
reformulated as a nesting condition on Young diagrams attached to the nodes
(Remark~\ref{rem:nested}).
\end{remark}

\subsection{The sibling-pair-truncated graph and the exponents}
\label{sec:exponents}

Fix a layer \(i\), a parent \(u\in V_{i-1}\), and siblings
\begin{equation}
p<q,
\qquad
p,q\in\Ch_i(u).
\end{equation}
Let \(G_{\mathrm{trunc}}(i,p,q)\) be the sibling-pair-truncated graph obtained by dropping
the layers to the left of \(i\), keeping only \(p,q\) at layer \(i\), and retaining all
their descendants and the connecting edges through layer \(L\).

For \(i=L\), set
\begin{equation}
\Exp(L,p,q)=1.
\label{eq:expfinal}
\end{equation}
For \(1\leq i<L\), define
\begin{equation}
\Exp(i,p,q)
=
\#\left\{
(d_{i+1},\ldots,d_L):
\left|
S^{(i)}_{[d_{i+1},\ldots,d_L]}
\bigl(G_{\mathrm{trunc}}(i,p,q)\bigr)
\right|=2
\right\}.
\label{eq:exptrunc}
\end{equation}
Equivalently, \(\Exp(i,p,q)\) is the number of graph monomials of
\(G_{\mathrm{trunc}}(i,p,q)\) which contain its first generator.

A final path contributes the factor
\begin{equation}
x_q^{(i)}-x_p^{(i)}
\end{equation}
to the diagonal product precisely when it passes through \(q\) at layer \(i\)
and \(p\) is an earlier admissible sibling of \(q\).  The remaining rank data
\begin{equation}
(d_{i+1},\ldots,d_L)
\end{equation}
must be such that both \(p\) and \(q\) survive in \(G_{\mathrm{trunc}}(i,p,q)\),
that is, both lie in \(S^{(i)}_{[d_{i+1},\ldots,d_L]}\).  Conversely each such
tuple determines exactly one final path through \(q\), by \eqref{eq:vi}--\eqref{eq:vL}
started at \(q\).  So the number of final paths contributing the factor is
the number of tuples in \eqref{eq:exptrunc}.  Therefore the multiplicity of
the sibling difference is \(\Exp(i,p,q)\).

\begin{theorem}[Kemp--Ramgoolam determinant formula]\label{thm:KR}
For a labelled Kemp--Ramgoolam degeneracy graph, with the rows and columns of the graph
monomial evaluation matrix \(\cM\) of \eqref{eq:Mmatrix} both labelled by
\(V_L\), through the bijection \eqref{eq:bijection}, and listed in the same
order,
\begin{equation}
\det\cM
=
\prod_{i=1}^{L}
\prod_{u\in V_{i-1}}
\prod_{\substack{p<q\\p,q\in\Ch_i(u)}}
\bigl(x_q^{(i)}-x_p^{(i)}\bigr)^{\Exp(i,p,q)}.
\label{eq:KRdet}
\end{equation}
\end{theorem}

\paragraph{Signs and labellings.}  Equation \eqref{eq:KRdet} is an exact
equality, not one up to sign.  This rests on two facts: rows and columns of
\(\cM\) are listed in the same order, so that a simultaneous reordering does
not change \(\det\cM\); and each factor on the right is oriented as
\(x^{(i)}_q-x^{(i)}_p\) with \(p<q\), as are the diagonal entries
\eqref{eq:diagonal}.  If the graph monomials are instead listed in some other
order relative to the final nodes, differing from \eqref{eq:bijection} by a
permutation \(\sigma\) of \(V_L\), then \(\det\cM\) is multiplied by
\(\operatorname{sgn}\sigma\), and \eqref{eq:KRdet} holds up to this sign.  A
change of the labelling of Section~\ref{sec:ordconv} changes both sides
together: the pairing of graph monomials with final nodes on the left, and the
orientation of the sibling differences on the right.  Since the proof uses
only the order among siblings, \eqref{eq:KRdet} holds exactly in every such
labelling.  In \cite{KempRamgoolam2026} the columns of the matrix (6.9) are
labelled by exponent vectors, with no fixed order relative to the rows, and
the formula is stated up to sign.

\begin{remark}
For \(L=1\), the graph monomials are
\begin{equation}
1,C_1,C_1^2,\ldots,C_1^{D_1-1},
\end{equation}
the graph Newton polynomials are the ordinary Newton polynomials, and
\eqref{eq:KRdet} reduces to the usual Vandermonde determinant; see the
one-layer example in Section~\ref{sec:evalNewton} and Figure~\ref{fig:L1}.
\end{remark}

\section{Geometry of the exponents: higher-dimensional Young diagrams}
\label{sec:young}

In Section~\ref{sec:exponents} the exponent of a sibling difference \(x^{(i)}_q-x^{(i)}_p\) is
given as a count, \eqref{eq:exptrunc}.  In this section we give it a geometric
meaning.  To each node \(p\in V_i\) with \(i<L\) we attach a
\(k\)-dimensional Young diagram, \(k=L-i\), and the exponent is the number of
points in the intersection of the diagrams of \(p\) and \(q\).  For \(k=2\)
these are ordinary Young diagrams; for \(k=3\) they are three-dimensional
Young diagrams, also called plane partitions, familiar from the melting
crystal picture of topological strings \cite{ORV}; for \(k=4\) they are
four-dimensional Young diagrams, or solid partitions, which appear for example
in \cite{Nekrasov4}.  In the combinatorics literature they are also called
Ferrers diagrams, or finite order ideals of \(\N^k\).

\subsection{The Young diagram of a node}

\begin{definition}\label{def:Y}
Let \(p\in V_i\) with \(1\le i\le L-1\), and let \(k=L-i\).  The Young diagram
of \(p\) is the set of points of \(\N_{\ge1}^{\,k}\), i.e.\ tuples of
continuation data, for which \(p\) belongs to the sets
\eqref{eq:SLm1}--\eqref{eq:Si}:
\begin{equation}
\cY(p)=\Bigl\{(d_{i+1},\ldots,d_L)\in\N_{\ge1}^{\,k}:\
p\in S^{(i)}_{[d_{i+1},\ldots,d_L]}\Bigr\}.
\label{eq:Ydef}
\end{equation}
For a final node \(T\in V_L\) (\(k=0\)) we set
\(\cY(T)=\{()\}\subset\N_{\ge1}^{\,0}\), the set whose only element is the
empty tuple.
\end{definition}

In words, following the description (3.4) of \cite{KempRamgoolam2026} and
the paragraph after \eqref{eq:Si}: \((d_{i+1},\ldots,d_L)\) is a point of
\(\cY(p)\) exactly when \(p\) has at least \(d_{i+1}\) children that have at
least \(d_{i+2}\) children that have \ldots, up to nodes of layer \(L-1\) with
at least \(d_L\) children; other children of \(p\) are not required to satisfy
these conditions.

The set \(\cY(p)\) is finite, and if it contains a point it contains every
point which is no larger in any coordinate: lowering any \(d_j\) weakens the
condition ``at least \(d_j\) children''.  A finite subset of
\(\N_{\ge1}^{\,k}\) with this property is by definition a \(k\)-dimensional
Young diagram.  For \(k=1\) it is an interval \(\{1,\ldots,n\}\).

\paragraph{Points and cubes.}  Definition~\ref{def:Y} describes \(\cY(p)\) as a
set of points of the lattice \(\N_{\ge1}^{\,k}\).  In pictures, a point
\(\vec d=(d_1,\ldots,d_k)\) is drawn as the unit cube
\begin{equation}
\square(\vec d)=[d_1-1,d_1]\times\cdots\times[d_k-1,d_k]\subset\mathbb R^k ,
\label{eq:cube}
\end{equation}
whose corner farthest from the origin is \(\vec d\).  For \(k=1\) this is a
unit segment, for \(k=2\) a unit square, and for \(k=3\) a unit cube; we use
the word cube in every dimension.  The union of the cubes of a
\(k\)-dimensional Young diagram is a stack of cubes pushed into the corner of
the positive orthant: with each cube, all the cubes between it and the origin
are present.  The number of cubes is the number of points.  We define the
\emph{volume} \(|\cY|\) of a Young diagram \(\cY\) to be this number, and the
\emph{overlap volume} of two \(k\)-dimensional Young diagrams to be the volume
\(|\cY\cap\cY'|\) of their intersection, which is again a \(k\)-dimensional
Young diagram.

For example, in the graph \(D=(2,4,7)\) of Figure~\ref{fig:k2}, node
\(1\in V_1\) has children \(1,2\in V_2\), with \(2\) and \(1\) children.  A
tuple \((d_2,d_3)\) lies in \(\cY(1)\) when node \(1\) has at least \(d_2\)
children that have at least \(d_3\) children.  This gives the points
\((1,1)\), \((2,1)\), \((1,2)\), drawn as three cubes (here unit squares) with
far corners at these points, which form the Young diagram of the partition
\((2,1)\).  The row at \(d_3=1\) has \(2\) points, the number of children of
node \(1\) with at least one child, and the row at \(d_3=2\) has \(1\) point,
the number of children with at least two children.  This is the stacking description of
Section~\ref{sec:stacking}, which in general builds \(\cY(p)\) from the Young diagrams of the
children of \(p\): the first coordinate \(d_{i+1}\) is the stacking direction,
and above each point of the remaining coordinates there is a column whose
length is the number of children whose diagrams contain that point.  The
stacking direction is drawn horizontally in the two-dimensional pictures
(Figure~\ref{fig:k2}) and vertically in the three-dimensional ones.

\subsection{The exponent as an overlap volume}

\begin{proposition}\label{prop:overlap}
Let \(p<q\) be siblings in \(V_i\), \(1\le i\le L-1\).  Then
\begin{equation}
\Exp(i,p,q)=\bigl|\cY(p)\cap\cY(q)\bigr| .
\label{eq:higheroverlap}
\end{equation}
\end{proposition}

\begin{proof}
Whether a node of layer \(i\) belongs to \(S^{(i)}_{[d_{i+1},\ldots,d_L]}\)
depends only on its descendants.  The sibling-pair-truncated graph
\(G_{\mathrm{trunc}}(i,p,q)\) keeps \(p\), \(q\) and all their descendants, so
\(S^{(i)}_{[d_{i+1},\ldots,d_L]}(G_{\mathrm{trunc}}(i,p,q))\) consists of those
of \(p,q\) which lie in \(S^{(i)}_{[d_{i+1},\ldots,d_L]}\).  It has two
elements exactly when \((d_{i+1},\ldots,d_L)\in\cY(p)\cap\cY(q)\), and
\eqref{eq:exptrunc} counts these tuples.
\end{proof}

Since \(|\cY(p)\cap\cY(q)|=1\) for final siblings, in agreement with
\eqref{eq:expfinal}, the determinant formula of Theorem~\ref{thm:KR} reads
\begin{equation}
\det\cM
=\prod_{i=1}^{L}\ \prod_{v_{i-1}\in V_{i-1}}\ \prod_{\substack{p<q\\ p,q\in\Ch_i(v_{i-1})}}
\bigl(x^{(i)}_q-x^{(i)}_p\bigr)^{|\cY(p)\cap\cY(q)|}.
\label{eq:KRdetY}
\end{equation}

The overlap volumes count the factors of \eqref{eq:KRdetY} pair by pair.  They
can also be collected for each set of siblings.  For a set \(A\) of siblings in
\(V_i\), write
\begin{equation}
\operatorname{Vdm}_A\bigl(x^{(i)}\bigr)
=\prod_{\substack{p<q\\ p,q\in A}}\bigl(x^{(i)}_q-x^{(i)}_p\bigr),
\end{equation}
with \(\operatorname{Vdm}_A=1\) when \(A\) has at most one element.  The
admissible sets \(R_i(v_{i-1};d_{i+1},\ldots,d_L)\) are those of
\eqref{eq:R1}--\eqref{eq:RL}.  For \(i=1\),
\(R_1(v_0;d_2,\ldots,d_L)=R_1(d_2,\ldots,d_L)\), as in Section~\ref{sec:admchld}; for \(i=L\) the tuple of
continuation data is empty, and \(R_L(v_{L-1})=\Ch_L(v_{L-1})\).

\begin{corollary}[The determinant as a product of Vandermonde determinants]\label{cor:vdm}
\begin{equation}
\det\cM=\prod_{i=1}^{L}\ \prod_{v_{i-1}\in V_{i-1}}\
\prod_{(d_{i+1},\ldots,d_L)\in\N_{\ge1}^{\,L-i}}
\operatorname{Vdm}_{R_i(v_{i-1};d_{i+1},\ldots,d_L)}\bigl(x^{(i)}\bigr),
\label{eq:KRdetVdm}
\end{equation}
where only finitely many factors differ from \(1\).
\end{corollary}

\begin{proof}
Every factor on the right of \eqref{eq:KRdetVdm} is a product of differences
\(x^{(i)}_q-x^{(i)}_p\) of siblings \(p<q\) in \(\Ch_i(v_{i-1})\), since
\(R_i(v_{i-1};d_{i+1},\ldots,d_L)\subset\Ch_i(v_{i-1})\).  For \(i=L\) each
such difference occurs once, in \(\operatorname{Vdm}_{\Ch_L(v_{L-1})}\).  For
\(i<L\), by \eqref{eq:R1}, \eqref{eq:Ri} and \eqref{eq:Ydef}, both \(p\) and
\(q\) lie in \(R_i(v_{i-1};d_{i+1},\ldots,d_L)\) exactly when
\((d_{i+1},\ldots,d_L)\in\cY(p)\cap\cY(q)\), so the difference occurs
\(|\cY(p)\cap\cY(q)|\) times, a finite number.  The right side of
\eqref{eq:KRdetVdm} is therefore the right side of \eqref{eq:KRdetY}, which is
\(\det\cM\) by Theorem~\ref{thm:KR} and Proposition~\ref{prop:overlap}.
\end{proof}

\begin{remark}[The Vandermonde factors in the proof, and tensor products]\label{rem:vdm}
Corollary~\ref{cor:vdm} is also the form in which the proof of
Theorem~\ref{thm:newtondet} produces the determinant.  When the recursion has
reached a graph whose final layer is layer \(i\), labelled by the tuple
\((d_{i+1},\ldots,d_L)\) as in the paragraph ``Where the factors arise'' of
Section~\ref{sec:general} (for \(i=L\) this is \(\Gamma\) itself), the children
of \(v_{i-1}\) in that graph form the set \(R_i(v_{i-1};d_{i+1},\ldots,d_L)\),
by Lemma~\ref{lem:restrict} applied repeatedly.  In that graph,
\(v_{i-1}\) plays the role of the node \(q\in V_{L-1}\) in the proof of
Theorem~\ref{thm:newtondet}, and the tuple \((d_{i+1},\ldots,d_L)\) records the
values of \(d\) chosen at the successive steps of the recursion,
\(\Gamma\to\Gamma_{[d_L]}\to\bigl(\Gamma_{[d_L]}\bigr)_{[d_{L-1}]}\to\cdots\).
For \(i\ge2\), \(\operatorname{Vdm}_{R_i(v_{i-1};d_{i+1},\ldots,d_L)}(x^{(i)})\)
is then the determinant \eqref{eq:fact2} of the matrix \(W_{v_{i-1}}\) for
these children in the factorisation \eqref{eq:NPfactor} of the Newton
evaluation matrix of that graph; for \(i=1\) it is the determinant of the
Newton evaluation matrix of that graph, which has one layer.  Every factor in
\eqref{eq:KRdetVdm} is therefore the determinant of a one-layer
Newton--Vandermonde matrix.

The simplest special case is a two-fold tensor product.  Let \(L=2\), with
\(m=D_1\) nodes in \(V_1\), suppose that every node of \(V_1\) has the same
number \(n\) of children, and that for each \(k\) the \(k\)-th children of all
the nodes of \(V_1\) have the same eigenvalue \(y_k\) of \(C_2\).  The joint
eigenvalues then form the product set
\(\{x^{(1)}_1,\ldots,x^{(1)}_m\}\times\{y_1,\ldots,y_n\}\), and \(\cA_2\) is
the tensor product of the algebras generated by \(C_1\) and by \(C_2\).  The
graph monomials are \(C_1^aC_2^b\), \(0\le a<m\), \(0\le b<n\), and, with
rows in path order and columns in the order of the bijection
\eqref{eq:bijection}, \(\cM\) is the Kronecker product of the Vandermonde
matrices of the two sets of eigenvalues.  In \eqref{eq:KRdetVdm} every
\(R_1(d_2)\), \(1\le d_2\le n\), is all of \(V_1\), and every
\(R_2(v_1)=\Ch_2(v_1)\) carries the eigenvalues \(y_1,\ldots,y_n\), so that
\begin{equation}
\det\cM=\Bigl(\operatorname{Vdm}_{V_1}\bigl(x^{(1)}\bigr)\Bigr)^{n}
\Bigl(\prod_{1\le k<l\le n}(y_l-y_k)\Bigr)^{m},
\end{equation}
the determinant of a Kronecker product.  Degeneracy graphs are much more
general: the numbers of children, and the eigenvalues of the children, vary
from node to node.  The evaluation matrix \(\cM\) is then not a Kronecker
product, and the Newton evaluation matrix \(\cN\) is in general not
triangular in any ordering (Remark~\ref{rem:L3}), but the determinant keeps
the form \eqref{eq:KRdetVdm} of a product of Vandermonde determinants.
\end{remark}

\subsection{Reading the diagram from the graph: stacking}
\label{sec:stacking}

To draw \(\cY(p)\), fix all coordinates except the first.  By the description following
Definition~\ref{def:Y}, \((d_{i+1},d_{i+2},\ldots,d_L)\in\cY(p)\) exactly when \(p\) has at
least \(d_{i+1}\) children \(u\) with \((d_{i+2},\ldots,d_L)\in\cY(u)\).  So
the allowed values of \(d_{i+1}\) are
\begin{equation}
\bigl\{d_{i+1}:\ p\in S^{(i)}_{[d_{i+1},\ldots,d_L]}\bigr\}
=\bigl\{1,\ldots,|R_{i+1}(p;d_{i+2},\ldots,d_L)|\bigr\},
\label{eq:survival-values}
\end{equation}
and the height of \(\cY(p)\) above the point \(c=(d_{i+2},\ldots,d_L)\) is
\begin{equation}
h_p(c)=|R_{i+1}(p;c)|=\#\bigl\{u\in\Ch_{i+1}(p):\ c\in\cY(u)\bigr\}.
\label{eq:stacking}
\end{equation}
(For \(i=L-1\), \(c\) is the empty tuple, every child counts, and
\(h_p=|\Ch_L(p)|\).)  Thus
\begin{equation}
\cY(p)=\bigl\{(t,c):\ 1\le t\le h_p(c)\bigr\}:
\end{equation}
above each point \(c\) of the \((k-1)\)-dimensional floor stands a column of
\(h_p(c)\) cubes, and \(h_p\) is weakly decreasing in each coordinate of \(c\),
since the diagrams of the children are.  In words: \(\cY(p)\) is obtained by stacking the
\((k-1)\)-dimensional Young diagrams of the children of \(p\).  Reading the stack in horizontal layers,
the \(t\)-th layer from the bottom consists of the points covered by at least
\(t\) children; each layer is a \((k-1)\)-dimensional Young diagram, and the
layers shrink going up.  The intersection of two diagrams has height
\(\min(h_p(c),h_q(c))\) above \(c\), so
\begin{equation}
\bigl|\cY(p)\cap\cY(q)\bigr|=\sum_{c}\min\bigl(h_p(c),h_q(c)\bigr).
\label{eq:profileoverlap}
\end{equation}

\subsection{Low depths}
\label{sec:lowdepths}

\paragraph{One layer before the final layer (\(k=1\)).}  For \(p\in V_{L-1}\)
with \(n_p=|\Ch_L(p)|\) children, \(d_L\in\cY(p)\) exactly when \(p\) has at
least \(d_L\) children, so \(\cY(p)=\{1,\ldots,n_p\}\) and
\begin{equation}
\Exp(L-1,p,q)=\min(n_p,n_q).
\label{eq:minchildren}
\end{equation}

\paragraph{Two layers before the final layer (\(k=2\)).}  For \(p\in V_{L-2}\)
write the coordinates as \((s,r)=(d_{L-1},d_L)\).  Then \((s,r)\in\cY(p)\)
exactly when at least \(s\) of the children of \(p\) have at least \(r\)
children each.  By \eqref{eq:stacking} the length of the row at height \(r\) is
the number of children with at least \(r\) children,
\begin{equation}
\lambda_r(p)=\#\bigl\{u\in\Ch_{L-1}(p):\ n_u\ge r\bigr\},
\qquad \lambda_1(p)\ge\lambda_2(p)\ge\cdots,
\label{eq:partitionlambda}
\end{equation}
and
\begin{equation}
\cY(p)=\bigl\{(s,r):\ 1\le s\le\lambda_r(p)\bigr\}
\label{eq:youngprofile}
\end{equation}
is the Young diagram of the partition \(\lambda(p)\).  For siblings,
\begin{equation}
\Exp(L-2,p,q)=\sum_{r\ge1}\min\bigl(\lambda_r(p),\lambda_r(q)\bigr).
\label{eq:youngoverlap}
\end{equation}

\paragraph{Three layers before the final layer (\(k=3\)).}  For
\(p\in V_{L-3}\) write the coordinates as \((t,s,r)=(d_{L-2},d_{L-1},d_L)\).
Then \((t,s,r)\in\cY(p)\) exactly when at least \(t\) of the children of \(p\)
each have at least \(s\) children with at least \(r\) children each.  By
\eqref{eq:stacking} the height above the point \((s,r)\) of the floor is the number of
children whose Young diagram contains \((s,r)\),
\begin{equation}
\Lambda_{s,r}(p)=\#\bigl\{u\in\Ch_{L-2}(p):\ \lambda_r(u)\ge s\bigr\},
\label{eq:planepartition}
\end{equation}
which is weakly decreasing in \(s\) and in \(r\), and
\begin{equation}
\cY(p)=\bigl\{(t,s,r):\ 1\le t\le\Lambda_{s,r}(p)\bigr\}
\label{eq:planeprofile}
\end{equation}
is the three-dimensional Young diagram with height array \(\Lambda(p)\).  For
siblings,
\begin{equation}
\Exp(L-3,p,q)=\sum_{r,s\ge1}\min\bigl(\Lambda_{s,r}(p),\Lambda_{s,r}(q)\bigr).
\label{eq:planeoverlap}
\end{equation}

At \(k=4\) the same stacking of three-dimensional diagrams gives
four-dimensional Young diagrams, and so on.

\subsection{Examples}
\label{sec:youngexamples}

\paragraph{Intervals.}  For \(D=(4,10)\), \(p_1=[2,2,3,3]\), the nodes of
\(V_1\) have \(2,2,3,3\) children, so their Young diagrams are
\(\{1,2\},\{1,2\},\{1,2,3\},\{1,2,3\}\).  The exponent of
\(x^{(1)}_q-x^{(1)}_p\) is \(2\) for every pair except \((p,q)=(3,4)\), where it
is \(3\), in agreement with the determinant computed in
\cite{KempRamgoolam2026} and derived in Appendix~\ref{app:D410}.

\paragraph{Two-dimensional Young diagrams.}  For \(D=(2,4,7)\),
\(p_1=[2,2]\), \(c_2=[2,1,2,2]\) (Figure~\ref{fig:k2}), node \(1\in V_1\) has
children with \(2\) and \(1\) children, so \(\lambda(1)=(2,1)\), and node
\(2\in V_1\) has children with \(2\) and \(2\) children, so
\(\lambda(2)=(2,2)\).  The intersection has \(3\) cubes.  The nodes of \(V_2\)
have intervals \(\{1,2\},\{1\},\{1,2\},\{1,2\}\).  Hence
\begin{equation}
\det\cM=
\bigl(x^{(1)}_2-x^{(1)}_1\bigr)^3
\bigl(x^{(2)}_2-x^{(2)}_1\bigr)
\bigl(x^{(2)}_4-x^{(2)}_3\bigr)^2
\bigl(x^{(3)}_2-x^{(3)}_1\bigr)
\bigl(x^{(3)}_5-x^{(3)}_4\bigr)
\bigl(x^{(3)}_7-x^{(3)}_6\bigr).
\end{equation}

\begin{figure}[ht]
\centering
\begin{tikzpicture}
\node[rt] (r) at (0,1.6) {};
\node[gn] (a1) at (1.3,0.5) {$1$};
\node[gn] (a2) at (1.3,2.7) {$2$};
\node[gn,minimum size=5mm] (b1) at (2.6,0) {$1$};
\node[gn,minimum size=5mm] (b2) at (2.6,1.1) {$2$};
\node[gn,minimum size=5mm] (b3) at (2.6,2.2) {$3$};
\node[gn,minimum size=5mm] (b4) at (2.6,3.3) {$4$};
\draw (r)--(a1) (r)--(a2) (a1)--(b1) (a1)--(b2) (a2)--(b3) (a2)--(b4);
\foreach \y/\n in {-0.25/1,0.25/2,1.1/3,1.95/4,2.45/5,3.05/6,3.55/7}
  \node[circle,draw,inner sep=0pt,minimum size=2.5mm] (t\n) at (3.8,\y) {};
\draw (b1)--(t1) (b1)--(t2) (b2)--(t3) (b3)--(t4) (b3)--(t5) (b4)--(t6) (b4)--(t7);
\foreach \x/\l in {1.3/V_1,2.6/V_2,3.8/V_3} \node[font=\small] at (\x,4.2) {$\l$};
\draw[->,black!60] (5.20,0.60) -- (7.20,0.60) node[right,font=\scriptsize]{$d_2$};
\draw[->,black!60] (5.20,0.60) -- (5.20,2.60) node[above,font=\scriptsize]{$d_3$};
\filldraw[fill=YAtop,draw=black!80] (5.20,0.60) rectangle (5.90,1.30);
\filldraw[fill=YAtop,draw=black!80] (5.20,1.30) rectangle (5.90,2.00);
\filldraw[fill=YAtop,draw=black!80] (5.90,0.60) rectangle (6.60,1.30);
\node[font=\small] at (6.00,0.15) {$\cY(1)$, $\lambda=(2,1)$};
\draw[->,black!60] (8.30,0.60) -- (10.30,0.60) node[right,font=\scriptsize]{$d_2$};
\draw[->,black!60] (8.30,0.60) -- (8.30,2.60) node[above,font=\scriptsize]{$d_3$};
\filldraw[fill=YBtop,draw=black!80] (8.30,0.60) rectangle (9.00,1.30);
\filldraw[fill=YBtop,draw=black!80] (8.30,1.30) rectangle (9.00,2.00);
\filldraw[fill=YBtop,draw=black!80] (9.00,0.60) rectangle (9.70,1.30);
\filldraw[fill=YBtop,draw=black!80] (9.00,1.30) rectangle (9.70,2.00);
\node[font=\small] at (9.10,0.15) {$\cY(2)$, $\lambda=(2,2)$};
\draw[->,black!60] (11.40,0.60) -- (13.40,0.60) node[right,font=\scriptsize]{$d_2$};
\draw[->,black!60] (11.40,0.60) -- (11.40,2.60) node[above,font=\scriptsize]{$d_3$};
\filldraw[fill=YCtop,draw=black!80] (11.40,0.60) rectangle (12.10,1.30);
\filldraw[fill=YCtop,draw=black!80] (11.40,1.30) rectangle (12.10,2.00);
\filldraw[fill=YCtop,draw=black!80] (12.10,0.60) rectangle (12.80,1.30);
\node[font=\small] at (12.20,0.15) {$\cY(1)\cap\cY(2)$};
\end{tikzpicture}
\caption{Left: the graph with \(D=(2,4,7)\), \(p_1=[2,2]\),
\(c_2=[2,1,2,2]\).  Right: the Young diagrams of the two nodes of \(V_1\), with
cubes corresponding to the points \((d_2,d_3)\).  Node \(1\) has children with \(2\) and \(1\)
children, so \(\lambda(1)=(2,1)\); node \(2\) has children with \(2\) and \(2\)
children, so \(\lambda(2)=(2,2)\).  The intersection has \(3\) cubes, the
exponent of \(x^{(1)}_2-x^{(1)}_1\).}
\label{fig:k2}
\end{figure}
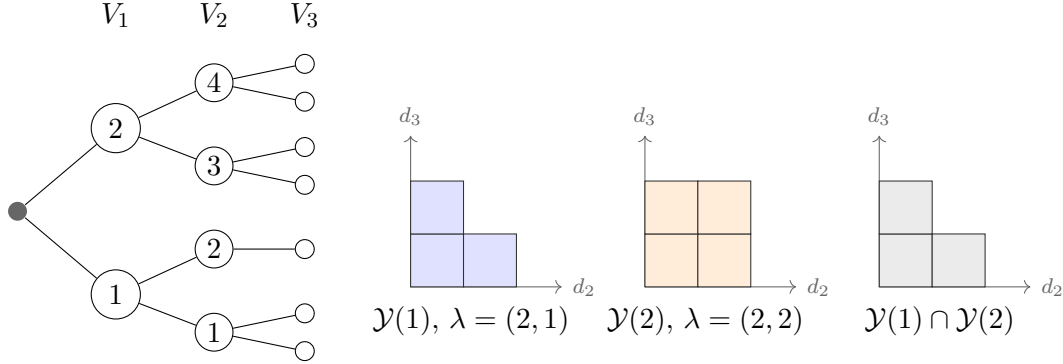

\paragraph{Three-dimensional Young diagrams.}  For \(D=(2,4,7,12)\),
\(p_1=[2,2]\), \(c_2=[2,1,2,2]\), \(c_3=[2,3,1,2,2,1,1]\), the nodes of
\(V_3\) have \(2,3,1,2,2,1,1\) children.  The nodes of \(V_2\) therefore have
\(\lambda=(2,2,1)\), \((1)\), \((2,2)\), \((2)\).  Stacking these
(Figure~\ref{fig:k3stack}), node \(1\in V_1\) has height array
\(\begin{smallmatrix}2&1&1\\1&1\end{smallmatrix}\) and node \(2\in V_1\) has
\(\begin{smallmatrix}2&1\\2&1\end{smallmatrix}\), with rows \(d_3=1,2\) and
columns \(d_4=1,2,3\).  Their intersection has
\(\begin{smallmatrix}2&1\\1&1\end{smallmatrix}\), with \(5\) cubes
(Figure~\ref{fig:k3overlap}).  Table~\ref{tab:D24712} lists all sibling pairs
of this graph.  In every case the overlap volume agrees with the exponent
found by factorising the determinant directly.

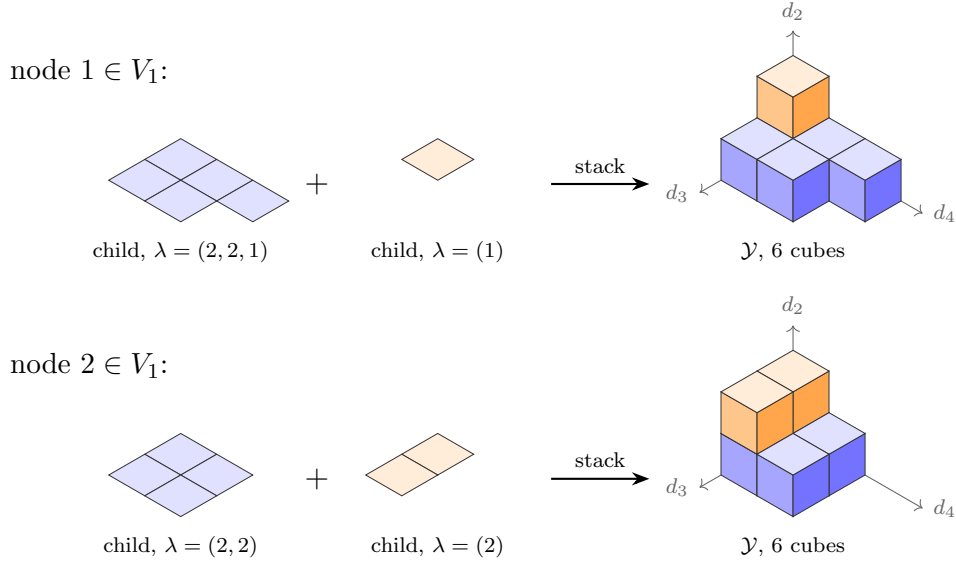
\begin{figure}[ht]
\centering
\begin{tikzpicture}
\node[font=\small,anchor=west] at (-1.2,0.90) {node $1\in V_1$:};
\filldraw[fill=YAtop,draw=black!80,line width=0.3pt] (1.200,0.000) -- (1.676,-0.275) -- (1.200,-0.550) -- (0.724,-0.275) -- cycle;
\filldraw[fill=YAtop,draw=black!80,line width=0.3pt] (1.676,-0.275) -- (2.153,-0.550) -- (1.676,-0.825) -- (1.200,-0.550) -- cycle;
\filldraw[fill=YAtop,draw=black!80,line width=0.3pt] (0.724,-0.275) -- (1.200,-0.550) -- (0.724,-0.825) -- (0.247,-0.550) -- cycle;
\filldraw[fill=YAtop,draw=black!80,line width=0.3pt] (2.153,-0.550) -- (2.629,-0.825) -- (2.153,-1.100) -- (1.676,-0.825) -- cycle;
\filldraw[fill=YAtop,draw=black!80,line width=0.3pt] (1.200,-0.550) -- (1.676,-0.825) -- (1.200,-1.100) -- (0.724,-0.825) -- cycle;
\node[font=\scriptsize] at (1.2,-1.50) {child, $\lambda=(2,2,1)$};
\node at (3.0,-0.60) {$+$};
\filldraw[fill=YBtop,draw=black!80,line width=0.3pt] (4.600,0.000) -- (5.076,-0.275) -- (4.600,-0.550) -- (4.124,-0.275) -- cycle;
\node[font=\scriptsize] at (4.6,-1.50) {child, $\lambda=(1)$};
\draw[-{Stealth},thick] (6.1,-0.60) -- node[above,font=\scriptsize]{stack} (7.4,-0.60);
\draw[->,black!60] (9.300,0.000) -- (8.062,-0.715) node[left,font=\scriptsize]{$d_3$};
\draw[->,black!60] (9.300,0.000) -- (11.015,-0.990) node[right,font=\scriptsize]{$d_4$};
\draw[->,black!60] (9.300,0.000) -- (9.300,1.430) node[above,font=\scriptsize]{$d_2$};
\filldraw[fill=YAleft,draw=black!80,line width=0.3pt] (8.824,-0.275) -- (9.300,-0.550) -- (9.300,0.000) -- (8.824,0.275) -- cycle;
\filldraw[fill=YAright,draw=black!80,line width=0.3pt] (9.776,-0.275) -- (9.300,-0.550) -- (9.300,0.000) -- (9.776,0.275) -- cycle;
\filldraw[fill=YAtop,draw=black!80,line width=0.3pt] (9.300,0.550) -- (9.776,0.275) -- (9.300,0.000) -- (8.824,0.275) -- cycle;
\filldraw[fill=YBleft,draw=black!80,line width=0.3pt] (8.824,0.275) -- (9.300,0.000) -- (9.300,0.550) -- (8.824,0.825) -- cycle;
\filldraw[fill=YBright,draw=black!80,line width=0.3pt] (9.776,0.275) -- (9.300,0.000) -- (9.300,0.550) -- (9.776,0.825) -- cycle;
\filldraw[fill=YBtop,draw=black!80,line width=0.3pt] (9.300,1.100) -- (9.776,0.825) -- (9.300,0.550) -- (8.824,0.825) -- cycle;
\filldraw[fill=YAleft,draw=black!80,line width=0.3pt] (9.300,-0.550) -- (9.776,-0.825) -- (9.776,-0.275) -- (9.300,0.000) -- cycle;
\filldraw[fill=YAright,draw=black!80,line width=0.3pt] (10.253,-0.550) -- (9.776,-0.825) -- (9.776,-0.275) -- (10.253,0.000) -- cycle;
\filldraw[fill=YAtop,draw=black!80,line width=0.3pt] (9.776,0.275) -- (10.253,0.000) -- (9.776,-0.275) -- (9.300,0.000) -- cycle;
\filldraw[fill=YAleft,draw=black!80,line width=0.3pt] (8.347,-0.550) -- (8.824,-0.825) -- (8.824,-0.275) -- (8.347,0.000) -- cycle;
\filldraw[fill=YAright,draw=black!80,line width=0.3pt] (9.300,-0.550) -- (8.824,-0.825) -- (8.824,-0.275) -- (9.300,0.000) -- cycle;
\filldraw[fill=YAtop,draw=black!80,line width=0.3pt] (8.824,0.275) -- (9.300,0.000) -- (8.824,-0.275) -- (8.347,0.000) -- cycle;
\filldraw[fill=YAleft,draw=black!80,line width=0.3pt] (9.776,-0.825) -- (10.253,-1.100) -- (10.253,-0.550) -- (9.776,-0.275) -- cycle;
\filldraw[fill=YAright,draw=black!80,line width=0.3pt] (10.729,-0.825) -- (10.253,-1.100) -- (10.253,-0.550) -- (10.729,-0.275) -- cycle;
\filldraw[fill=YAtop,draw=black!80,line width=0.3pt] (10.253,0.000) -- (10.729,-0.275) -- (10.253,-0.550) -- (9.776,-0.275) -- cycle;
\filldraw[fill=YAleft,draw=black!80,line width=0.3pt] (8.824,-0.825) -- (9.300,-1.100) -- (9.300,-0.550) -- (8.824,-0.275) -- cycle;
\filldraw[fill=YAright,draw=black!80,line width=0.3pt] (9.776,-0.825) -- (9.300,-1.100) -- (9.300,-0.550) -- (9.776,-0.275) -- cycle;
\filldraw[fill=YAtop,draw=black!80,line width=0.3pt] (9.300,0.000) -- (9.776,-0.275) -- (9.300,-0.550) -- (8.824,-0.275) -- cycle;
\node[font=\scriptsize] at (9.3,-1.50) {$\cY$, 6 cubes};
\node[font=\small,anchor=west] at (-1.2,-3.00) {node $2\in V_1$:};
\filldraw[fill=YAtop,draw=black!80,line width=0.3pt] (1.200,-3.900) -- (1.676,-4.175) -- (1.200,-4.450) -- (0.724,-4.175) -- cycle;
\filldraw[fill=YAtop,draw=black!80,line width=0.3pt] (0.724,-4.175) -- (1.200,-4.450) -- (0.724,-4.725) -- (0.247,-4.450) -- cycle;
\filldraw[fill=YAtop,draw=black!80,line width=0.3pt] (1.676,-4.175) -- (2.153,-4.450) -- (1.676,-4.725) -- (1.200,-4.450) -- cycle;
\filldraw[fill=YAtop,draw=black!80,line width=0.3pt] (1.200,-4.450) -- (1.676,-4.725) -- (1.200,-5.000) -- (0.724,-4.725) -- cycle;
\node[font=\scriptsize] at (1.2,-5.40) {child, $\lambda=(2,2)$};
\node at (3.0,-4.50) {$+$};
\filldraw[fill=YBtop,draw=black!80,line width=0.3pt] (4.600,-3.900) -- (5.076,-4.175) -- (4.600,-4.450) -- (4.124,-4.175) -- cycle;
\filldraw[fill=YBtop,draw=black!80,line width=0.3pt] (4.124,-4.175) -- (4.600,-4.450) -- (4.124,-4.725) -- (3.647,-4.450) -- cycle;
\node[font=\scriptsize] at (4.6,-5.40) {child, $\lambda=(2)$};
\draw[-{Stealth},thick] (6.1,-4.50) -- node[above,font=\scriptsize]{stack} (7.4,-4.50);
\draw[->,black!60] (9.300,-3.900) -- (8.062,-4.615) node[left,font=\scriptsize]{$d_3$};
\draw[->,black!60] (9.300,-3.900) -- (11.015,-4.890) node[right,font=\scriptsize]{$d_4$};
\draw[->,black!60] (9.300,-3.900) -- (9.300,-2.470) node[above,font=\scriptsize]{$d_2$};
\filldraw[fill=YAleft,draw=black!80,line width=0.3pt] (8.824,-4.175) -- (9.300,-4.450) -- (9.300,-3.900) -- (8.824,-3.625) -- cycle;
\filldraw[fill=YAright,draw=black!80,line width=0.3pt] (9.776,-4.175) -- (9.300,-4.450) -- (9.300,-3.900) -- (9.776,-3.625) -- cycle;
\filldraw[fill=YAtop,draw=black!80,line width=0.3pt] (9.300,-3.350) -- (9.776,-3.625) -- (9.300,-3.900) -- (8.824,-3.625) -- cycle;
\filldraw[fill=YBleft,draw=black!80,line width=0.3pt] (8.824,-3.625) -- (9.300,-3.900) -- (9.300,-3.350) -- (8.824,-3.075) -- cycle;
\filldraw[fill=YBright,draw=black!80,line width=0.3pt] (9.776,-3.625) -- (9.300,-3.900) -- (9.300,-3.350) -- (9.776,-3.075) -- cycle;
\filldraw[fill=YBtop,draw=black!80,line width=0.3pt] (9.300,-2.800) -- (9.776,-3.075) -- (9.300,-3.350) -- (8.824,-3.075) -- cycle;
\filldraw[fill=YAleft,draw=black!80,line width=0.3pt] (9.300,-4.450) -- (9.776,-4.725) -- (9.776,-4.175) -- (9.300,-3.900) -- cycle;
\filldraw[fill=YAright,draw=black!80,line width=0.3pt] (10.253,-4.450) -- (9.776,-4.725) -- (9.776,-4.175) -- (10.253,-3.900) -- cycle;
\filldraw[fill=YAtop,draw=black!80,line width=0.3pt] (9.776,-3.625) -- (10.253,-3.900) -- (9.776,-4.175) -- (9.300,-3.900) -- cycle;
\filldraw[fill=YAleft,draw=black!80,line width=0.3pt] (8.347,-4.450) -- (8.824,-4.725) -- (8.824,-4.175) -- (8.347,-3.900) -- cycle;
\filldraw[fill=YAright,draw=black!80,line width=0.3pt] (9.300,-4.450) -- (8.824,-4.725) -- (8.824,-4.175) -- (9.300,-3.900) -- cycle;
\filldraw[fill=YAtop,draw=black!80,line width=0.3pt] (8.824,-3.625) -- (9.300,-3.900) -- (8.824,-4.175) -- (8.347,-3.900) -- cycle;
\filldraw[fill=YBleft,draw=black!80,line width=0.3pt] (8.347,-3.900) -- (8.824,-4.175) -- (8.824,-3.625) -- (8.347,-3.350) -- cycle;
\filldraw[fill=YBright,draw=black!80,line width=0.3pt] (9.300,-3.900) -- (8.824,-4.175) -- (8.824,-3.625) -- (9.300,-3.350) -- cycle;
\filldraw[fill=YBtop,draw=black!80,line width=0.3pt] (8.824,-3.075) -- (9.300,-3.350) -- (8.824,-3.625) -- (8.347,-3.350) -- cycle;
\filldraw[fill=YAleft,draw=black!80,line width=0.3pt] (8.824,-4.725) -- (9.300,-5.000) -- (9.300,-4.450) -- (8.824,-4.175) -- cycle;
\filldraw[fill=YAright,draw=black!80,line width=0.3pt] (9.776,-4.725) -- (9.300,-5.000) -- (9.300,-4.450) -- (9.776,-4.175) -- cycle;
\filldraw[fill=YAtop,draw=black!80,line width=0.3pt] (9.300,-3.900) -- (9.776,-4.175) -- (9.300,-4.450) -- (8.824,-4.175) -- cycle;
\node[font=\scriptsize] at (9.3,-5.40) {$\cY$, 6 cubes};
\end{tikzpicture}
\caption{The three-dimensional Young diagrams of the two nodes of \(V_1\) for
\(D=(2,4,7,12)\), \(p_1=[2,2]\), \(c_2=[2,1,2,2]\), \(c_3=[2,3,1,2,2,1,1]\),
built by stacking the two-dimensional Young diagrams of their children
\eqref{eq:stacking}.  The floor coordinates are \(d_3\) and \(d_4\), the
height is \(d_2\).  The bottom layer (blue) consists of the floor points covered
by at least one child, the second layer (orange) of those covered by two.}
\label{fig:k3stack}
\end{figure}

\begin{figure}[ht]
\centering
\begin{tikzpicture}
\draw[->,black!60] (0.000,0.000) -- (-1.351,-0.780) node[left,font=\scriptsize]{$d_3$};
\draw[->,black!60] (0.000,0.000) -- (1.871,-1.080) node[right,font=\scriptsize]{$d_4$};
\draw[->,black!60] (0.000,0.000) -- (0.000,1.560) node[above,font=\scriptsize]{$d_2$};
\filldraw[fill=YAleft,draw=black!80,line width=0.3pt] (-0.520,-0.300) -- (0.000,-0.600) -- (0.000,0.000) -- (-0.520,0.300) -- cycle;
\filldraw[fill=YAright,draw=black!80,line width=0.3pt] (0.520,-0.300) -- (0.000,-0.600) -- (0.000,0.000) -- (0.520,0.300) -- cycle;
\filldraw[fill=YAtop,draw=black!80,line width=0.3pt] (0.000,0.600) -- (0.520,0.300) -- (0.000,0.000) -- (-0.520,0.300) -- cycle;
\filldraw[fill=YAleft,draw=black!80,line width=0.3pt] (-0.520,0.300) -- (0.000,0.000) -- (0.000,0.600) -- (-0.520,0.900) -- cycle;
\filldraw[fill=YAright,draw=black!80,line width=0.3pt] (0.520,0.300) -- (0.000,0.000) -- (0.000,0.600) -- (0.520,0.900) -- cycle;
\filldraw[fill=YAtop,draw=black!80,line width=0.3pt] (0.000,1.200) -- (0.520,0.900) -- (0.000,0.600) -- (-0.520,0.900) -- cycle;
\filldraw[fill=YAleft,draw=black!80,line width=0.3pt] (0.000,-0.600) -- (0.520,-0.900) -- (0.520,-0.300) -- (0.000,0.000) -- cycle;
\filldraw[fill=YAright,draw=black!80,line width=0.3pt] (1.039,-0.600) -- (0.520,-0.900) -- (0.520,-0.300) -- (1.039,0.000) -- cycle;
\filldraw[fill=YAtop,draw=black!80,line width=0.3pt] (0.520,0.300) -- (1.039,0.000) -- (0.520,-0.300) -- (0.000,0.000) -- cycle;
\filldraw[fill=YAleft,draw=black!80,line width=0.3pt] (-1.039,-0.600) -- (-0.520,-0.900) -- (-0.520,-0.300) -- (-1.039,0.000) -- cycle;
\filldraw[fill=YAright,draw=black!80,line width=0.3pt] (0.000,-0.600) -- (-0.520,-0.900) -- (-0.520,-0.300) -- (0.000,0.000) -- cycle;
\filldraw[fill=YAtop,draw=black!80,line width=0.3pt] (-0.520,0.300) -- (0.000,0.000) -- (-0.520,-0.300) -- (-1.039,0.000) -- cycle;
\filldraw[fill=YAleft,draw=black!80,line width=0.3pt] (0.520,-0.900) -- (1.039,-1.200) -- (1.039,-0.600) -- (0.520,-0.300) -- cycle;
\filldraw[fill=YAright,draw=black!80,line width=0.3pt] (1.559,-0.900) -- (1.039,-1.200) -- (1.039,-0.600) -- (1.559,-0.300) -- cycle;
\filldraw[fill=YAtop,draw=black!80,line width=0.3pt] (1.039,0.000) -- (1.559,-0.300) -- (1.039,-0.600) -- (0.520,-0.300) -- cycle;
\filldraw[fill=YAleft,draw=black!80,line width=0.3pt] (-0.520,-0.900) -- (0.000,-1.200) -- (0.000,-0.600) -- (-0.520,-0.300) -- cycle;
\filldraw[fill=YAright,draw=black!80,line width=0.3pt] (0.520,-0.900) -- (0.000,-1.200) -- (0.000,-0.600) -- (0.520,-0.300) -- cycle;
\filldraw[fill=YAtop,draw=black!80,line width=0.3pt] (0.000,0.000) -- (0.520,-0.300) -- (0.000,-0.600) -- (-0.520,-0.300) -- cycle;
\node[font=\small] at (0,-1.7) {$\cY(1)$: $\begin{smallmatrix}2&1&1\\1&1\end{smallmatrix}$};
\draw[->,black!60] (4.600,0.000) -- (3.249,-0.780) node[left,font=\scriptsize]{$d_3$};
\draw[->,black!60] (4.600,0.000) -- (6.471,-1.080) node[right,font=\scriptsize]{$d_4$};
\draw[->,black!60] (4.600,0.000) -- (4.600,1.560) node[above,font=\scriptsize]{$d_2$};
\filldraw[fill=YBleft,draw=black!80,line width=0.3pt] (4.080,-0.300) -- (4.600,-0.600) -- (4.600,0.000) -- (4.080,0.300) -- cycle;
\filldraw[fill=YBright,draw=black!80,line width=0.3pt] (5.120,-0.300) -- (4.600,-0.600) -- (4.600,0.000) -- (5.120,0.300) -- cycle;
\filldraw[fill=YBtop,draw=black!80,line width=0.3pt] (4.600,0.600) -- (5.120,0.300) -- (4.600,0.000) -- (4.080,0.300) -- cycle;
\filldraw[fill=YBleft,draw=black!80,line width=0.3pt] (4.080,0.300) -- (4.600,0.000) -- (4.600,0.600) -- (4.080,0.900) -- cycle;
\filldraw[fill=YBright,draw=black!80,line width=0.3pt] (5.120,0.300) -- (4.600,0.000) -- (4.600,0.600) -- (5.120,0.900) -- cycle;
\filldraw[fill=YBtop,draw=black!80,line width=0.3pt] (4.600,1.200) -- (5.120,0.900) -- (4.600,0.600) -- (4.080,0.900) -- cycle;
\filldraw[fill=YBleft,draw=black!80,line width=0.3pt] (4.600,-0.600) -- (5.120,-0.900) -- (5.120,-0.300) -- (4.600,0.000) -- cycle;
\filldraw[fill=YBright,draw=black!80,line width=0.3pt] (5.639,-0.600) -- (5.120,-0.900) -- (5.120,-0.300) -- (5.639,0.000) -- cycle;
\filldraw[fill=YBtop,draw=black!80,line width=0.3pt] (5.120,0.300) -- (5.639,0.000) -- (5.120,-0.300) -- (4.600,0.000) -- cycle;
\filldraw[fill=YBleft,draw=black!80,line width=0.3pt] (3.561,-0.600) -- (4.080,-0.900) -- (4.080,-0.300) -- (3.561,0.000) -- cycle;
\filldraw[fill=YBright,draw=black!80,line width=0.3pt] (4.600,-0.600) -- (4.080,-0.900) -- (4.080,-0.300) -- (4.600,0.000) -- cycle;
\filldraw[fill=YBtop,draw=black!80,line width=0.3pt] (4.080,0.300) -- (4.600,0.000) -- (4.080,-0.300) -- (3.561,0.000) -- cycle;
\filldraw[fill=YBleft,draw=black!80,line width=0.3pt] (3.561,0.000) -- (4.080,-0.300) -- (4.080,0.300) -- (3.561,0.600) -- cycle;
\filldraw[fill=YBright,draw=black!80,line width=0.3pt] (4.600,0.000) -- (4.080,-0.300) -- (4.080,0.300) -- (4.600,0.600) -- cycle;
\filldraw[fill=YBtop,draw=black!80,line width=0.3pt] (4.080,0.900) -- (4.600,0.600) -- (4.080,0.300) -- (3.561,0.600) -- cycle;
\filldraw[fill=YBleft,draw=black!80,line width=0.3pt] (4.080,-0.900) -- (4.600,-1.200) -- (4.600,-0.600) -- (4.080,-0.300) -- cycle;
\filldraw[fill=YBright,draw=black!80,line width=0.3pt] (5.120,-0.900) -- (4.600,-1.200) -- (4.600,-0.600) -- (5.120,-0.300) -- cycle;
\filldraw[fill=YBtop,draw=black!80,line width=0.3pt] (4.600,0.000) -- (5.120,-0.300) -- (4.600,-0.600) -- (4.080,-0.300) -- cycle;
\node[font=\small] at (4.6,-1.7) {$\cY(2)$: $\begin{smallmatrix}2&1\\2&1\end{smallmatrix}$};
\draw[->,black!60] (9.200,0.000) -- (7.849,-0.780) node[left,font=\scriptsize]{$d_3$};
\draw[->,black!60] (9.200,0.000) -- (11.071,-1.080) node[right,font=\scriptsize]{$d_4$};
\draw[->,black!60] (9.200,0.000) -- (9.200,1.560) node[above,font=\scriptsize]{$d_2$};
\filldraw[fill=YCleft,draw=black!80,line width=0.3pt] (8.680,-0.300) -- (9.200,-0.600) -- (9.200,0.000) -- (8.680,0.300) -- cycle;
\filldraw[fill=YCright,draw=black!80,line width=0.3pt] (9.720,-0.300) -- (9.200,-0.600) -- (9.200,0.000) -- (9.720,0.300) -- cycle;
\filldraw[fill=YCtop,draw=black!80,line width=0.3pt] (9.200,0.600) -- (9.720,0.300) -- (9.200,0.000) -- (8.680,0.300) -- cycle;
\filldraw[fill=YCleft,draw=black!80,line width=0.3pt] (8.680,0.300) -- (9.200,0.000) -- (9.200,0.600) -- (8.680,0.900) -- cycle;
\filldraw[fill=YCright,draw=black!80,line width=0.3pt] (9.720,0.300) -- (9.200,0.000) -- (9.200,0.600) -- (9.720,0.900) -- cycle;
\filldraw[fill=YCtop,draw=black!80,line width=0.3pt] (9.200,1.200) -- (9.720,0.900) -- (9.200,0.600) -- (8.680,0.900) -- cycle;
\filldraw[fill=YCleft,draw=black!80,line width=0.3pt] (9.200,-0.600) -- (9.720,-0.900) -- (9.720,-0.300) -- (9.200,0.000) -- cycle;
\filldraw[fill=YCright,draw=black!80,line width=0.3pt] (10.239,-0.600) -- (9.720,-0.900) -- (9.720,-0.300) -- (10.239,0.000) -- cycle;
\filldraw[fill=YCtop,draw=black!80,line width=0.3pt] (9.720,0.300) -- (10.239,0.000) -- (9.720,-0.300) -- (9.200,0.000) -- cycle;
\filldraw[fill=YCleft,draw=black!80,line width=0.3pt] (8.161,-0.600) -- (8.680,-0.900) -- (8.680,-0.300) -- (8.161,0.000) -- cycle;
\filldraw[fill=YCright,draw=black!80,line width=0.3pt] (9.200,-0.600) -- (8.680,-0.900) -- (8.680,-0.300) -- (9.200,0.000) -- cycle;
\filldraw[fill=YCtop,draw=black!80,line width=0.3pt] (8.680,0.300) -- (9.200,0.000) -- (8.680,-0.300) -- (8.161,0.000) -- cycle;
\filldraw[fill=YCleft,draw=black!80,line width=0.3pt] (8.680,-0.900) -- (9.200,-1.200) -- (9.200,-0.600) -- (8.680,-0.300) -- cycle;
\filldraw[fill=YCright,draw=black!80,line width=0.3pt] (9.720,-0.900) -- (9.200,-1.200) -- (9.200,-0.600) -- (9.720,-0.300) -- cycle;
\filldraw[fill=YCtop,draw=black!80,line width=0.3pt] (9.200,0.000) -- (9.720,-0.300) -- (9.200,-0.600) -- (8.680,-0.300) -- cycle;
\node[font=\small] at (9.2,-1.7) {$\cY(1)\cap\cY(2)$: $\begin{smallmatrix}2&1\\1&1\end{smallmatrix}$};
\end{tikzpicture}
\caption{The overlap of the two three-dimensional Young diagrams of
Figure~\ref{fig:k3stack}.  Below each diagram is its array of heights
(rows \(d_3=1,2\), columns \(d_4=1,2,3\)).  The heights of the intersection
are the entrywise minima, and its volume \(5\) is the exponent of
\(x^{(1)}_2-x^{(1)}_1\).}
\label{fig:k3overlap}
\end{figure}
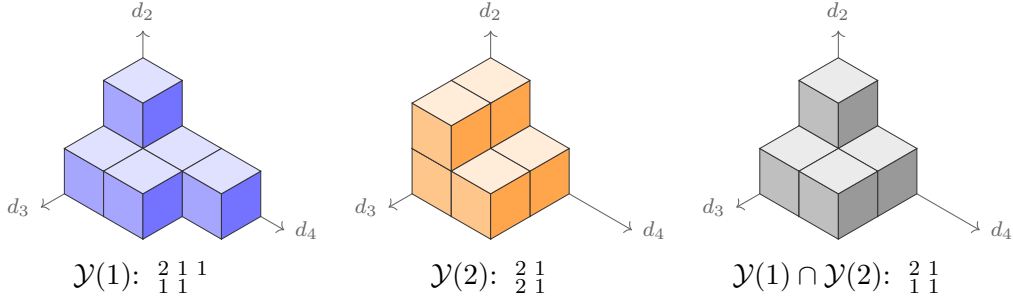

\begin{table}[ht]
\centering
\renewcommand{\arraystretch}{1.25}
\begin{tabular}{c|c|c|c|c}
layer & siblings \(p<q\) & \(\cY(p)\) & \(\cY(q)\) & \(|\cY(p)\cap\cY(q)|\) = exponent\\\hline
1 & \(1,2\) & \(\begin{smallmatrix}2&1&1\\1&1\end{smallmatrix}\) & \(\begin{smallmatrix}2&1\\2&1\end{smallmatrix}\) & 5\\
2 & \(1,2\) & \((2,2,1)\) & \((1)\) & 1\\
2 & \(3,4\) & \((2,2)\) & \((2)\) & 2\\
3 & \(1,2\) & \(\{1,2\}\) & \(\{1,2,3\}\) & 2\\
3 & \(4,5\) & \(\{1,2\}\) & \(\{1,2\}\) & 2\\
3 & \(6,7\) & \(\{1\}\) & \(\{1\}\) & 1\\
4 & all six pairs & -- & -- & 1
\end{tabular}
\caption{All sibling pairs of the graph \(D=(2,4,7,12)\), \(p_1=[2,2]\),
\(c_2=[2,1,2,2]\), \(c_3=[2,3,1,2,2,1,1]\).  Layer-\(1\) diagrams are given by
height arrays, layer-\(2\) diagrams by partitions, layer-\(3\) diagrams by
intervals.  The last column was also obtained independently from the order of
vanishing of the determinant as \(x^{(i)}_q\to x^{(i)}_p\).}
\label{tab:D24712}
\end{table}

\begin{figure}[ht]
\centering
\begin{tikzpicture}[scale=0.8]
\draw[->] (0,0)--(3.2,0) node[right,font=\small] {$d_2$};
\draw[->] (0,0)--(0,3.2) node[above,font=\small] {$d_3$};
\draw[fill=blue!10] (0,0) rectangle (1,1) (0,1) rectangle (1,2);
\node[font=\small] at (1.5,-0.7) {$\cY(1)$, \(\lambda=(1,1)\)};
\begin{scope}[xshift=6cm]
\draw[->] (0,0)--(3.2,0) node[right,font=\small] {$d_2$};
\draw[->] (0,0)--(0,3.2) node[above,font=\small] {$d_3$};
\draw[fill=blue!10] (0,0) rectangle (1,1) (1,0) rectangle (2,1);
\node[font=\small] at (1.5,-0.7) {$\cY(2)$, \(\lambda=(2)\)};
\end{scope}
\end{tikzpicture}
\caption{The Young diagrams \eqref{eq:youngprofile} of the two nodes of
\(V_1\) in Remark~\ref{rem:L3}, with cubes corresponding to the points \((d_2,d_3)\).  Node
\(1\) is admissible for \((d_2,d_3)=(1,2)\) but not \((2,1)\); node \(2\) the
reverse.  The objects are not nested, so neither order of the two nodes makes
every admissible set an initial segment.  Their overlap is one cube, the
exponent of \(x^{(1)}_2-x^{(1)}_1\).}
\label{fig:L3ferrers}
\end{figure}
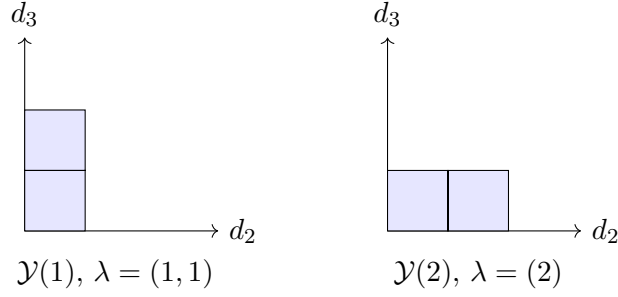

\begin{remark}[Nesting and triangularity]\label{rem:nested}
A node \(u\in V_i\) is admissible for continuation data
\((d_{i+1},\ldots,d_L)\) exactly when this tuple is a point of \(\cY(u)\).
Hence the condition of Remark~\ref{rem:nontri}, that every admissible set be
an initial segment of the children, holds when \(\cY(u)\supseteq\cY(u')\) for
all siblings \(u<u'\).  For \(L=2\) the Young diagrams of the nodes of \(V_1\)
are the intervals \(\{1,\ldots,p_{1,a}\}\), and nesting is the condition
\(p_{1,1}\ge p_{1,2}\ge\cdots\) of Remark~\ref{rem:nontri}.  In the graph of
Remark~\ref{rem:L3}, the Young diagrams of the two nodes of \(V_1\) are those
of the partitions \((1,1)\) and \((2)\) (Figure~\ref{fig:L3ferrers}).  They are
not nested in either order, which is why, for generic eigenvalues, no
labelling makes the evaluation matrix triangular there.  Their overlap is one cube, the exponent of
\(x^{(1)}_2-x^{(1)}_1\) found in Remark~\ref{rem:L3}.
\end{remark}

\begin{remark}\label{rem:young-operations}
Proposition~\ref{prop:overlap} is a reformulation of the count
\eqref{eq:exptrunc}.  What remains to be developed is an intrinsic theory of
these Young diagrams, including the operations on them induced by grafting and
truncating degeneracy graphs.
\end{remark}

\section{Projectors in the monomial basis}
\label{sec:projectors}

The motivating question of \cite{KempRamgoolam2026} is the change of basis
between the projector basis of the algebra and monomials in the
generators.  In this section we recall the algebra, deduce from
Theorem~\ref{thm:KR} that the graph monomials form a basis, and give an
explicit procedure for inverting the evaluation matrix, and hence for computing
the coefficients of the projectors in the monomial basis.

\subsection{The algebra and its generators}
\label{sec:algebra}

The final nodes label the projector basis of a commutative
associative semisimple algebra
\begin{equation}
\cA_L
=
\bigoplus_{T\in V_L}\C P_T,
\qquad
P_TP_{T'}=\delta_{T,T'}P_T ,
\qquad \sum_{T\in V_L}P_T=1,
\label{eq:CASS}
\end{equation}
denoted \(\cA\) in \cite{KempRamgoolam2026}.  For a node \(u\in V_i\), let
\(P_u^{(i)}\) be the sum of the \(P_T\) over the final nodes \(T\) descending
from \(u\), so that \(P_u^{(i)}=\sum_{w\in\Ch_{i+1}(u)}P_w^{(i+1)}\).  The
generators are
\begin{equation}
C_i=\sum_{u\in V_i}x_u^{(i)}P_u^{(i)},\qquad 1\le i\le L,
\label{eq:CiCASS}
\end{equation}
and the subalgebras generated by \(C_1\), by \(C_1,C_2\), \ldots, by
\(C_1,\ldots,C_L\) have projector bases labelled by the nodes of \(V_1\),
\(V_2\), \ldots, \(V_L\).  Siblings have distinct eigenvalues, since the
children of a node are distinguished precisely by the eigenvalue of the next
generator.  For a polynomial \(P\),
\begin{equation}
P(C_1,\ldots,C_L)\,P_T=P(T)\,P_T ,
\label{eq:polyacts}
\end{equation}
where \(P(T)\) is the value at the final node \(T\), as in
Section~\ref{sec:evalNewton}.

Sections~\ref{sec:proof} and~\ref{sec:young} use only the graph and the eigenvalues.  Every tower of
commutative semisimple algebras generated by \(C_1,\ldots,C_L\) defines such a
graph \cite{KempRamgoolam2026}, so the results apply to any such tower, not
only to the symmetric group algebras which motivated \cite{KempRamgoolam2026}.

\subsection{The monomial basis}
\label{sec:monomial-basis}

Let \(\cM\) be the graph monomial evaluation matrix \eqref{eq:Mmatrix}, with
rows and columns labelled by final nodes in path order.  By \eqref{eq:polyacts} and
\(\sum_{T'}P_{T'}=1\),
\begin{equation}
m_T(C_1,\ldots,C_L)=\sum_{T'\in V_L}\cM_{T',T}\,P_{T'} .
\label{eq:monomial-projector-expansion}
\end{equation}
By Theorem~\ref{thm:KR}, \(\det\cM\) is a product of powers of differences of
sibling eigenvalues, and since siblings have distinct eigenvalues
(Section~\ref{sec:algebra}) it is nonzero.  Hence the graph monomials
\(m_T(C_1,\ldots,C_L)\) form a basis of \(\cA_L\), as conjectured in
\cite{KempRamgoolam2026}, and
\begin{equation}
P_T=\sum_{S\in V_L}\bigl(\cM^{-1}\bigr)_{S,T}\,m_{S}(C_1,\ldots,C_L).
\label{eq:projector-monomial-expansion}
\end{equation}
The \(T\)-th column of \(\cM^{-1}\) gives the coefficients of \(P_T\) in the
monomial basis.

\subsection{Inverting the evaluation matrix}
\label{sec:inverse-evaluation}

Let \(U\) be the change of basis \eqref{eq:Umatrix} from graph monomials to
graph Newton polynomials, and \(\cN\) the Newton evaluation matrix
\eqref{eq:Nmatrix}.  As in Section~\ref{sec:graphNewton},
\begin{equation}
\cN=\cM U,\qquad \cM^{-1}=U\cN^{-1}.
\label{eq:inverse-evaluation-factorisation}
\end{equation}
Equivalently, \(N_T(C_1,\ldots,C_L)=\sum_{T'}\cN_{T',T}P_{T'}\): the columns of
\(\cN\) expand the Newton polynomials in projectors, and inverting \(\cN\)
expresses the projectors in terms of Newton polynomials, which are then
expanded in monomials using \(U\).  By Theorem~\ref{thm:KR} both matrices are
invertible, and both inversions can be made explicit.

\paragraph{The unitriangular change of basis.}  Writing \(U=1+K_U\), the
matrix \(K_U\) is strictly upper triangular when the final nodes are listed by
the total degree of \(m_T\) (Section~\ref{sec:graphNewton}), hence nilpotent, and
\begin{equation}
U^{-1}=\sum_{r=0}^{D_L-1}(-K_U)^r .
\label{eq:inverse-unitriangular-change}
\end{equation}

\paragraph{The triangular case.}  When every admissible set is an initial
segment of the children (Remark~\ref{rem:nontri}), \(\cN\) is lower triangular.
Writing \(\cN=(1+K_{\cN})\,\cN_{\rm diag}\) with \(\cN_{\rm diag}\) the diagonal
part and \(K_{\cN}\) strictly lower triangular,
\begin{equation}
\cN^{-1}=\cN_{\rm diag}^{-1}\Bigl(\sum_{r=0}^{D_L-1}(-K_{\cN})^r\Bigr),
\qquad
\cM^{-1}=U\,\cN_{\rm diag}^{-1}\Bigl(\sum_{r=0}^{D_L-1}(-K_{\cN})^r\Bigr).
\label{eq:inverse-monomial-evaluation}
\end{equation}
Entrywise, the powers of a strictly triangular matrix are sums over monotone
chains of final nodes, so the coefficients in
\eqref{eq:projector-monomial-expansion} are finite alternating sums over such
chains.

\paragraph{The general case.}  In general \(\cN\) is not triangular
(Section~\ref{sec:evalNewton}), but the proof of Theorem~\ref{thm:newtondet} showed that
\(\cN=W\cN'\), with \(\cN'\) block triangular, \eqref{eq:NPfactor}.  The same
structure inverts it.  By \eqref{eq:inverse-evaluation-factorisation} and
\(N_T(C_1,\ldots,C_L)=\sum_{T'}\cN_{T',T}P_{T'}\),
\begin{equation}
P_T=\sum_{S\in V_L}\bigl(\cN^{-1}\bigr)_{S,T}\,N_{S}(C_1,\ldots,C_L),
\label{eq:projector-Newton-expansion}
\end{equation}
so it suffices to compute the columns of \(\cN^{-1}\).  For \(L=1\), \(\cN\) is
the lower triangular Newton matrix of Section~\ref{sec:proof}, and its inverse is the
matrix of divided differences \(Q_{v_0}=W_{v_0}^{-1}\) of Remark~\ref{rem:divdiff} for the root,
\((\cN^{-1})_{k,j}=1/\prod_{m\le k,\,m\ne j}(x^{(1)}_j-x^{(1)}_m)\) for \(j\le k\)
and \(0\) otherwise, the coefficients in \eqref{eq:divdiff}.  For
\(L\ge2\), \eqref{eq:NPfactor} gives
\begin{equation}
\cN^{-1}=\cN'^{-1}W^{-1}=\cN'^{-1}Q ,
\label{eq:Einverse-general}
\end{equation}
where \(Q=W^{-1}\) is block diagonal, with one block \(Q_q=W_q^{-1}\) for the
children of each \(q\in V_{L-1}\): the matrix of the divided differences
\eqref{eq:divdiff} (Remark~\ref{rem:divdiff}).  Here \(\cN'\) is block triangular after ordering by
\(d=d_L\), with the Newton evaluation matrices of the final-layer-excised
graphs \(\Gamma_{[d]}\) as diagonal blocks.  Thus the determinant is computed
from the Newton expansion, and the inverse from its closed form, the divided
differences.
Theorem~\ref{thm:KR} guarantees that all the matrices inverted here are
invertible.  Explicitly, the column \(c=\cN^{-1}e_T\), where \(e_T\) is the unit
vector of the final node \(T\), is computed as follows.  As in
Corollary~\ref{cor:factor}, label each final node by the pair \((q,d)\), where
\(q=v_{L-1}(T)\) and \(T\) is the \(d\)-th child of \(q\), write
\(c_{q,d}=c_T\), and let \(d_{\max}\) be, as in the proof of
Theorem~\ref{thm:newtondet}, the largest number of children of a node of
\(V_{L-1}\).

\begin{enumerate}
\item \emph{Divided differences among siblings (the matrix \(Q\)).}  For each
  \(q\in V_{L-1}\), with \(n\) children, compute \(g(q,d)=\Delta_{q,d}\,e_T\)
  for \(1\le d\le n\), where \(e_T\) is regarded as the function on \(V_L\)
  equal to \(1\) at \(T\) and \(0\) elsewhere.  If \(T\) is the \(k\)-th child
  of \(q_0\), and \(\xi_1,\xi_2,\ldots\) are the eigenvalues of \(C_L\) on the
  children of \(q_0\), then \(g(q,d)=0\) unless \(q=q_0\) and \(d\ge k\), and
  \(g(q_0,d)=\prod_{m\le d,\,m\ne k}(\xi_k-\xi_m)^{-1}\).
\item \emph{The entries of \(\cN'\) above the diagonal blocks.}  For \(d<d'\),
  \(q\in S^{(L-1)}_{[d]}\) and \(q'\in S^{(L-1)}_{[d']}\), let \(T\) be the
  final node labelled \((q',d')\), and compute
  \(\cN'_{(q,d),T}=\Delta_{q,d}N_{T}\), the \(d\)-th Newton coefficient of
  \(N_{T}\) on the children of \(q\).  (For \(d>d'\) these
  entries vanish, and for \(d=d'\) they equal \(\cN(\Gamma_{[d]})_{q,q'}\), by
  the proof of Theorem~\ref{thm:newtondet}.)
\item \emph{Block back-substitution.}  For \(d=d_{\max},d_{\max}-1,\ldots,1\)
  in turn, with the \(c_{q',d'}\), \(d'>d\), already known, set
  \begin{equation}
  h_d(q)=g(q,d)-\sum_{d'>d}\ \sum_{q'\in S^{(L-1)}_{[d']}}
  \cN'_{(q,d),(q',d')}\,c_{q',d'},
  \qquad q\in S^{(L-1)}_{[d]},
  \end{equation}
  and determine the \(c_{q',d}\) from the diagonal block,
  \begin{equation}
  \sum_{q'\in S^{(L-1)}_{[d]}}\cN(\Gamma_{[d]})_{q,q'}\;c_{q',d}=h_d(q),
  \qquad q\in S^{(L-1)}_{[d]} .
  \label{eq:block-back-substitution}
  \end{equation}
  The matrix of this system is the Newton evaluation matrix of
  \(\Gamma_{[d]}\), which has \(L-1\) layers; it is inverted by the same
  procedure.
\end{enumerate}
The recursion ends at one layer, where the inverse is the matrix of divided
differences.  Finally \(\cM^{-1}=U\cN^{-1}\): expanding the \(N_{S}\) in
\eqref{eq:projector-Newton-expansion} in graph monomials gives the
coefficients in \eqref{eq:projector-monomial-expansion}.

The only divisions in the procedure are by differences of eigenvalues of
siblings, in step 1 and in the corresponding steps for the graphs
\(\Gamma_{[d]}\), whose siblings are siblings in the original graph.  Since
\(U\) has polynomial entries, the coefficients of the projectors in the
monomial basis are polynomials in the eigenvalues divided by products of
differences of sibling eigenvalues, consistent with Theorem~\ref{thm:KR}.  We
have implemented the procedure and checked, in exact arithmetic, that it
reproduces \(\cM^{-1}\) obtained by direct inversion for the graphs
\(D=(2,3)\), \((2,5)\) with \(p_1=[2,3]\), \((4,10)\), \((2,3,4)\), \((2,4,7)\) and
\((2,4,7,12)\), all but \((2,5)\) treated in this paper.

\paragraph{Example.}  For the graph of Figure~\ref{fig:L2}, write
\(x_k=x^{(1)}_k\) and \(y_k=x^{(2)}_{T_k}\).  The graph monomials are
\(1,C_1,C_2\) and the Newton polynomials are \(N_{T_1}=1\),
\(N_{T_2}=C_1-x_1\), \(N_{T_3}=C_2-y_2\), so that, with rows and columns in
the order \(T_1,T_2,T_3\),
\begin{equation}
U=\begin{pmatrix}1&-x_1&-y_2\\0&1&0\\0&0&1\end{pmatrix}.
\end{equation}  The final nodes
are \(T_1=(1,1)\), \(T_2=(2,1)\), \(T_3=(2,2)\) and \(d_{\max}=2\).  We compute
the column of \(\cN^{-1}\) for \(T_3\), and hence \(P_{T_3}\).
\begin{enumerate}
\item \(T_3\) is the second child of node \(2\), so \(g(1,1)=g(2,1)=0\) and
  \(g(2,2)=1/(y_3-y_2)\).  This is the difference quotient in
  \eqref{eq:rowopexample}.
\item The entries of \(\cN'\) above the diagonal blocks are
  \(\cN'_{(q,1),T_3}=N_{T_3}(\text{first child of }q)\), equal to \(y_1-y_2\) for
  \(q=1\) and \(0\) for \(q=2\): the top right entry of \eqref{eq:E23after}.
\item For \(d=2\), the diagonal block is the \(1\times1\) matrix \((1)\) of
  \(\Gamma_{[2]}\), so \(c_{T_3}=g(2,2)=1/(y_3-y_2)\).  For \(d=1\),
  \begin{equation}
  h_1(1)=0-(y_1-y_2)\,c_{T_3}=-\frac{y_1-y_2}{y_3-y_2},\qquad
  h_1(2)=0 ,
  \end{equation}
  and the diagonal block is the \(2\times2\) Newton--Vandermonde matrix of
  \(\Gamma_{[1]}\) in \eqref{eq:E23after}, whose inverse is the matrix of
  divided differences: \(c_{T_1}=h_1(1)\) and
  \(c_{T_2}=\bigl(h_1(2)-h_1(1)\bigr)/(x_2-x_1)\).
\end{enumerate}
By \eqref{eq:projector-Newton-expansion},
\begin{equation}
P_{T_3}=\frac{(C_2-y_2)-(y_1-y_2)+(y_1-y_2)\dfrac{C_1-x_1}{x_2-x_1}}{y_3-y_2}
=\frac{(C_2-y_2)-(y_1-y_2)\,\dfrac{C_1-x_2}{x_1-x_2}}{y_3-y_2}.
\end{equation}
The same steps for \(T_1\) give \(g(1,1)=1\) and all other \(g=0\), hence
\(c_{T_3}=0\), \(c_{T_1}=1\), \(c_{T_2}=-1/(x_2-x_1)\) and
\(P_{T_1}=(C_1-x_2)/(x_1-x_2)\).  Finally
\(P_{T_2}=1-P_{T_1}-P_{T_3}\).  These are the expansions
\eqref{eq:projector-monomial-expansion} of the three projectors in the
monomial basis \(1,C_1,C_2\).

\section{Conclusions and Perspectives}
We have proved the Kemp–Ramgoolam determinant conjecture \cite{KempRamgoolam2026} for degeneracy graphs with any number of layers (Theorem~\ref{thm:KR}). The determinant of the graph monomial evaluation matrix factorises into powers of differences of sibling eigenvalues, with exponents determined by the branching structure (equation~\eqref{eq:exptrunc}). Since distinct children of a node are distinguished by distinct eigenvalues of the next generator (Section~\ref{sec:algebra}), every factor in the determinant formula is non-zero, so the graph monomials form a basis of the algebra (Proposition~\ref{prop:bijection} and Section~\ref{sec:monomial-basis}). This establishes the basis property beyond the one- and two-layer cases shown in \cite{KempRamgoolam2026}, for the full class of algebras described by these graphs.

The method employs Newton polynomial analogues for graphs (equation~\eqref{eq:graphNewton}). The leading part of these polynomials defines a unitriangular transformation from the monomials of the graph (equations~\eqref{eq:leadingterm}--\eqref{eq:changebasisdet}), so that the determinant of the monomial evaluation matrix equals that of the Newton evaluation matrix, which is the product of its diagonal entries (Theorem~\ref{thm:newtondet}). The proof expands the rows of siblings in the Newton basis of their eigenvalues (Lemma~\ref{lem:newton}): the Newton evaluation matrix factorises as a block diagonal matrix of one-layer Newton--Vandermonde matrices times a matrix which, by a degree argument (Lemma~\ref{lem:degree}), is block triangular. The diagonal blocks correspond to subgraphs that arise from deleting the last layer and keeping only those nodes that meet some fixed child condition (Definition~\ref{def:final-layer-excised}, Lemma~\ref{lem:restrict} and Corollary~\ref{cor:factor}). This yields an induction on the number of layers, which applies also when the Newton evaluation matrix is not triangular. The same structure provides a recursive inverse, in which the divided differences invert the Newton expansion among siblings (Remark~\ref{rem:divdiff}), and thus explicit expansions for the projector basis in the monomial basis (Section~\ref{sec:inverse-evaluation}, equations~\eqref{eq:projector-Newton-expansion}--\eqref{eq:block-back-substitution}).

We have also given a geometric interpretation of the determinant exponents. The admissible continuation data of a node form a higher-dimensional Young diagram (Definition~\ref{def:Y}), and the exponent associated with two siblings is the volume of the intersection of their diagrams (Proposition~\ref{prop:overlap}). Intervals, ordinary Young diagrams and plane partitions arise at successive depths (Section~\ref{sec:lowdepths}). This description expresses the multiplicities in the determinant through the overlap of the branching data below the two nodes. It also suggests studying how these diagrams change under operations on degeneracy graphs (Remark~\ref{rem:young-operations}).

Behind both the proof and the geometric description lies a simple structural fact: the Vandermonde determinant is the basic building block. The graph determinant is a product of Vandermonde determinants of eigenvalues of siblings, one for each set of siblings and each value of the continuation data (Corollary~\ref{cor:vdm}). In the proof, each of these factors is the determinant of a one-layer Newton--Vandermonde matrix (Remark~\ref{rem:vdm}): either a block of the factor produced by the Newton expansion among siblings at some step of the induction (Lemma~\ref{lem:newton}), or the Newton evaluation matrix of one of the one-layer graphs at which the induction ends. The two-fold tensor product of algebras generated by single elements, whose evaluation matrix is the Kronecker product of two Vandermonde matrices, is the simplest special case. Degeneracy graphs are much more general: the branching varies from node to node, the evaluation matrix is no longer a Kronecker product, and the Newton evaluation matrix is in general not triangular, yet the determinant keeps the structure of a product of Vandermonde determinants.

Our proof is purely linear-algebraic: it uses the graph Newton polynomials
and Newton expansions of the rows of the evaluation matrix, and does not rely on the theory
of Gröbner bases.  Nevertheless there is a natural connection.  Let
\(X\subset\C^L\) be the set of joint eigenvalue tuples of \(C_1,\ldots,C_L\),
one for each final node, and \(I(X)\) the ideal of polynomials vanishing on
\(X\).  Then \(t_i\mapsto C_i\) induces
\(\cA_L\cong\C[t_1,\ldots,t_L]/I(X)\), and by Theorem~\ref{thm:KR} the graph
monomials form a basis of this quotient.  Their exponent set is closed under
decreasing any coordinate, as is the set of standard monomials of a Gröbner
basis of \(I(X)\) \cite{CLO}; for the lexicographic order the standard
monomials can be constructed combinatorially from \(X\)
\cite{CerliencoMureddu,Ceria}.  The relation between these constructions is
left for future work.

Moreover, we would like to extend our method to the case of Brauer and walled Brauer algebras. The determinant theorem (Theorem~\ref{thm:KR}) will be immediately applicable to the commutative semisimple subalgebras of those algebras provided that some selected family of commuting elements separates the necessary labels. A family of Jucys–Murphy elements along with branching paths is a natural candidate for this programme. The construction of degeneracy graphs would provide us with monomial bases and expansions of projectors in terms of spectral information.

Our main objective in these applications is the construction of matrix units. We plan to use the theorem to obtain the projectors on the diagonal and then combine them with suitable elements connecting paths in the same simple block. If \(\lambda\) labels such a block and \(T,S\) label its paths, the aim is to construct operators satisfying
\begin{equation}
E^{\lambda}_{TS}E^{\mu}_{UV}
=\delta_{\lambda\mu}\delta_{SU}E^{\lambda}_{TV},
\qquad E^{\lambda}_{TT}=P^{\lambda}_{T}.
\end{equation}
The projector expansions (equations~\eqref{eq:projector-monomial-expansion} and \eqref{eq:projector-Newton-expansion}) supply the diagonal part of this construction. The remaining task is to choose the connecting elements and fix their normalisations consistently. We expect this approach to be useful for explicit calculations in tensor representations of Brauer and walled Brauer algebras, including applications in quantum information. We will first consider semisimple settings; extending the construction to non-semisimple specialisations will require additional information beyond the joint eigenvalues.

\section*{Note added}

As we were reaching the final stages of the write-up of this paper, we received
communication of an independent proof of the determinant and monomial-basis
conjectures of Kemp and Ramgoolam \cite{KempRamgoolam2026} by Haobo Ma and
Wenlin Zhang, formalized in Lean~4
(\url{https://github.com/the-omega-institute/degeneracy-graph-lean}).

\section*{Acknowledgments}

We thank Garreth Kemp for the collaboration
\cite{KempRamgoolam2026} in which the degeneracy graphs, the monomial bases and
the determinant conjectures proved here were formulated.  It is a pleasure to
acknowledge useful conversations with Yang-Hui He, Vishnu Jejjala, Garreth Kemp, Yang Lei, Brien Nolan, Denjoe O'Connor and Adrian Padellaro.  We
acknowledge useful conversations with ChatGPT (OpenAI) and Claude (Anthropic),
which helped in conceiving and developing the proof, in identifying and repairing a gap in an
earlier version of the triangularity argument, and in checking the examples
and exponent formulae computationally.  SR is supported by the Science and Technology Facilities Council
(STFC) Consolidated Grant ST/T000686/1 ``Amplitudes, strings and duality''
and the Royal Society International Exchanges Grant
IEC\textbackslash NSFC\textbackslash 242376.  The work of MS is carried out under IRA
Programme, project no.\ FENG.02.01-IP.05-0006/23, financed by the FENG program
2021-2027, Priority FENG.02, Measure FENG.02.01., with the support of the FNP.

\appendix

\section{Further examples of the block triangular reduction}
\label{app:examples}

In this appendix we work through two examples of \cite{KempRamgoolam2026}
whose Newton evaluation matrices are not triangular.  The first is spelled out
step by step, following the treatment of the graph of Figure~\ref{fig:L2} in
Section~\ref{sec:evalNewton}.  The second is worked out to the same level of detail,
with each step referred back to Lemma~\ref{lem:restrict}, Corollary~\ref{cor:factor},
Lemmas~\ref{lem:newton} and~\ref{lem:degree}, and the proof of
Theorem~\ref{thm:newtondet}.

\subsection{The running example \texorpdfstring{\(D=(2,4,7)\)}{D=(2,4,7)}}
\label{app:D247}

The graph has \(D=(2,4,7)\), \(p_1=[2,2]\), \(c_2=[2,1,2,2]\), and is drawn in
Figure~\ref{fig:D247app}, with the labels and eigenvalues used below.  Node \(1\in V_1\) has children \(1,2\in V_2\)
and node \(2\in V_1\) has children \(3,4\in V_2\).  In \(V_3\), node \(1\in V_2\)
has children \(T_1<T_2\), node \(2\) has the single child \(T_3\), node \(3\) has
children \(T_4<T_5\), and node \(4\) has children \(T_6<T_7\).  To keep the
matrices readable we write
\begin{equation}
x_k=x^{(1)}_k,\qquad y_k=x^{(2)}_k,\qquad z_k=x^{(3)}_{T_k},\qquad
\alpha=x_2-x_1 .
\end{equation}

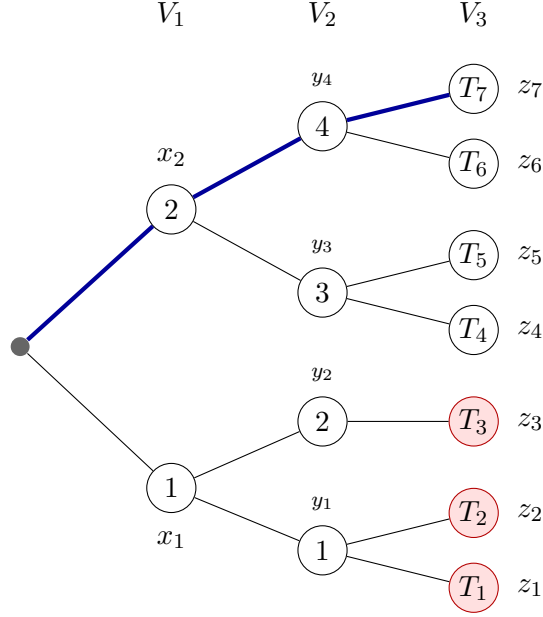
\begin{figure}[ht]
\centering
\begin{tikzpicture}[yscale=1.1]
\node[rt] (r) at (0,2.9) {};
\node[gn] (a1) at (2.0,1.2) {$1$};
\node[gn] (a2) at (2.0,4.55) {$2$};
\node[gn] (b1) at (4.0,0.45) {$1$};
\node[gn] (b2) at (4.0,2.0) {$2$};
\node[gn] (b3) at (4.0,3.55) {$3$};
\node[gn] (b4) at (4.0,5.55) {$4$};
\foreach \k/\y in {1/0,2/0.9,3/2.0}
  \node[gn,bad] (t\k) at (6.0,\y) {$T_\k$};
\foreach \k/\y in {4/3.1,5/4.0,6/5.1,7/6.0}
  \node[gn] (t\k) at (6.0,\y) {$T_\k$};
\draw (r)--(a1) (a1)--(b1) (a1)--(b2) (a2)--(b3)
  (b1)--(t1) (b1)--(t2) (b2)--(t3) (b3)--(t4) (b3)--(t5) (b4)--(t6);
\draw[pth] (r)--(a2) (a2)--(b4) (b4)--(t7);
\node[below=1mm,font=\small] at (a1.south) {$x_1$};
\node[above=1mm,font=\small] at (a2.north) {$x_2$};
\foreach \k in {1,2,3,4} \node[above=0.5mm,font=\scriptsize] at (b\k.north) {$y_\k$};
\foreach \k in {1,...,7} \node[right=1mm,font=\small] at (t\k.east) {$z_\k$};
\foreach \x/\l in {2.0/V_1,4.0/V_2,6.0/V_3} \node[font=\small] at (\x,6.9) {$\l$};
\end{tikzpicture}
\caption{The graph \(D=(2,4,7)\), \(p_1=[2,2]\), \(c_2=[2,1,2,2]\) of
Figure~\ref{fig:k2}, with the final nodes labelled and the eigenvalues
\(x_k=x^{(1)}_k\), \(y_k=x^{(2)}_k\), \(z_k=x^{(3)}_{T_k}\).  The path of
\(T_7\) is drawn thick.  Its rank data are \((d_2,d_3)=(2,2)\), and
\(R_1(2,2)=\{2\}\), so \(N_{T_7}=(C_2-y_3)(C_3-z_6)\) has no factor
\(C_1-x_1\) and, for generic eigenvalues, does not vanish at the shaded final nodes \(T_1,T_2,T_3\),
which descend from node \(1\in V_1\).}
\label{fig:D247app}
\end{figure}

\paragraph{Step 1: rank data, graph monomials and graph Newton polynomials.}
Reading each path from right to left as in Section~\ref{sec:admchld} gives
\begin{equation}
{\renewcommand{\arraystretch}{1.3}\begin{array}{c|c|c|c|l}
T & \text{path }v_1\to v_2 & (a,d_2,d_3) & m_T & N_T\\\hline
T_1 & 1\to1 & (0,1,1) & 1 & 1\\
T_2 & 1\to1 & (0,1,2) & C_3 & C_3-z_1\\
T_3 & 1\to2 & (0,2,1) & C_2 & C_2-y_1\\
T_4 & 2\to3 & (1,1,1) & C_1 & C_1-x_1\\
T_5 & 2\to3 & (1,1,2) & C_1C_3 & (C_1-x_1)(C_3-z_4)\\
T_6 & 2\to4 & (1,2,1) & C_1C_2 & (C_1-x_1)(C_2-y_3)\\
T_7 & 2\to4 & (0,2,2) & C_2C_3 & (C_2-y_3)(C_3-z_6)
\end{array}}
\end{equation}
The last row is the one which breaks triangularity.  For \(T_7\),
\((d_2,d_3)=(2,2)\), and by \eqref{eq:R1} and \eqref{eq:SLm1}--\eqref{eq:Si},
\begin{equation}
R_1(2,2)=S^{(1)}_{[2,2]}
=\bigl\{u\in V_1:\ |\Ch_2(u)\cap S^{(2)}_{[2]}|\ge2\bigr\},
\end{equation}
where
\begin{equation}
S^{(2)}_{[2]}=\bigl\{w\in V_2:\ |\Ch_3(w)|\ge2\bigr\}=\{1,3,4\}.
\end{equation}
Node \(1\in V_1\) has children \(\{1,2\}\), of which only \(1\) lies in
\(S^{(2)}_{[2]}\); node \(2\in V_1\) has children \(\{3,4\}\), both in
\(S^{(2)}_{[2]}\).  Hence \(R_1(2,2)=\{2\}\): node \(1\) is an earlier sibling
of \(v_1(T_7)=2\) but is not admissible, so \(N_{T_7}\) has no factor
\(C_1-x_1\), and its rank datum \(a\) is \(0\).

\paragraph{Step 2: the Newton evaluation matrix.}  With rows \(T'\) and columns
\(N_T\) in path order,
\begin{equation}
{\footnotesize\setlength{\arraycolsep}{2pt}
\cN=
\begin{pmatrix}
1&0&0&0&0&0&(y_1-y_3)(z_1-z_6)\\
1&z_2-z_1&0&0&0&0&(y_1-y_3)(z_2-z_6)\\
1&z_3-z_1&y_2-y_1&0&0&0&(y_2-y_3)(z_3-z_6)\\
1&z_4-z_1&y_3-y_1&\alpha&0&0&0\\
1&z_5-z_1&y_3-y_1&\alpha&\alpha(z_5-z_4)&0&0\\
1&z_6-z_1&y_4-y_1&\alpha&\alpha(z_6-z_4)&\alpha(y_4-y_3)&0\\
1&z_7-z_1&y_4-y_1&\alpha&\alpha(z_7-z_4)&\alpha(y_4-y_3)&(y_4-y_3)(z_7-z_6)
\end{pmatrix}
\begin{array}{l}
\leftarrow T_1\\ \leftarrow T_2\\ \leftarrow T_3\\ \leftarrow T_4\\
\leftarrow T_5\\ \leftarrow T_6\\ \leftarrow T_7
\end{array}}
\label{eq:N247}
\end{equation}
The only entries above the diagonal are in the column of \(N_{T_7}\), in the
rows of \(T_1,T_2,T_3\): these are the final nodes descending from node
\(1\in V_1\), at which \(N_{T_7}\) generically does not vanish because node \(1\) is not
admissible for \(T_7\) (the shaded nodes of Figure~\ref{fig:D247app}).

\paragraph{Step 3: Newton expansion among the siblings of layer 3.}  The
nodes of \(V_2\) with more than one child are \(1,3,4\), with children
\(T_1<T_2\), \(T_4<T_5\) and \(T_6<T_7\).  As in \eqref{eq:rowopexample}, we
expand the rows of each pair of siblings in the Newton basis \(1,\ \xi-\xi_1\)
of their eigenvalues of \(C_3\) (Lemma~\ref{lem:newton}):
\begin{equation}
\mathbf r_{T_2}=\mathbf r_{T_1}+(z_2-z_1)\,\Delta_{1,2},\qquad
\mathbf r_{T_5}=\mathbf r_{T_4}+(z_5-z_4)\,\Delta_{3,2},\qquad
\mathbf r_{T_7}=\mathbf r_{T_6}+(z_7-z_6)\,\Delta_{4,2},
\end{equation}
where \(\mathbf r_{T'}\) is the row of \(T'\) in \eqref{eq:N247}.  The rows of
\(T_1,T_3,T_4,T_6\) are kept, as the Newton coefficients \(\Delta_{q,1}\), and
the new rows \(\Delta_{1,2},\Delta_{3,2},\Delta_{4,2}\) replace the rows of
\(T_2,T_5,T_7\).  Thus \(\cN=W\cN'\) as in \eqref{eq:NPfactor}, where the
blocks of \(W\) for the nodes \(1,3,4\in V_2\) are the one-layer
Newton--Vandermonde matrices \(\bigl(\begin{smallmatrix}1&0\\1&z_2-z_1\end{smallmatrix}\bigr)\),
\(\bigl(\begin{smallmatrix}1&0\\1&z_5-z_4\end{smallmatrix}\bigr)\),
\(\bigl(\begin{smallmatrix}1&0\\1&z_7-z_6\end{smallmatrix}\bigr)\), and the
block for node \(2\) is \((1)\); hence \(\det W=(z_2-z_1)(z_5-z_4)(z_7-z_6)\).  We now list rows and columns by the position
\(d_3\) of the final node among its siblings: first the final nodes with
\(d_3=1\), namely \(T_1,T_3,T_4,T_6\), the first children of the nodes
\(1,2,3,4\in V_2\), then those with \(d_3=2\), namely \(T_2,T_5,T_7\), the
second children of the nodes \(1,3,4\).  The nodes of \(V_2\) have \(2,1,2,2\)
children, so there are \(d_{\max}=2\) diagonal blocks, of sizes \(4\) and \(3\):
the parts of the conjugate of the partition \((2,2,2,1)\) of \(7\).  The result
is
\begin{equation}
{\footnotesize\setlength{\arraycolsep}{2pt}
\cN'=
\left(\begin{array}{cccc|ccc}
1&0&0&0&0&0&(y_1-y_3)(z_1-z_6)\\
1&y_2-y_1&0&0&z_3-z_1&0&(y_2-y_3)(z_3-z_6)\\
1&y_3-y_1&\alpha&0&z_4-z_1&0&0\\
1&y_4-y_1&\alpha&\alpha(y_4-y_3)&z_6-z_1&\alpha(z_6-z_4)&0\\ \hline
0&0&0&0&1&0&y_1-y_3\\
0&0&0&0&1&\alpha&0\\
0&0&0&0&1&\alpha&y_4-y_3
\end{array}\right)
\begin{array}{l}
\leftarrow T_1\\ \leftarrow T_3\\ \leftarrow T_4\\ \leftarrow T_6\\
\leftarrow \Delta_{1,2}\\ \leftarrow \Delta_{3,2}\\ \leftarrow \Delta_{4,2}
\end{array}}
\label{eq:N247after}
\end{equation}
with columns \(N_{T_1},N_{T_3},N_{T_4},N_{T_6}\,|\,N_{T_2},N_{T_5},N_{T_7}\).
The entries can be read off from the degree property, Lemma~\ref{lem:degree},
as in the proof of Theorem~\ref{thm:newtondet}.  The columns \(N_{T_1},N_{T_3},
N_{T_4},N_{T_6}\) do not involve \(C_3\), so they take equal values on
siblings, and their Newton coefficients \(\Delta_{q,2}\) vanish: this is the
zero block in the lower left.  The columns \(N_{T_2},N_{T_5},N_{T_7}\) are
linear in \(C_3\), with coefficients of \(C_3\) equal to \(1\), \(C_1-x_1\) and
\(C_2-y_3\), and \(\Delta_{q,2}\) returns these coefficients evaluated on the
path to \(q\).

\paragraph{Step 4: the diagonal blocks.}  The upper left block, labelled by
\(d_3=1\), has rows and columns labelled by \(S^{(2)}_{[1]}=\{1,2,3,4\}\).  It
is the Newton evaluation matrix of \(\Gamma_{[1]}\), the graph obtained by
deleting \(V_3\), which is the graph \(D=(2,4)\), \(p_1=[2,2]\) of
Figure~\ref{fig:D24}; indeed, with \(x^{(2)}_{T_k}\) replaced by \(y_k\), the
block is the matrix of Section~\ref{sec:evalNewton} for that graph.  It is
lower triangular, with determinant
\begin{equation}
\det\cN(\Gamma_{[1]})=\alpha^2\,(y_2-y_1)(y_4-y_3).
\end{equation}
The lower right block, labelled by \(d_3=2\), has rows and columns labelled by
\(S^{(2)}_{[2]}=\{1,3,4\}\), the nodes of \(V_2\) with at least two children.
It is the Newton evaluation matrix of \(\Gamma_{[2]}\): node \(1\in V_1\) keeps
the single child \(1\), and node \(2\in V_1\) keeps both children \(3<4\).
This is the graph \(D=(2,3)\), \(p_1=[1,2]\) of Figure~\ref{fig:L2}, and with
\(x^{(2)}_{T_1},x^{(2)}_{T_2},x^{(2)}_{T_3}\) replaced by \(y_1,y_3,y_4\) the
block is the matrix \eqref{eq:E23}.  It is not triangular, and the recursion is
applied to it once more: the Newton expansion \eqref{eq:rowopexample} reduces
it to the block triangular form \eqref{eq:E23after}, and
\begin{equation}
\det\cN(\Gamma_{[2]})=\alpha\,(y_4-y_3).
\end{equation}

\paragraph{Step 5: the determinant.}  Combining the steps,
\begin{equation}
\begin{aligned}
\det\cN
&=(z_2-z_1)(z_5-z_4)(z_7-z_6)\,\det\cN(\Gamma_{[1]})\,\det\cN(\Gamma_{[2]})\\
&=\alpha^3\,(y_2-y_1)(y_4-y_3)^2\,(z_2-z_1)(z_5-z_4)(z_7-z_6),
\end{aligned}
\end{equation}
which is the determinant given in Section~\ref{sec:youngexamples}.  The
exponent \(3\) of \(\alpha=x^{(1)}_2-x^{(1)}_1\) is assembled as \(2+1\) from the
two diagonal blocks.  In the language of Section~\ref{sec:young}, the points
of \(\cY(1)\cap\cY(2)\) are \((d_2,d_3)=(1,1),(2,1),(1,2)\): the two points
with \(d_3=1\) are accounted for by \(\Gamma_{[1]}\), and the point with
\(d_3=2\) by \(\Gamma_{[2]}\).  Similarly, the factor \(y_4-y_3\) appears once
in each block, matching \(|\cY(3)\cap\cY(4)|=2\) for the intervals
\(\cY(3)=\cY(4)=\{1,2\}\).  All the matrices above have been checked in exact
arithmetic.

\subsection{The example \texorpdfstring{\(D=(4,10)\)}{D=(4,10)}}
\label{app:D410}

Take \(D=(4,10)\), \(p_1=[2,2,3,3]\): the nodes \(1,2,3,4\in V_1\) have
children \(T_1<T_2\), \(T_3<T_4\), \(T_5<T_6<T_7\) and \(T_8<T_9<T_{10}\)
(Figure~\ref{fig:D410}).  This is a two-layer graph, \(L=2\), so the
recursion in the proof of Theorem~\ref{thm:newtondet} has a single step: removing layer
\(2\) leaves one-layer graphs, whose Newton evaluation matrices are the
ordinary Newton--Vandermonde matrices.  The new feature compared with
Appendix~\ref{app:D247} is that nodes \(3\) and \(4\) have three children, so
the Newton expansion among siblings involves quadratic basis polynomials, and
the matrix splits into three diagonal blocks instead of two.  We write
\begin{equation}
x_k=x^{(1)}_k,\qquad y_k=x^{(2)}_{T_k},\qquad
x_{m,k}=x_m-x_k,\qquad y_{m,k}=y_m-y_k ,
\end{equation}
so that, for instance, \(y_{10,8}=y_{10}-y_8\).  Since siblings have distinct
eigenvalues (Section~\ref{sec:algebra}), \(x_{m,k}\neq0\) for \(m\neq k\),
and \(y_{m,k}\neq0\) whenever \(T_m,T_k\) are siblings.

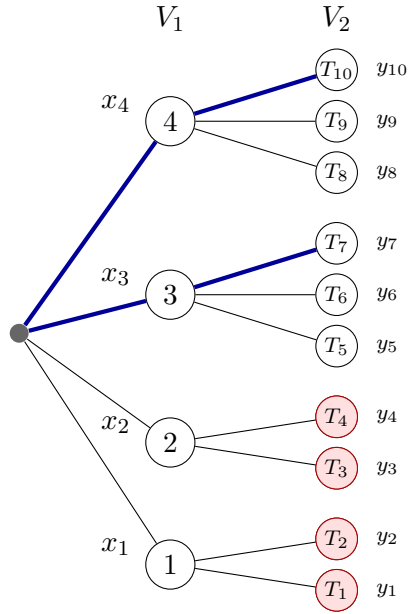
\begin{figure}[ht]
\centering
\begin{tikzpicture}[yscale=0.85]
\node[rt] (r) at (0,4.0) {};
\foreach \k/\y in {1/0.4,2/2.3,3/4.6,4/7.3} \node[gn] (a\k) at (2.0,\y) {$\k$};
\foreach \k/\y/\c in {1/0/T_1,2/0.8/T_2,3/1.9/T_3,4/2.7/T_4,5/3.8/T_5,6/4.6/T_6,7/5.4/T_7,8/6.5/T_8,9/7.3/T_9,10/8.1/T_{10}}
  \node[gn,minimum size=5.5mm,font=\scriptsize] (t\k) at (4.2,\y) {$\c$};
\foreach \k in {1,2,3,4} \node[gn,bad,minimum size=5.5mm,font=\scriptsize] at (t\k) {$T_\k$};
\draw (r)--(a1) (r)--(a2) (a1)--(t1) (a1)--(t2) (a2)--(t3) (a2)--(t4)
  (a3)--(t5) (a3)--(t6) (a4)--(t8) (a4)--(t9);
\draw[pth] (r)--(a3) (a3)--(t7) (r)--(a4) (a4)--(t10);
\foreach \k in {1,2,3,4} \node[left=1.5mm,font=\small] at (a\k.north west) {$x_\k$};
\foreach \k in {1,...,10} \node[right=1mm,font=\scriptsize] at (t\k.east) {$y_{\k}$};
\node[font=\small] at (2.0,8.9) {$V_1$};
\node[font=\small] at (4.2,8.9) {$V_2$};
\end{tikzpicture}
\caption{The graph with \(D=(4,10)\), \(p_1=[2,2,3,3]\), with the eigenvalues
\(x_k=x^{(1)}_k\) of \(C_1\) on the nodes of \(V_1\) and \(y_k=x^{(2)}_{T_k}\)
of \(C_2\) on the final nodes.  The paths of \(T_7\) and \(T_{10}\), the third
children of nodes \(3\) and \(4\), are drawn thick.  Only nodes \(3,4\) have
three children, so the graph Newton polynomials \(N_{T_7}\) and
\(N_{T_{10}}\) have no factors \(C_1-x_1\), \(C_1-x_2\), and, for generic eigenvalues, they do not
vanish at the shaded final nodes \(T_1,\ldots,T_4\).}
\label{fig:D410}
\end{figure}

\paragraph{Step 1: admissible sets, rank data, graph monomials and graph
Newton polynomials.}  For \(L=2\) the only sets \eqref{eq:SLm1} are
\(S^{(1)}_{[d_2]}\), the nodes of \(V_1\) with at least \(d_2\) children, and
the admissible set at layer \(1\) is \(R_1(d_2)=S^{(1)}_{[d_2]}\) by
\eqref{eq:R1}.  The numbers of children are \(2,2,3,3\), so
\begin{equation}
R_1(1)=R_1(2)=\{1,2,3,4\},\qquad R_1(3)=\{3,4\},\qquad R_1(d_2)=\emptyset\ \ (d_2\ge4).
\label{eq:R1D410}
\end{equation}
The rank data are read off as in Section~\ref{sec:admchld}: \(d_2\) is the
position of \(T\) among its siblings, and \(a+1\) is the position of
\(v_1(T)\) in \(R_1(d_2)\).  For \(T_1,\ldots,T_6,T_8,T_9\) the node
\(v_1(T)\) is the \(k\)-th element of \(R_1(d_2)=\{1,2,3,4\}\) when
\(v_1(T)=k\), so \(a=v_1(T)-1\).  For the third children \(T_7\) and
\(T_{10}\), \(d_2=3\) and \(R_1(3)=\{3,4\}\): node \(3\) is its first element
and node \(4\) its second, so \(a=0\) for \(T_7\) and \(a=1\) for \(T_{10}\).
The graph Newton polynomials \eqref{eq:graphNewton} have a factor
\(C_1-x_w\) for each \(w\in R_1(d_2)\) with \(w<v_1(T)\), and a factor
\(C_2-y_w\) for each earlier sibling \(w\) of \(T\):
\begin{equation}
{\renewcommand{\arraystretch}{1.3}\begin{array}{c|c|c|c|l}
T & v_1(T) & (a,d_2) & m_T & N_T\\\hline
T_1 & 1 & (0,1) & 1 & 1\\
T_2 & 1 & (0,2) & C_2 & C_2-y_1\\
T_3 & 2 & (1,1) & C_1 & C_1-x_1\\
T_4 & 2 & (1,2) & C_1C_2 & (C_1-x_1)(C_2-y_3)\\
T_5 & 3 & (2,1) & C_1^2 & (C_1-x_1)(C_1-x_2)\\
T_6 & 3 & (2,2) & C_1^2C_2 & (C_1-x_1)(C_1-x_2)(C_2-y_5)\\
T_7 & 3 & (0,3) & C_2^2 & (C_2-y_5)(C_2-y_6)\\
T_8 & 4 & (3,1) & C_1^3 & (C_1-x_1)(C_1-x_2)(C_1-x_3)\\
T_9 & 4 & (3,2) & C_1^3C_2 & (C_1-x_1)(C_1-x_2)(C_1-x_3)(C_2-y_8)\\
T_{10} & 4 & (1,3) & C_1C_2^2 & (C_1-x_3)(C_2-y_8)(C_2-y_9)
\end{array}}
\label{eq:tableD410}
\end{equation}
In each row the leading term of \(N_T\) is \(m_T\), as in
\eqref{eq:leadingterm}.  The ten graph monomials are those of
\eqref{eq:Mset}: since \(|S^{(1)}_{[1]}|=|S^{(1)}_{[2]}|=4\) and
\(|S^{(1)}_{[3]}|=2\),
\(\Monom(1)=\{1,C_1,C_1^2,C_1^3\}\), \(\Monom(2)=\{1,C_1,C_1^2,C_1^3\}\times C_2\)
and \(\Monom(3)=\{1,C_1\}\times C_2^2\).

\paragraph{Step 2: the Newton evaluation matrix.}  With rows \(T'\) and
columns \(N_T\) in path order, \(\cN_{T',T}=N_T(T')\) is obtained by
substituting \(C_1\to x_{v_1(T')}\), \(C_2\to y_{T'}\) in the last column of
\eqref{eq:tableD410}:
\begin{equation}
\resizebox{\textwidth}{!}{$\displaystyle
\setlength{\arraycolsep}{3pt}
\cN=
\begin{pmatrix}
1&0&0&0&0&0&y_{1,5}y_{1,6}&0&0&x_{1,3}y_{1,8}y_{1,9}\\
1&y_{2,1}&0&0&0&0&y_{2,5}y_{2,6}&0&0&x_{1,3}y_{2,8}y_{2,9}\\
1&y_{3,1}&x_{2,1}&0&0&0&y_{3,5}y_{3,6}&0&0&x_{2,3}y_{3,8}y_{3,9}\\
1&y_{4,1}&x_{2,1}&x_{2,1}y_{4,3}&0&0&y_{4,5}y_{4,6}&0&0&x_{2,3}y_{4,8}y_{4,9}\\
1&y_{5,1}&x_{3,1}&x_{3,1}y_{5,3}&x_{3,1}x_{3,2}&0&0&0&0&0\\
1&y_{6,1}&x_{3,1}&x_{3,1}y_{6,3}&x_{3,1}x_{3,2}&x_{3,1}x_{3,2}y_{6,5}&0&0&0&0\\
1&y_{7,1}&x_{3,1}&x_{3,1}y_{7,3}&x_{3,1}x_{3,2}&x_{3,1}x_{3,2}y_{7,5}&y_{7,5}y_{7,6}&0&0&0\\
1&y_{8,1}&x_{4,1}&x_{4,1}y_{8,3}&x_{4,1}x_{4,2}&x_{4,1}x_{4,2}y_{8,5}&y_{8,5}y_{8,6}&x_{4,1}x_{4,2}x_{4,3}&0&0\\
1&y_{9,1}&x_{4,1}&x_{4,1}y_{9,3}&x_{4,1}x_{4,2}&x_{4,1}x_{4,2}y_{9,5}&y_{9,5}y_{9,6}&x_{4,1}x_{4,2}x_{4,3}&x_{4,1}x_{4,2}x_{4,3}y_{9,8}&0\\
1&y_{10,1}&x_{4,1}&x_{4,1}y_{10,3}&x_{4,1}x_{4,2}&x_{4,1}x_{4,2}y_{10,5}&y_{10,5}y_{10,6}&x_{4,1}x_{4,2}x_{4,3}&x_{4,1}x_{4,2}x_{4,3}y_{10,8}&x_{4,3}y_{10,8}y_{10,9}
\end{pmatrix}
\begin{array}{l}
\leftarrow T_1\\ \leftarrow T_2\\ \leftarrow T_3\\ \leftarrow T_4\\
\leftarrow T_5\\ \leftarrow T_6\\ \leftarrow T_7\\ \leftarrow T_8\\
\leftarrow T_9\\ \leftarrow T_{10}
\end{array}$}
\label{eq:N410}
\end{equation}
The diagonal entries are the values \eqref{eq:diagonal}, e.g.\
\(\cN_{T_{10},T_{10}}=x_{4,3}\,y_{10,8}\,y_{10,9}\): one factor for the single
earlier element \(3\) of \(R_1(3)=\{3,4\}\), and one for each of the two
earlier siblings \(T_8,T_9\).

The nonzero entries above the diagonal are the eight entries of the columns of
\(N_{T_7}\) and \(N_{T_{10}}\) in the rows of \(T_1,\ldots,T_4\).  The
mechanism is that of the paragraph before Section~\ref{sec:general}.  If
\(T'\) precedes \(T\) and their paths first differ at layer \(1\), then
\(N_T(T')=0\) as soon as \(v_1(T')\in R_1(d_2(T))\).  For \(d_2(T)\le2\) the
admissible set is all of \(V_1\), so this always applies; for \(T_7\) and
\(T_{10}\), \(d_2=3\), and \(v_1(T')\in\{1,2\}\) is not admissible.  If the
paths first differ at layer \(2\), i.e.\ \(T'\) is an earlier sibling of \(T\),
then \(R_2=\Ch_2(v_1)\) by \eqref{eq:RL} and \(N_T(T')=0\) always; this is why,
for instance, the column of \(N_{T_7}\) vanishes in the rows of \(T_5,T_6\).
In addition, \(N_{T_{10}}\) contains \(C_1-x_3\), so its column vanishes at
all three children of node \(3\).

As in \eqref{eq:E23}, for generic eigenvalues every row of \eqref{eq:N410}
has at least two nonzero entries, so no reordering of the rows and columns of
this fixed matrix makes it triangular.  Relabelling the graph is a different
operation: it changes the admissible sets, and hence the graph Newton
polynomials themselves.  By Remark~\ref{rem:nontri}, if the nodes of \(V_1\)
were labelled with \(p_{1,1}\ge p_{1,2}\ge\cdots\), i.e.\ with the two nodes
with three children first, the resulting Newton evaluation matrix would be
lower triangular; with the increasing convention of \cite{KempRamgoolam2026}
it is not.

\paragraph{Step 3: the final-layer-excised graphs and the factorisation of
\(N_T\) (Lemma~\ref{lem:restrict} and Corollary~\ref{cor:factor}).}  Before
doing the Newton expansion, we identify the blocks it will produce.  By
Definition~\ref{def:final-layer-excised}, \(\Gamma_{[d]}\) is obtained by
deleting \(V_2\) and keeping the nodes of \(V_1\) with at least \(d\)
children, \(S^{(1)}_{[d]}\):
\begin{equation}
\Gamma_{[1]}=\Gamma_{[2]}:\ \text{one layer, nodes }1,2,3,4,
\qquad
\Gamma_{[3]}:\ \text{one layer, nodes }3,4 .
\end{equation}
In \(\Gamma_{[3]}\) the surviving nodes are relabelled \(1,2\) in the same
order, but they keep their eigenvalues \(x_3,x_4\).  Since \(L=2\),
Lemma~\ref{lem:restrict} has content only for \(i=L-1=1\), where
\eqref{eq:Rrestrict} says that the admissible set \(R_1(d)\) of \(\Gamma\) is
the whole final layer of \(\Gamma_{[d]}\), i.e.\ the set of children of the
root in \(\Gamma_{[d]}\).  This is \eqref{eq:R1D410}.  The graph Newton
polynomials of the one-layer graphs \(\Gamma_{[d]}\) are ordinary Newton
polynomials in \(C_1\):
\begin{equation}
\begin{aligned}
\Gamma_{[1]},\ \Gamma_{[2]}:&\quad N_1=1,\ \ N_2=C_1-x_1,\ \ N_3=(C_1-x_1)(C_1-x_2),\\
&\quad N_4=(C_1-x_1)(C_1-x_2)(C_1-x_3),\\
\Gamma_{[3]}:&\quad N_3=1,\ \ N_4=C_1-x_3,
\end{aligned}
\label{eq:NGammaD410}
\end{equation}
where we label the nodes by their names \(q\in V_1\) in \(\Gamma\).
Corollary~\ref{cor:factor} labels each final node \(T\) by the pair
\((q,d)=(v_1(T),d_2(T))\), and writes \(N_T\) as \(N_q(\Gamma_{[d]})\) times
the product over the earlier siblings of \(T\):
\begin{equation}
{\renewcommand{\arraystretch}{1.3}\begin{array}{c|c|l|l}
T & (q,d) & N_q(\Gamma_{[d]}) & \prod_{w<T,\,w\in\Ch_2(q)}(C_2-y_w)\\\hline
T_1,\,T_3,\,T_5,\,T_8 & (q,1) & N_q(\Gamma_{[1]}) & 1\\
T_2 & (1,2) & 1 & C_2-y_1\\
T_4 & (2,2) & C_1-x_1 & C_2-y_3\\
T_6 & (3,2) & (C_1-x_1)(C_1-x_2) & C_2-y_5\\
T_9 & (4,2) & (C_1-x_1)(C_1-x_2)(C_1-x_3) & C_2-y_8\\
T_7 & (3,3) & 1 & (C_2-y_5)(C_2-y_6)\\
T_{10} & (4,3) & C_1-x_3 & (C_2-y_8)(C_2-y_9)
\end{array}}
\label{eq:factorD410}
\end{equation}
which reproduces the last column of \eqref{eq:tableD410}.  In particular the
rank datum \(a\) of \(T\) in \(\Gamma\) is the rank datum of \(q\) in
\(\Gamma_{[d]}\): \(a=1\) for \(T_{10}\) because node \(4\) is the second node
of \(\Gamma_{[3]}\), although it is the fourth node of \(\Gamma\).

The factorisation is what makes the Newton expansion work.  On the three
children of node \(q=3\), say, the generator \(C_1\) takes the single value
\(x_3\), so every column \(T'\mapsto N_T(T')\), restricted to
\(T'\in\{T_5,T_6,T_7\}\), is a polynomial in the single variable \(\xi=y_{T'}\):
a polynomial of degree at most \(d'-1\), where \((q',d')\) is the label of
\(T\), whose coefficient of \(\xi^{d'-1}\) is the number
\(N_{q'}(\Gamma_{[d']})\) evaluated at \(C_1=x_3\).  For example, on the
children of node \(3\), the column of \(N_{T_{10}}\) is
\((x_3-x_3)(\xi-y_8)(\xi-y_9)=0\), the column of \(N_{T_9}\) is
\(x_{3,1}x_{3,2}\,x_{3,3}(\xi-y_8)=0\), and the column of \(N_{T_6}\) is
\(x_{3,1}x_{3,2}(\xi-y_5)\), of degree \(1\) with leading coefficient
\(N_3(\Gamma_{[2]})(x_3)=x_{3,1}x_{3,2}\).

\paragraph{Step 4: Newton expansion among siblings.}
Following the proof of Theorem~\ref{thm:newtondet}, for each \(q\in V_1\) we
expand the rows of its children in the Newton basis of their eigenvalues of
\(C_2\) (Lemma~\ref{lem:newton}).  For the nodes \(q=1,2\), with two children, this reads
\begin{equation}
\mathbf r_{T_1}=\Delta_{1,1},\quad
\mathbf r_{T_2}=\Delta_{1,1}+y_{2,1}\,\Delta_{1,2},\qquad
\mathbf r_{T_3}=\Delta_{2,1},\quad
\mathbf r_{T_4}=\Delta_{2,1}+y_{4,3}\,\Delta_{2,2},
\end{equation}
exactly as in Appendix~\ref{app:D247}, where \(\mathbf r_{T'}\) is the row of
\(T'\) in \eqref{eq:N410}.  For \(q=3\), with children \(T_5<T_6<T_7\), the
Newton basis is \(1,\ \xi-y_5,\ (\xi-y_5)(\xi-y_6)\), and
\begin{equation}
\mathbf r_{T_5}=\Delta_{3,1},\qquad
\mathbf r_{T_6}=\Delta_{3,1}+y_{6,5}\,\Delta_{3,2},\qquad
\mathbf r_{T_7}=\Delta_{3,1}+y_{7,5}\,\Delta_{3,2}+y_{7,5}y_{7,6}\,\Delta_{3,3};
\label{eq:Newton410}
\end{equation}
similarly for \(q=4\) with \(T_8<T_9<T_{10}\) and \(y_8,y_9,y_{10}\).  Thus
\(\cN=W\cN'\), as in \eqref{eq:NPfactor}, with
\(W=W_1\oplus W_2\oplus W_3\oplus W_4\),
\begin{equation}
W_1=\begin{pmatrix}1&0\\ 1&y_{2,1}\end{pmatrix},
\end{equation}
\(W_2\) the same with \(y_3,y_4\) in place of \(y_1,y_2\), and \(W_3\), \(W_4\)
given by \eqref{eq:Wn3} with \((\xi_1,\xi_2,\xi_3)=(y_5,y_6,y_7)\) and
\((y_8,y_9,y_{10})\) respectively.  By Lemma~\ref{lem:newton},
\begin{equation}
\det W=y_{2,1}\,y_{4,3}\,y_{6,5}y_{7,5}y_{7,6}\,y_{9,8}y_{10,8}y_{10,9}.
\label{eq:detWD410}
\end{equation}
These are precisely the eight differences of eigenvalues of siblings in the
final layer, each to the first power.  Solving \eqref{eq:Newton410} for the
new rows gives the divided differences \eqref{eq:divdiff123} of
Remark~\ref{rem:divdiff}, for instance
\(\Delta_{3,2}=(\mathbf r_{T_6}-\mathbf r_{T_5})/y_{6,5}\); this inverse
\(Q=W^{-1}\) is needed only for the inversion of
Section~\ref{sec:inverse-evaluation}.

\paragraph{Step 5: the block triangular matrix \(\cN'\).}  We now list
the rows \((q,d)\) of \(\cN'\), and the columns \(N_T\) with \(T\) labelled by
\((q',d')\), by increasing \(d\), and within each \(d\) by increasing \(q\):
\begin{equation}
\begin{aligned}
d=1:&\quad \text{rows }\Delta_{1,1},\Delta_{2,1},\Delta_{3,1},\Delta_{4,1},
\qquad \text{columns }N_{T_1},N_{T_3},N_{T_5},N_{T_8},\\
d=2:&\quad \text{rows }\Delta_{1,2},\Delta_{2,2},\Delta_{3,2},\Delta_{4,2},
\qquad \text{columns }N_{T_2},N_{T_4},N_{T_6},N_{T_9},\\
d=3:&\quad \text{rows }\Delta_{3,3},\Delta_{4,3},
\qquad\qquad\qquad\quad\ \ \text{columns }N_{T_7},N_{T_{10}} .
\end{aligned}
\end{equation}
The same permutation is applied to rows and columns, which does not change
the determinant.  The result is
\begin{equation}
\resizebox{\textwidth}{!}{$\displaystyle\setlength{\arraycolsep}{2.5pt}
\cN'=
\left(\begin{array}{cccc|cccc|cc}
1&0&0&0&0&0&0&0&y_{1,5}y_{1,6}&x_{1,3}y_{1,8}y_{1,9}\\
1&x_{2,1}&0&0&y_{3,1}&0&0&0&y_{3,5}y_{3,6}&x_{2,3}y_{3,8}y_{3,9}\\
1&x_{3,1}&x_{3,1}x_{3,2}&0&y_{5,1}&x_{3,1}y_{5,3}&0&0&0&0\\
1&x_{4,1}&x_{4,1}x_{4,2}&x_{4,1}x_{4,2}x_{4,3}&y_{8,1}&x_{4,1}y_{8,3}&x_{4,1}x_{4,2}y_{8,5}&0&y_{8,5}y_{8,6}&0\\
\hline
0&0&0&0&1&0&0&0&y_1+y_2-y_5-y_6&x_{1,3}(y_1+y_2-y_8-y_9)\\
0&0&0&0&1&x_{2,1}&0&0&y_3+y_4-y_5-y_6&x_{2,3}(y_3+y_4-y_8-y_9)\\
0&0&0&0&1&x_{3,1}&x_{3,1}x_{3,2}&0&0&0\\
0&0&0&0&1&x_{4,1}&x_{4,1}x_{4,2}&x_{4,1}x_{4,2}x_{4,3}&y_8+y_9-y_5-y_6&0\\
\hline
0&0&0&0&0&0&0&0&1&0\\
0&0&0&0&0&0&0&0&1&x_{4,3}
\end{array}\right)
\begin{array}{l}
\leftarrow \Delta_{1,1}\\ \leftarrow \Delta_{2,1}\\ \leftarrow \Delta_{3,1}\\ \leftarrow \Delta_{4,1}\\
\leftarrow \Delta_{1,2}\\ \leftarrow \Delta_{2,2}\\ \leftarrow \Delta_{3,2}\\ \leftarrow \Delta_{4,2}\\
\leftarrow \Delta_{3,3}\\ \leftarrow \Delta_{4,3}
\end{array}$}
\label{eq:N410after}
\end{equation}
with columns \(N_{T_1},N_{T_3},N_{T_5},N_{T_8}\,|\,N_{T_2},N_{T_4},N_{T_6},N_{T_9}\,|\,N_{T_7},N_{T_{10}}\).
Every entry follows from Step 3 and the degree property, Lemma~\ref{lem:degree},
without solving \eqref{eq:Newton410} for the new rows; this is the computation in the proof of
Theorem~\ref{thm:newtondet}, which we now follow block by block.  Write
\(\cN'_{(q,d),T}=\Delta_{q,d}N_T\), with \(T\) labelled by \((q',d')\).  On the
children of \(q\), the column of \(N_T\) is a polynomial in \(\xi\) of degree
at most \(d'-1\) with coefficient \(N_{q'}(\Gamma_{[d']})(q)\) of
\(\xi^{d'-1}\) (Step 3).

\emph{Blocks below the diagonal, \(d>d'\): zero.}  Here the degree is at most
\(d'-1<d-1\), and Lemma~\ref{lem:degree} gives \(\Delta_{q,d}N_T=0\).  For \(d=2\),
\(d'=1\), this is the statement that the columns of
\(N_{T_1},N_{T_3},N_{T_5},N_{T_8}\) do not involve \(C_2\), so they take equal
values on siblings, as in Appendix~\ref{app:D247}.  The new case is \(d=3\),
\(d'\in\{1,2\}\): the columns of \(N_{T_2},N_{T_4},N_{T_6},N_{T_9}\) are linear
in \(C_2\), and the Newton expansion of a linear function stops at the basis
polynomial \(\xi-\xi_1\).  For example, on the children of node \(3\) the
column of \(N_{T_4}\) is
\(x_{3,1}(\xi-y_3)=x_{3,1}y_{5,3}+x_{3,1}(\xi-y_5)\), so that
\(\Delta_{3,1}N_{T_4}=x_{3,1}y_{5,3}\), \(\Delta_{3,2}N_{T_4}=x_{3,1}\) and
\(\Delta_{3,3}N_{T_4}=0\).

\emph{Diagonal blocks, \(d=d'\): the Newton evaluation matrices of
\(\Gamma_{[d]}\).}  Here Lemma~\ref{lem:degree} returns the leading coefficient,
\(\Delta_{q,d}N_T=N_{q'}(\Gamma_{[d]})(q)=\cN(\Gamma_{[d]})_{q,q'}\).  For
\(d=1\) this is just the value of \(N_{q'}(\Gamma_{[1]})\) at \(C_1=x_q\), and
the block is the \(4\times4\) Newton--Vandermonde matrix of
\(x_1,x_2,x_3,x_4\).  For \(d=2\), the columns of \(N_{T_2},N_{T_4},N_{T_6},
N_{T_9}\) are linear in \(C_2\) with coefficients \(1,\ C_1-x_1,\
(C_1-x_1)(C_1-x_2),\ (C_1-x_1)(C_1-x_2)(C_1-x_3)\), by \eqref{eq:factorD410},
and evaluating these at \(C_1=x_q\) gives the same \(4\times4\) matrix again,
since \(\Gamma_{[2]}=\Gamma_{[1]}\).  For \(d=3\) the columns of \(N_{T_7}\) and
\(N_{T_{10}}\) are quadratic in \(C_2\) with coefficients \(1\) and
\(C_1-x_3\), the graph Newton polynomials of \(\Gamma_{[3]}\) in
\eqref{eq:NGammaD410}; evaluated at \(C_1=x_3,x_4\) they give the \(2\times2\)
Newton--Vandermonde matrix of \(x_3,x_4\).  Explicitly, on the children of
node \(4\) the column of \(N_{T_{10}}\) is
\(x_{4,3}(\xi-y_8)(\xi-y_9)=x_{4,3}\,\pi_{4,3}(\xi)\), which is already in Newton
form, so \(\Delta_{4,3}N_{T_{10}}=x_{4,3}\).

\emph{Blocks above the diagonal, \(d<d'\).}  These entries do not enter the
determinant, but they are needed for the inversion of
Section~\ref{sec:inverse-evaluation}.  For \(d=1\), \(\Delta_{q,1}N_T\) is the
value of \(N_T\) at the first child of \(q\), so the first four rows of
\eqref{eq:N410after} are the rows of \(T_1,T_3,T_5,T_8\) in \eqref{eq:N410},
with the columns permuted.  For \(d=2\), \(d'=3\), the columns are those of
\(N_{T_7}\) and \(N_{T_{10}}\).  By \eqref{eq:factorD410}, on the children of
\(q\) each of them is a quadratic polynomial in \(\xi\),
\begin{equation}
h(\xi)=c\,(\xi-\rho_1)(\xi-\rho_2),\qquad
\begin{cases}
c=1,\ \ \rho_1=y_5,\ \rho_2=y_6 & \text{for }N_{T_7},\\
c=x_q-x_3,\ \ \rho_1=y_8,\ \rho_2=y_9 & \text{for }N_{T_{10}},
\end{cases}
\end{equation}
where \(x_q\) is the eigenvalue of \(C_1\) on \(q\).  Let \(\xi_1,\xi_2\) be the
eigenvalues of \(C_2\) on the first two children of \(q\).  Rewriting \(h\) in
the Newton basis \(1,\ \xi-\xi_1,\ (\xi-\xi_1)(\xi-\xi_2)\),
\begin{equation}
h(\xi)=h(\xi_1)+c\,(\xi_1+\xi_2-\rho_1-\rho_2)\,(\xi-\xi_1)
+c\,(\xi-\xi_1)(\xi-\xi_2),
\label{eq:dd2quadratic}
\end{equation}
so that \(\Delta_{q,2}N_T=c\,(\xi_1+\xi_2-\rho_1-\rho_2)\).  (For \(q=1,2\),
with two children, the last term vanishes on the children; for \(q=3,4\) its
coefficient is \(\Delta_{q,3}N_T=c\), as in the diagonal block above.)  With
\((\xi_1,\xi_2)=(y_1,y_2),(y_3,y_4),(y_5,y_6),(y_8,y_9)\) for \(q=1,2,3,4\),
this gives the last two columns of the middle rows of \eqref{eq:N410after}.
This is step 2 of the inversion procedure of
Section~\ref{sec:inverse-evaluation}.  Some of these entries vanish, e.g.\
\(\Delta_{3,2}N_{T_7}=0\) because \((\xi_1,\xi_2)=(\rho_1,\rho_2)=(y_5,y_6)\),
and \(\Delta_{3,2}N_{T_{10}}=\Delta_{4,2}N_{T_{10}}=0\) because \(c=0\) for
\(q=3\) and \(\xi_1+\xi_2=\rho_1+\rho_2\) for \(q=4\), but nothing in the
argument uses this.

The block sizes are \(|S^{(1)}_{[d]}|=4,4,2\).  The numbers of children of the
nodes of \(V_1\), arranged as a partition of \(D_2=10\), are \((3,3,2,2)\),
whose conjugate is \((4,4,2)\), as stated in the proof of
Theorem~\ref{thm:newtondet}.

\paragraph{Step 6: the determinant.}  Since \(\cN=W\cN'\), and the block
triangular form \eqref{eq:N410after} is reached by a simultaneous reordering
of rows and columns, \eqref{eq:detWD410} gives
\begin{equation}
\begin{aligned}
\det\cN
&=\det W\,\det\cN(\Gamma_{[1]})\,\det\cN(\Gamma_{[2]})\,\det\cN(\Gamma_{[3]})\\
&=y_{2,1}\,y_{4,3}\,(y_{6,5}y_{7,5}y_{7,6})\,(y_{9,8}y_{10,8}y_{10,9})
\times\Bigl(\prod_{1\le k<m\le4}x_{m,k}\Bigr)\times\Bigl(\prod_{1\le k<m\le4}x_{m,k}\Bigr)
\times x_{4,3}.
\end{aligned}
\label{eq:detD410}
\end{equation}
Here the recursion stops after one step, because the graphs \(\Gamma_{[d]}\)
have one layer and their Newton evaluation matrices are already triangular.
This agrees with the product of the diagonal entries of \eqref{eq:N410},
\begin{equation}
1\cdot y_{2,1}\cdot x_{2,1}\cdot x_{2,1}y_{4,3}\cdot x_{3,1}x_{3,2}\cdot
x_{3,1}x_{3,2}y_{6,5}\cdot y_{7,5}y_{7,6}\cdot x_{4,1}x_{4,2}x_{4,3}\cdot
x_{4,1}x_{4,2}x_{4,3}y_{9,8}\cdot x_{4,3}y_{10,8}y_{10,9},
\end{equation}
as asserted by \eqref{eq:diagonalproduct}.  The two ways of grouping the
factors are those of the end of the proof of Theorem~\ref{thm:newtondet}: the
\(y\)-factors of the diagonal entries are collected by \(\det W\), and the
\(x\)-factors of the diagonal entry of \(T\) labelled \((q,d)\) form the
diagonal entry \(\cN(\Gamma_{[d]})_{q,q}\).

The exponent of \(x_{m,k}\) is \(2\), except for \(x_{4,3}\), whose exponent
is \(3\).  In the language of Section~\ref{sec:young}, the Young diagrams of
the nodes of \(V_1\) are the intervals
\(\cY(1)=\cY(2)=\{1,2\}\), \(\cY(3)=\cY(4)=\{1,2,3\}\), and the exponent
\(|\cY(p)\cap\cY(q)|=\min(n_p,n_q)\) of \eqref{eq:minchildren} counts the
diagonal blocks containing the factor: each point \(d_2=d\) of
\(\cY(p)\cap\cY(q)\) is the statement that both \(p\) and \(q\) survive in
\(\Gamma_{[d]}\), and so contribute \(x_{q,p}\) to its Newton--Vandermonde
determinant.  For \((p,q)=(3,4)\) the three points \(d_2=1,2,3\) correspond to
the three blocks; for the other pairs only the blocks \(d_2=1,2\) contribute.
Since the monomial and Newton evaluation matrices have the same determinant
\eqref{eq:changebasisdet},
\begin{equation}
\det\cM=y_{2,1}\,y_{4,3}\,y_{6,5}y_{7,5}y_{7,6}\,y_{9,8}y_{10,8}y_{10,9}
\Bigl(\prod_{1\le k<m\le4}x_{m,k}\Bigr)^{2}\,x_{4,3},
\end{equation}
in agreement with \cite{KempRamgoolam2026} and with Theorem~\ref{thm:KR}.
The matrices \eqref{eq:N410} and \eqref{eq:N410after}, the identity
\(\cN=W\cN'\), and the determinants have been checked in exact arithmetic.

\end{document}